%% file: _main.tex
\documentclass[twoside]{article}
\input{common}

\usepackage[accepted]{aistats2023}
\usepackage{natbib}

\begin{document}

\twocolumn[

\runningauthor{El-Mhamdi, Farhadkhani, Guerraoui, Hoang}

\aistatstitle{On the Strategyproofness of the Geometric Median}

\aistatsauthor{ El-Mahdi El-Mhamdi \And Sadegh Farhadkhani$^*$}

\aistatsaddress{ Calicarpa,  École Polytechnique \And EPFL} 

\aistatsauthor{ Rachid Guerraoui \And Lê-Nguyên Hoang$^*$ }

\aistatsaddress{ EPFL \And Calicarpa, Tournesol } 
]

\begin{NoHyper}

\renewcommand{\thefootnote}{\fnsymbol{footnote}}
\footnotetext{\hspace{-5mm}Authors are listed in alphabetical order.
\\$^*$Correspondence to:
\texttt{sadegh.farhadkhani@epfl.ch}, \\ and 
\texttt{
len@tournesol.app}.
} 
\end{NoHyper}

\renewcommand{\thefootnote}{\arabic{footnote}}

\input{abstract}

\input{introduction2}
\input{model2}
\input{strategyproofness}

\input{asymptotic}
\input{skewed_preferences}

\input{skewed_geometric_median}
\input{experiments}

\input{related_work2}
\input{conclusion2}
\section*{Acknowledgements}
We thank Rafael Pinot for the very helpful comments on the introduction of the paper. We thank the anonymous reviewers for their constructive comments. This work has been supported in part by the Swiss National
Science Foundation (SNSF) project $200021{\_}200477$.
\bibliography{references}

\onecolumn
\appendix

\input{Appendix}

\vfill

\end{document}

%% file: common.tex
\usepackage{xcolor}

\usepackage{amsmath,amssymb}
\usepackage{bm,bbm,xfrac,url}
\usepackage[linktocpage=true]{hyperref}
\usepackage{authblk}
\usepackage{dsfont}
\usepackage{appendix}
\usepackage{cancel}
\usepackage{environ}

\usepackage{wrapfig}
\usepackage{todonotes}

\usepackage{amsthm}

\makeatletter
\newtheorem*{rep@theorem}{\rep@title}
\newcommand{\newreptheorem}[2]{%
\newenvironment{rep#1}[1]{%
 \def\rep@title{#2 \ref{##1}}%
 \begin{rep@theorem}}%
 {\end{rep@theorem}}}
\makeatother
\newreptheorem{theorem}{Theorem}

\makeatletter
\newtheorem*{rep@assumption}{\rep@title}
\newcommand{\newrepassumption}[2]{%
\newenvironment{rep#1}[1]{%
 \def\rep@title{#2 \ref{##1}}%
 \begin{rep@assumption}}%
 {\end{rep@assumption}}}
\makeatother
\newreptheorem{assumption}{Assumption}

\newcommand{\mathsep}{,~}
\newcommand{\st}{\,\middle|\,}
\newcommand{\set}[1]{\left\lbrace #1 \right\rbrace}
\newcommand{\card}[1]{\left\lvert{#1}\right\rvert}
\newcommand{\absv}[1]{\card{#1}}
\newcommand{\norm}[2]{\left\lVert{#1}\right\rVert_{#2}}

\newcommand{\trace}[1]{\textsc{Tr}\left(#1\right)}

\newcommand{\setR}{\mathbb R}
\newcommand{\setN}{\mathbb N}

\newcommand{\probability}[1]{\mathbb P\left[#1\right]}
\newcommand{\probabilitysub}[2]{\mathbb P_{#1}\left[{#2}\right]}
\newcommand{\expect}[1]{\mathbb E #1}

\newcommand{\expectVariable}[2]{\mathbb E_{#1}\left[#2\right]}

\DeclareMathOperator*{\argmin}{arg\,min}

\newtheorem{definition}{Definition}
\newtheorem{proposition}{Proposition}
\newtheorem{theorem}{Theorem}
\newtheorem{lemma}{Lemma}
\newtheorem{corollary}{Corollary}

\newtheorem{assumption}{Assumption}
\newtheorem{remark}{Remark}

\newcommand{\param}{\theta}
\newcommand{\paramsub}[1]{\param_{#1}}
\newcommand{\paramfamily}{\vec \theta}
\newcommand{\paramdistribution}{\tilde \theta}
\newcommand{\PARAM}{\Theta}

\newcommand{\strategicvote}[1]{s}
\newcommand{\targetvector}{t}

\newcommand{\vote}{\textsc{Vote}}
\newcommand{\average}{\textsc{Avg}}

\newcommand{\ball}{\mathcal B}

\newcommand{\skewness}{\textsc{Skew}}
\newcommand{\spectrum}{\textsc{Sp}}
\newcommand{\symmetric}[1]{\textsc{Sym}_{#1}}

\newcommand{\hessian}{H}

\newcommand{\sdpmatrix}{S}
\newcommand{\adaptivematrix}{\Sigma}
\newcommand{\convexset}{\mathcal C}
\newcommand{\subgradient}{h}

\newcommand{\parx}{x}
\newcommand{\parxsub}[1]{\parx_{#1}}
\newcommand{\parxfamily}{\vec{\parx}}
\newcommand{\pary}{y}
\newcommand{\parysub}[1]{\pary_{#1}}

\newcommand{\parz}{z}

\newcommand{\parw}{w}
\newcommand{\parwsub}[1]{\parw_{#1}}

\newcommand{\voter}{v}
\newcommand{\voterbis}{w}

\newcommand{\VOTER}{V}
\newcommand{\truthful}{T}
\newcommand{\strategists}{S}
\newcommand{\eigenvalue}{\lambda}
\newcommand{\randomvar}{X}

\newcommand{\unitvector}[1]{{\bf u}_{#1}}
\newcommand{\GeometricMedian}{\textsc{Gm}}
\newcommand{\CWMedian}{\textsc{Cw}}

\newcommand{\geometricmedian}{g}
\newcommand{\strategyproofbound}{\alpha}

\newcommand{\Loss}{\mathcal L}

\newcommand{\event}{\mathcal E}

\newcommand{\error}{e}
\newcommand{\ra}{r_1}
\newcommand{\rb}{r_3}
\newcommand{\rc}{r_2}
\newcommand{\setAchieve}{\mathcal A_{\VOTER}}
\newcommand{\Var}{\mathrm{Var}}
\newcommand{\helement}{a}
\newcommand{\oneDelement}[2]{{#1}[{#2}]}
\newcommand{\twoDelement}[3]{{#1}[{#2},{#3}]}
\newcommand{\threeDelement}[4]{{#1}[{#2},{#3},{#4}]}

%% file: abstract.tex
\begin{abstract}
  The \emph{geometric median}, an instrumental component of the secure machine learning toolbox, is known to be effective when robustly aggregating models (or gradients), gathered from potentially malicious (or strategic) users. What is less known is the extent to which the geometric median incentivizes dishonest behaviors. This paper addresses this fundamental question by 
  quantifying its \emph{strategyproofness}.
  While we observe that the geometric median is not even approximately strategyproof, we prove that it is \emph{asymptotically $\strategyproofbound$-strategyproof}: %
  when the number of users is large enough, 
  a user that misbehaves can gain at most a multiplicative factor $\strategyproofbound$, which 
  we compute as a function of the distribution followed by the users.
  We then generalize our results to the case where users actually care more about specific dimensions, determining how this impacts   $\strategyproofbound$. 
 We also show how the \emph{skewed geometric medians} can be used to improve strategyproofness.
\end{abstract}

%% file: introduction2.tex
\color{black}
\section{INTRODUCTION}
There has recently been a growing interest in collaborative
machine learning to efficiently utilize the ever-increasing
amount of data and computational resources~\citep{McMahan17,Kairouz21,Tensorflow2015}. 
Collaborative learning gathers information from multiple users (e.g., gradient vectors~\citep{Zinkevich10}, local model parameters~\citep{DinhTN20,equivalence} or users' preferences~\citep{NoothigattuGADR18,AllouahGHV22}) and typically summarizes it in a single vector. 
While averaging is the most widely used method for aggregating multiple vectors into a single vector~\citep{polyak92}, it suffers from severe security flaws: 
 averaging can be arbitrarily manipulated by a single strategic user
~\citep{BlanchardMGS17}.  

The geometric median is a promising ``robust'' alternative to averaging. 
It has been widely used in collaborative learning as it is a provably good approximation of the average~\citep{STANISLAV15} and it is robust to a minority of malicious users~\citep{lopuhaa1989}.
A large body of research known as ``Byzantine learning''~\citep{BlanchardMGS17,chen17,pmlr-v80-mhamdi18a,rajput19,alistarh18} uses the geometric median to ensure safe learning despite the presence of participants with arbitrarily malicious behavior~\citep{pmlr-v162-farhadkhani22a, karimireddy2022byzantinerobust,acharya2021,Wu20,So21,Gu21,pillutla2019robust, equivalence}. 
Interestingly, the geometric median also satisfies the fairness principle  \emph{``one voter, one vote {with a unit force}''} (see Section \ref{sec:gm_def}), making it ethically appealing. 

In this paper, we study the extent to which the geometric median \emph{incentivizes} strategic manipulations\footnote{Hence, we often use the term ``voter'' instead of ``user''.}.
Ideally, we would like the geometric median to be \emph{strategyproof}~\citep{gibbard1973manipulation,satterthwaite1975strategy,brandt16}, i.e., 
we want it to be in each voter's best interest to report their true preferred vector. 
Put differently, honesty would ideally be a \emph{dominant strategy}~\citep{chung2007foundations}.
This is very different from \emph{Byzantine learning}, 
which only focuses on the resilience of the training, 
usually assuming a \emph{majority} of honest users. 
Conversely, we consider the more realistic case where \emph{every} user wants to bias the algorithm towards their specific target state.
Such considerations are critical 
for high-stake life-endangering applications such as content moderation and recommendation~\citep{yue2019weaponization,whitten2020poison}, 
in which different people have diverging preferences over 
what should be removed~\citep{ribeiro20,bhat20}, 
accompanied with a warning message~\citep{mena20},
and be promoted at scale~\citep{michelman20}. 
Clearly, activists, companies and politicians all want to bias algorithms to promote certain views, products or ideologies~\citep{hoang20a}.
These entities should thus be expected to behave untruthfully, if they can easily game the algorithms with fabricated behaviors. 

Now, assuming that each user wants to minimize the distance between the computed geometric median and their target vector, 
it is actually known that the geometric median fails to be strategyproof \citep{KIM198429}
(see Figure \ref{fig:geometricmedian}).
However, raw strategyproofness is a binary worst-case analysis. 
In practice, optimizing strategic reporting may be costly (e.g., information gathering and computational costs, and the risk of being exposed), and hence may not be profitable if the potential gain is small.
This prompts us to \emph{quantify} the strategyproofness of the geometric median:
how much can a strategic voter gain by misreporting their preferred vector~\citep{lubin12,wang15,Han15}? 

\paragraph{\textbf{Contributions.}}
 Our first contribution is to show that the geometric median fails to guarantee approximate strategyproofness. 
More precisely, for any $\strategyproofbound$, we show that there exists a configuration where a strategic voter can gain a factor $\strategyproofbound$ by behaving strategically rather than truthfully. %

Our main contribution is to then study the more specific case where voters' reported vectors come independently from an underlying distribution. We prove that, in the limit where the number of voters is large enough, and with high probability, the geometric median is indeed $\strategyproofbound$-strategyproof. This goes through 
introducing and formalizing the notion of \emph{asymptotic strategyproofness} with respect to the distribution of reported vectors.
We  show how to compute the bound $\strategyproofbound$ as a function of this distribution. 

Our two first contributions
apply to the case where a voter wants to minimize the \emph{Euclidean} distance between the geometric median and their target vector. 
Essentially, this amounts to saying that the voters' preferences are isotropic, i.e., all dimensions have the same importance for the voters. However, in practical applications, a voter may care a lot more about certain dimensions than others, 
Our third contribution is a generalization to this setting,  proving that, in a rigorous sense, the geometric median becomes {\it less} strategyproof if some dimensions are both more polarized and more important than others. 

As a fourth important contribution, we show how strategyproofness can be improved by introducing and analyzing the \emph{skewed geometric median}. 
Intuitively, this corresponds to skewing the feature space using a linear transformation $\adaptivematrix$, computing the geometric median in the skewed space, and de-skewing the computed geometric median by applying $\adaptivematrix^{-1}$. In essence, the skewed geometric median can be used to weaken pulls along polarized dimensions, and strengthen pulls along others. 
This helps limit the incentives to exaggerate preferences along more polarized dimensions, by intuitively giving voters more voting power along orthogonal dimensions ``at the same cost''.

{
\paragraph{\textbf{Background.}} 
Classically called the Fermat-Weber solution~\citep{brimberg17}, the geometric median solves a version of the widely studied (optimal) facility location problem~\citep{hansen85,walsh20,pinyan09,AAAIW1510182,pingzhong20,esco11,Sui15,kyr10,fotakis13}, as it minimizes the sum of distances of the agents to the chosen location. 
In one dimension, the geometric median coincides with the median, which was shown~\citep{Moulin80} to be (group) strategyproof. 
But in higher dimensions, the geometric median is known to be \emph{not} strategyproof~\citep{KIM198429}. 
To the best of our knowledge, however, our paper is the first to analyze the geometric median in high dimension, with weakened forms of strategyproofness like (asymptotic) $\strategyproofbound$-strategyproofness. 
As far as we know, we are also the first to investigate skewed geometric medians and skewed preferences.}

  \paragraph{\textbf{Roadmap.}}
The rest of the paper is organized as follows. 
Section~\ref{sec:model} formally defines different notions of strategyproofness and the geometric median aggregation rule. 
Section~\ref{sec:non_strategyproofness} proves that this rule is not $\strategyproofbound$-strategyproof, whilst Section~\ref{sec:asymptotic_strategyproofness} proves that it is asymptotically $\strategyproofbound$-strategyproof. 
In Section~\ref{sec:skewed_preferences}, we generalize our result to non-isotropic voters' preferences and to the skewed geometric median. We provide a simple experiment in Section~\ref{sec:exp}.
{ Section~\ref{sec:related_work} discusses related work,}
and Section~\ref{sec:conclusion} concludes. 
{ Due to space limitations, most of the proofs and some auxiliary results are provided in the appendices.}

%% file: model2.tex
\section{MODEL}
\label{sec:model}

We consider $1+\VOTER$ voters. 
Each voter $\voter \in [\VOTER] \triangleq \{ 1, \ldots, \VOTER \}$ reports a { (potentially fabricated)} vector $\paramsub{\voter} \in \setR^d$. We denote by $\paramfamily{} \triangleq (\paramsub{1}, \ldots, \paramsub{\VOTER})$ the family of other voters' reported vectors. 
We then, without loss of generality\footnote{Because all the votes that we consider are permutation invariant (Proposition \ref{prop:anonymity} in Appendix~\ref{sec:GM}).}, analyze the incentives of voter $0$. 
We assume that voter $0$ has a preferred \emph{target} vector $\targetvector \in \setR^d$,
but they report a potentially different, \emph{strategically} crafted, vector $\strategicvote{0} \in \setR^d$. 
A voting algorithm $\vote$ then aggregates all voters' vectors into a common decision vector $\vote(\strategicvote{0}, \paramfamily{}) \in \setR^d$,
which voter $0$ would prefer to be close to their target vector $\targetvector$.

\subsection{The Many Faces of Strategyproofness}
\label{sec:strategyproofness}

We define the strategic gain as the best multiplicative gain that voter $0$ can obtain by misreporting their preference, i.e. by reporting $\strategicvote{}$ instead of $\targetvector$.
Strategyproofness bounds the maximal strategic gain.

\begin{definition}[$\strategyproofbound$-strategyproofness]
\label{def:strategyproofness}
  $\vote$ is $\strategyproofbound$-strategyproof if, for any others' vectors $\paramfamily{} \in \setR^{d \times \VOTER}$, any target vector $\targetvector \in \setR^d$ and any strategic vote $\strategicvote{0} \in \setR^d$, 
  the strategic gain is at most $1+\strategyproofbound$, i.e.
  \begin{equation*}
    \forall \paramfamily{}, \targetvector, \strategicvote{} \mathsep 
    \norm{\vote(\targetvector, \paramfamily{}) - \targetvector}{2} \leq (1+\strategyproofbound) \norm{\vote(\strategicvote{0}, \paramfamily{}) - \targetvector}{2}.
  \end{equation*}
  Smaller values of $\strategyproofbound$ yield stronger guarantees. If $\strategyproofbound =0$, then we simply say that $\vote$ is strategyproof.
\end{definition}

The opposite of strategyproofness is an arbitrarily manipulable vote, which we define as follows.

\begin{definition}[Arbitrarily manipulable]
\label{def:arbitrarily-manipulable}
  $\vote$ is arbitrarily manipulable by a single voter if, for any others' vectors $\paramfamily{} \in \setR^{  d \times  \VOTER}$ and any target vector $\targetvector \in \setR^d$, there exists $\strategicvote{0} \in \setR^d$ such that $\vote (\strategicvote{0}, \paramfamily{}) = \targetvector$.
\end{definition}

\textcolor{black}{It is possible for a vector aggregation rule to be neither} $\strategyproofbound$-strategyproof nor arbitrarily manipulable. In fact, we show that this is the case for the geometric median. This remark calls for more subtle definitions of strategyproofness.
In particular, it may be unreasonable to demand $\strategyproofbound$-strategyproofness for \emph{all} other voters' inputs $\paramfamily{} \in \setR^{d \times \VOTER}$ (this is known as \emph{dominant strategy incentive compatibility}).
In practice, other voters are usually expected to report some vectors more often than others.
This motivates us to consider an alternative high-probability definition of $\strategyproofbound$-strategyproofness\footnote{Our definition does not coincide with \emph{Bayesian incentive compatibility}, which aims to bound one's \emph{expected} strategic gain.} taking into account the distribution of vectors.
We thus introduce and study \emph{asymptotic $\strategyproofbound$-strategyproofness}.
To define this notion, we first assume that other voters' vectors are drawn\footnote{This setting is similar to ``Worst-case IID susceptibility'' proposed by~\cite{lubin12}. But, we consider high
probability bounds on the gain which is different from the expected regret defined  by~\cite{lubin12}.} independently from some distribution $\paramdistribution$ over $\setR^d$. Asymptotic strategyproofness then corresponds to strategyproofness in the limit where $\VOTER$ is large enough. 

{
\begin{definition}[Asymptotic $\strategyproofbound$-strategyproofness]
\label{def:asymptotic-strategyproofness}
  $\vote$ is asymptotically $\strategyproofbound$-strategyproof if, 
  for any $\varepsilon, \delta >0$, 
  there exists $\VOTER_0 \geq 1$ such that,
  as long as there are $\VOTER \geq \VOTER_0$ other voters whose reported vectors are drawn independently from distribution $\paramdistribution$, 
  then with probability at least $1-\delta$, 
  for any target vector $\targetvector \in \setR^d$,
  and any strategic vote $\strategicvote{0} \in \setR^d$, 
  the strategic gain is bounded by $1+\strategyproofbound+\varepsilon$, i.e., 
  \begin{equation*}
      \probabilitysub{\paramfamily \sim (\paramdistribution)^\VOTER}{
        \forall \targetvector, \strategicvote{0} :
        \event(\strategyproofbound+\varepsilon, \targetvector, \strategicvote{})
        } \geq 1-\delta,
  \end{equation*}
  where $\event(\strategyproofbound+\varepsilon, \targetvector, \strategicvote{})$ is the event
  \begin{equation*}
       \set{ \norm{\vote(\targetvector, \paramfamily{}) - \targetvector}{} \leq (1+\strategyproofbound+\varepsilon) \norm{\vote(\strategicvote{0}, \paramfamily{}) - \targetvector}{}}.
  \end{equation*}
  If $\strategyproofbound = 0$, we say that $\vote$ is asymptotically strategyproof.
\end{definition}
}
Note that this definition implicitly depends on the distribution $\paramdistribution$ of voters' inputs. In fact, we prove that the geometric median is asymptotically $\strategyproofbound$-strategyproof, for a value of $\strategyproofbound$ that we derive from the distribution $\paramdistribution$. 

Finally, we also study the more general case of non-isotropic preferences. To model this, we replace the Euclidean norm by the $\sdpmatrix$-Mahalanobis norm, for some positive definite matrix $\sdpmatrix \succ 0$, which is given by $\norm{\parx}{\sdpmatrix} \triangleq \norm{\sdpmatrix \parx}{2}$.
Intuitively, the eigenvectors with larger eigenvalues of $\sdpmatrix$ represent the directions that matter more to the voter.
Now, if voter $0$ has an $\sdpmatrix$-skewed preference, then they aim to minimize the $\sdpmatrix$-Mahalanobis norm between the result of $\vote(\strategicvote{}, \paramfamily{})$ and the target vector $\targetvector$. 
This leads us to define strategyproofness for skewed preferences as follows.

\begin{definition}%
\label{def:skewed_preferences}
  $\vote$ is $\strategyproofbound$-strategyproof for an $\sdpmatrix$-skewed preference if, for any others' vectors $\paramfamily{} \in \setR^{d \times \VOTER}$, any target vector $\targetvector \in \setR^d$ and any strategic vote $\strategicvote{0} \in \setR^d$, 
   The maximal strategic $\sdpmatrix$-skewed gain is at most $1+\strategyproofbound$, i.e.
  \begin{equation*}
    \forall \targetvector, \strategicvote{} \mathsep
    \norm{\vote(\targetvector, \paramfamily{}) - \targetvector}{\sdpmatrix} \leq (1+\strategyproofbound) \norm{\vote(\strategicvote{0}, \paramfamily{}) - \targetvector}{\sdpmatrix}.
  \end{equation*}
\end{definition}

This notion can then be straightforwardly adapted to define asymptotic $\strategyproofbound$-strategyproofness.

{
\subsection{The Geometric Median}
\label{sec:gm_def}
In this paper, we study the strategyproofness property of a particular \vote{}, i.e., the geometric median. It can be defined for $1+\VOTER$ voters using the average of distances between a vector $\parz$ and the reported vectors:
\begin{equation*}
  \Loss (\strategicvote{0}, \paramfamily{}, \parz) \triangleq \frac{1}{1+\VOTER} \left( \norm{\parz - \strategicvote{0}}{2} + \sum_{\voter \in [\VOTER]} \norm{\parz - \paramsub{\voter}}{2}\right).
\end{equation*}
We can now precisely define the geometric median.

\begin{definition}
A geometric median $\GeometricMedian$ operator is a function $\setR^{d\times{(1+\VOTER)}} \rightarrow \setR^d$ that outputs a minimizer of this average of distances, i.e., for any inputs $\strategicvote{0} \in \setR^d$ and $\paramfamily \in \setR^{d \times {\VOTER}}$, we must have $\GeometricMedian(\strategicvote{0}, \paramfamily) \in \argmin_{\parz \in \setR^d} \Loss (\strategicvote{0}, \paramfamily, \parz)$.
\end{definition}

In dimension $d \geq 2$, the uniqueness of $\GeometricMedian(\strategicvote{0}, \paramfamily)$ can be guaranteed \textcolor{black}{when all vectors do not lie on a $1$-dimensional line} (Proposition~\ref{prop:uniqueness} in Appendix~\ref{sec:GM_uinique}).
Interestingly, the geometric median can be regarded as the result of a dynamic process, 
where, each voter pulls a point $\parz$ towards their preferred vector with a unitary force. 
The geometric median is the equilibrium point, 
when all forces acting on $\parz$ cancel out.
It thus verifies the fairness principle \emph{``one voter, one vote with a unit force''}. Formal discussion is provided in Appendix~\ref{sec:unit_force}.
}

%% file: strategyproofness.tex
\section{MANIPULABILITY AND STRATEGYPROOFNESS}
\label{sec:non_strategyproofness}

While the average is arbitrarily manipulable by a single voter \citep{BlanchardMGS17}, the geometric median is robust even to a collusion of a strict minority of voters. However, we prove that the geometric median is not (even approximately) strategyproof in the general case.

\subsection{The Geometric Median is Not Arbitrarily Manipulable}
\label{sec:not_manipulable}

As opposed to the average, a strategic voter cannot arbitrarily manipulate the geometric median. 
This property is sometimes known as \emph{Byzantine resilience} in distributed computing, or as \emph{statistical breakdown} in robust statistics. 
Here, we state it in the terminology of computational social choice, and we consider a slightly more general setting than \emph{individual} manipulation.
Namely, we consider \emph{group} manipulation, by allowing a set of voters to collude.
Even then, strategic voters can at most have a bounded impact.
The proof of this result \textcolor{black}{which is adapted from~\citet{lopuhaa1989}} is given in Appendix \ref{sec:proof_nonmanipulabe}.

\begin{proposition}[\cite{lopuhaa1989}]
\label{th:byzantine}
  The geometric median is not arbitrarily manipulable by any minority of colluding voters.
\end{proposition}

This result shows that a  minority of strategic voters whose target vectors differ a lot from a large majority of other voters' reported vectors do not
have full control over the output of the geometric median.

\subsection{The~Geometric~Median~is~Not~\texorpdfstring{$\strategyproofbound$}{a}-Strategyproof}
\label{sec:not_strategyproff}

\begin{figure}%

    \centering
    \includegraphics[width=.43\textwidth]{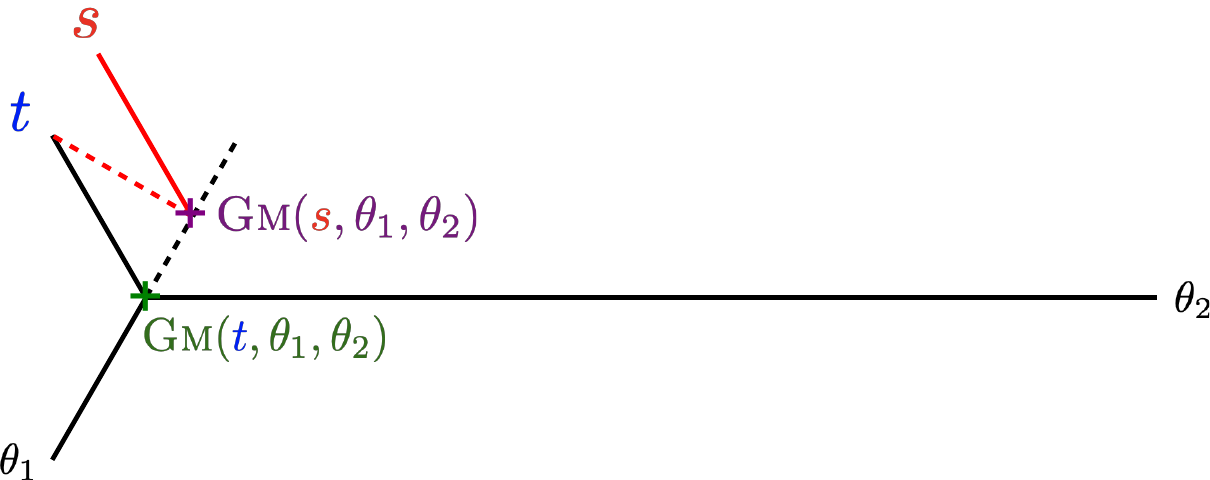}
    \caption{A simple example where the geometric median fails to be strategyproof. This example is easy to analyze in the limit where $\paramsub{2}$ is infinitely far on the right, in which case its pull is always towards the right. Since the unit pulls of all voters must cancel out, there must be a third of a turn between any two unit pull. This shows why, as the strategic voter reports $\strategicvote{0}$ rather than their target vector $\targetvector$, the geometric median moves up the dotted line, closer to $\targetvector$.}
    \label{fig:geometricmedian}
\end{figure}

The (geometric) median is slightly ill-behaved in dimension~1, when $1+\VOTER$ is even.
Typically, if $\VOTER =1$, $\strategicvote{} =  \targetvector= 0$ and $\paramsub{1} = 1$, then any point between $0$ and $1$ is a geometric median (according to our definition). 
A common solution for this case is to take the middle point of the interval of the middle vectors. However, this solution now fails to be strategyproof. Indeed, voter $0$ could now obtain $\GeometricMedian{} (\strategicvote{0}, \paramsub{1}) = \targetvector$ by reporting $\strategicvote{0} = -1$. To retrieve strategyproofness in this setting, \citet{Moulin80} essentially proposed to add one (or any odd number of) fictitious voters.
But, in higher dimensions, even when it is perfectly well-defined, the geometric median fails to guarantee strategyproofness.
Figure \ref{fig:geometricmedian} provides a simple proof of this, where voter $0$ can gain by a factor of nearly $2\sqrt{3}/3 \approx 1.15$.
Below, we prove a stronger result.

\begin{theorem}
\label{th:geometric_median_not_strategyproof}
  Even under $\dim \paramfamily \geq 2$, there is no value of $\strategyproofbound$ for which the geometric median is $\strategyproofbound$-strategyproof.
\end{theorem}

This more precise result has important implications: if a voter knows  they gain a lot by strategic misreporting, then they will more likely invest in, e.g., business intelligence, to optimize their (mis)reporting. 
Their reported preferences will then more likely diverge from their honest preferences.
We sketch the proof of Theorem~\ref{th:geometric_median_not_strategyproof} below. The full proof is highly non-trivial and is given in Appendix~\ref{sec:proof_th_geometric_median_not_strategyproof}.%

\begin{proof}[Sketch of proof]
We study the achievable set $\setAchieve$,
gathering all the possible values of the geometric median that a strategic voter can achieve by strategically choosing their reported vector.
First we show that this set 
is the set
\begin{equation}
        \setAchieve \triangleq \set{\parz \in \setR^d \st \exists \subgradient \in \nabla_\parz \Loss_{}(\paramfamily,\parz) \mathsep \norm{\subgradient}{2}\leq 1/\VOTER }, \label{eq:achivable_set}
\end{equation}
of points $\parz$ where the loss restricted to other voters $\voter \in [\VOTER]$ has a subgradient of norm at most $1/\VOTER$ (Lemma~\ref{lemma:achievable_set}).
The proof of the theorem then corresponds to the example of Figure \ref{fig:ellipsoid}, where other voters' vectors are nearly one-dimensional.
For a large number of voters, we prove, the achievable set is approximately a very flat ellipsoid defined by a matrix $H$ that has very different eigenvalues.
Then we show that the target vector $\targetvector$'s pull is heavily skewed compared to the normal to the ellipsoid. 
This implies that voter $0$ can obtain a significantly better geometric median by misreporting their target vector. 
\end{proof}

\begin{figure}[ht!]
    \centering
    \includegraphics[width=0.45\textwidth]{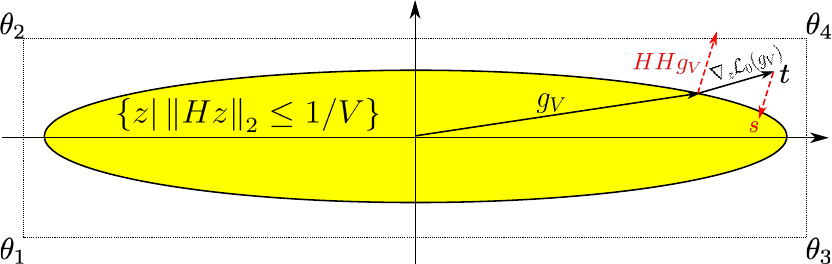}
    \caption{Illustration of the example where the geometric median fails to be $\strategyproofbound$-strategyproof, for any value of $\strategyproofbound$. Voters $v \in [4]$ report vectors that are nearly one dimensional. In the limit of a large number of voters, the achievable set for voter $0$ is an ellipsoid. But the pull of voter $0$'s preferred vector turns out to be skewed compared to the normal to the ellipsoid. This means that voter $0$ can obtain a significantly better geometric median by misreporting their preference.}
    \label{fig:ellipsoid}
\end{figure}

Interestingly, on the positive side, the proof of Theorem~\ref{th:geometric_median_not_strategyproof} requires the strategic voter's target vector to take very precise locations to gain a lot by lying.
Thus, while the geometric median has failure modes where some voters have strong incentives to misreport their preferences, in practice, such incentives are unlikely to be strong.
On the negative side, our proof suggests the possibility of a vicious cycle.
Namely, it underlines the fact that a strategic voter's optimal strategy is to report a vector that is closer to the subspaces where other voters' vectors mostly lie. 
These subspaces may be interpreted as the more polarized dimensions.
As a result, if all voters behave strategically, we should expect the reported vectors to be even more flattened on these subspaces than voters' true target vectors.
But then, if voters react strategically to the other voters' strategic votes, there are even more incentives to vote according to the one-dimensional line.
In other words, the geometric median seems to initiate a vicious cycle where strategic voters are incentivized to escalate strategic behaviours, and this would lead them to essentially ignore all the dimensions of the vote except the most polarized one.

%% file: asymptotic.tex
\section{ASYMPTOTIC STRATEGYPROOFNESS}
\label{sec:asymptotic_strategyproofness}

Our negative result of the previous section encourages us to weaken the notion of strategyproofness.
We do so by replacing the bound on voters' strategic gains for \emph{all} other voters' inputs with a bound for \emph{most} of other inputs.
We assume that each voter $\voter$ reports a vector $\paramsub{\voter}$ drawn independently from a probability distribution $\paramdistribution$.
We then study the maximal strategic gain of voter $0$, when there are many other voters whose reported vectors are obtained this way.
Any bound $\strategyproofbound$ that holds with high probability as the number $\VOTER$ of voters is sufficiently large guarantees what we call \emph{asymptotic $\strategyproofbound$-strategyproofness} (see Definition~\ref{def:asymptotic-strategyproofness}).

Throughout this section, we consider a given fixed distribution $\paramdistribution$ of voters' reported vectors.
Our main result relies on the following mild smoothness assumption about the distribution $\paramdistribution$ of other voters' vectors, which is clearly satisfied by numerous classical probability distributions over $\setR^d$, like the normal distribution (with $\PARAM \triangleq \setR^d$).

\begin{assumption}
\label{ass:pdf}
  There is a convex open set $\PARAM \subseteq \setR^d$, with $d \geq 5$, 
  such that the distribution $\paramdistribution$ yields a probability density function $p$ continuously differentiable on $\PARAM$, 
  and such that $\probability{\param \in \PARAM} = 1$ and $\expect{\norm{\param}{2}} = \int_{\setR^d} \norm{\param}{2} p(\param) d\param < \infty$.
\end{assumption}

To simplify notations, we leave the dependence to the distribution implicit.
For any number $\VOTER \in \setN$ of other voters, we denote by $\paramfamily_{\VOTER} \in \setR^{d \times{\VOTER}}$ the random tuple of the $\VOTER$ voters' reported vectors,
and we define 
\begin{equation}
  \Loss_{1:\VOTER} (\parz) \triangleq \frac{1}{\VOTER} \sum_{\voter \in [\VOTER]}\norm{\parz-\paramsub{\voter}}{2},
\end{equation}
 and $\geometricmedian_{1:\VOTER} \triangleq \GeometricMedian(\paramfamily_{\VOTER})$ the random average of distances and the geometric median for the voters $\voter \in [\VOTER]$.
We denote by $\Loss_{0:\VOTER} (\strategicvote{0}, \parz)$ and $\geometricmedian_{0:\VOTER} \triangleq \GeometricMedian(\strategicvote{0}, \paramfamily_{\VOTER})$ the similar quantities that also include voter $0$'s strategic vote $\strategicvote{0}$,
and $\geometricmedian^\dagger_{0:\VOTER} \triangleq \GeometricMedian(\targetvector, \paramfamily_\VOTER)$ the truthful geometric median, which results from voter $0$'s truthful reporting of $\targetvector$.

\input{asymptotic_infinite_limit}

\input{asymptotic_main_result}

\input{asymptotic_skewness}

%% file: asymptotic_infinite_limit.tex
\subsection{Infinite Limit}
\label{sec:infinite_limit}

Consider the limit where $\VOTER \rightarrow \infty$. The distribution $\paramdistribution$ defines its own average-of-distance function:
\begin{equation}
  \Loss_\infty (\parz) \triangleq \expectVariable{\param\sim\paramdistribution}{\norm{\parz - \param}{2}}.
\end{equation}
We say that $\geometricmedian_\infty$ is a geometric median of the distribution $\paramdistribution$ if it minimizes the loss $\Loss_\infty$. 
Under Assumption \ref{ass:pdf}, the support of $\paramdistribution$ is of full dimension $d$, which guarantees the uniqueness of the geometric median (Proposition~\ref{porp:positive difinite} in Appendix~\ref{sec:proof_asymptotic_strategyproofness}). 
We denote by $\hessian_\infty \triangleq \nabla^2 \Loss_\infty (\geometricmedian_\infty)$ the Hessian at the geometric median. The properties of this matrix will be central to the strategyproofness of the geometric median. 

{
\begin{remark} [on the smoothness assumption]
Note that Assumption~\ref{ass:pdf} is a mild technical assumption, which intuitively guarantees that, for a sufficiently large number of voters, 
the infinite limit case will be approximately recovered. This will allow us to invoke some statistics of $\paramdistribution$ to derive our  strategyproofness bounds. In practice, assuming there are sufficiently many voters, then $\paramdistribution$ may be estimated by the empirical distribution of the reported vectors.
\end{remark}
}

%% file: asymptotic_main_result.tex
\subsection{The Geometric Median is Asymptotically \texorpdfstring{$\strategyproofbound$}{a}-Strategyproof}

One of our main results is that the geometric median is asymptotically $\strategyproofbound$-strategyproof, for some appropriate value of $\strategyproofbound$ that depends on the skewness of the Hessian matrix $\hessian_\infty$. 
 We define
the skewness of a positive definite matrix $\sdpmatrix$ by 
  \begin{align}
  \label{equ:skewness_definition}
    \skewness(\sdpmatrix)
    &\triangleq \sup_{\parx \neq 0} \set{\frac{\norm{\parx}{2} \norm{\sdpmatrix \parx}{2}}{\parx^T \sdpmatrix \parx} -1}\\
    &= \sup_{\norm{\unitvector{}}{2} = 1} \set{\frac{\norm{\sdpmatrix \unitvector{}}{2}}{\unitvector{}^T \sdpmatrix \unitvector{}} -1}.\nonumber
  \end{align}
This quantity bounds the angle between a vector $x$ and its linear transformation $Sx$. It is straightforward that $\skewness(\beta \sdpmatrix) = \skewness(\sdpmatrix)$ for all $\beta >0$. 
Also the identity matrix has no skewness $(\skewness(I) = 0)$. Intuitively, the more $\sdpmatrix$ distorts the space, typically by having very different eigenvalues, the more skewed it is. In Section~\ref{sec:skewness}, we derive upper and lower bounds on $\skewness$. 
We can now present our main theorem. 

\begin{theorem}
\label{th:asymptotic_strategyproofness}
Under Assumption~\ref{ass:pdf}, the geometric median is asymptotically $\skewness(\hessian_\infty)$-strategyproof.
\end{theorem}

Intuitively, the more the distribution of the reported vectors is flattened along some dimensions, which can be interpreted as more polarized dimensions, the worse the strategyproofness bound is. The proof of this theorem is given in Appendix~\ref{sec:proof_steps_main}.  In the next section, we provide a brief proof sketch to help the readers follow our reasoning.
\subsection{Proof Techniques and Technical Challenges}

\input{proof_sketch_new}

\begin{figure}[t]
    \begin{center}
    \includegraphics[width=.44\textwidth]{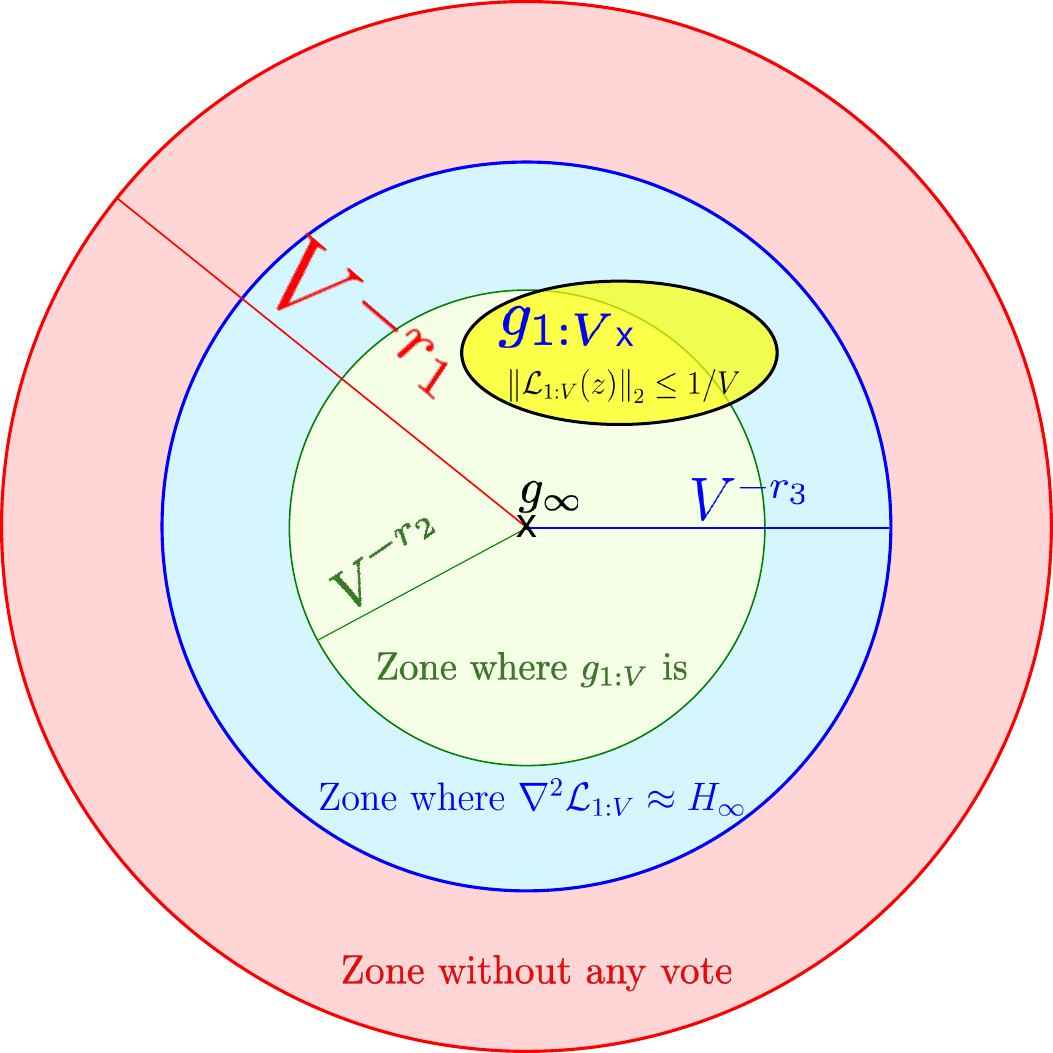}
    \end{center}
    \vspace{-5mm}
    \caption{Illustration of the proof strategy for Theorem \ref{th:asymptotic_strategyproofness}, which is based on the following claims that hold with high probability, for $2/d<2\ra<\rb<\rc<1/2$, and for $\VOTER$ large enough.
    First, there is no vote in $\ball(\geometricmedian_\infty, \VOTER^{-\ra})$ (a ball centered on $\geometricmedian_\infty$, and of radius $\VOTER^{-\ra}$). Thus, $\Loss_{1:\VOTER}$ is infinitely differentiable there.
    Moreover, the second and third derivatives of $\Loss_{1:\VOTER}$ cannot be too different from the second and third derivatives of $\Loss_\infty$ in $\ball(\geometricmedian_\infty, \VOTER^{-\rb})$. 
    Plus,  $\geometricmedian_{1:\VOTER}$ lies in $\ball(\geometricmedian_\infty, \VOTER^{-\rc})$, and the set of geometric medians that voter $0$ can obtain by misreporting their preferences is approximately an ellipsoid centered on $\geometricmedian_{1:\VOTER}$. This ellipsoid lies completely inside $\ball(\geometricmedian_\infty, \VOTER^{-\rb})$.
    }
    \label{fig:proof_strategy}
\end{figure}

\color{black}

%% file: proof_sketch_new.tex
\begin{figure}%
    \centering
    \includegraphics[width=.3\textwidth]{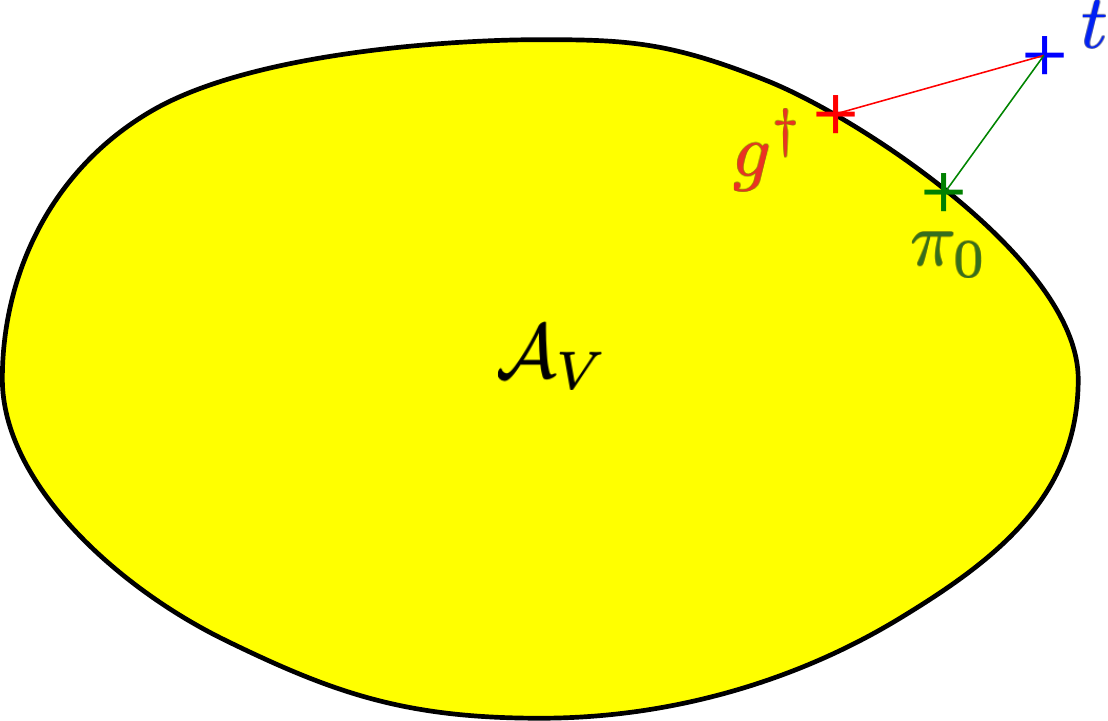}
    \caption{Illustration of our proof strategy. 
    For a large number of voters, the achievable set $\setAchieve$ is approximately an ellipsoid. 
    To derive the strategyproof bounds, we study the orthogonal projection $\pi_0$ of the target vector $\targetvector$.
Strategyproofness then depends on the angle between $\targetvector - g^\dagger$ and $\targetvector - \pi_0$,
    which we derive from the skewness of the positive definite matrix that approximately defines the ellipsoid. 
    }
    \label{fig:sketch_proof}
\end{figure}

The proof of Theorem~\ref{th:asymptotic_strategyproofness} relies on the following steps:
{
\paragraph{\textbf{1. Approximating the Achievable Set with an Ellipsoid.}} We first consider the infinite-case   assuming a strategic voter with a very small voting power $\varepsilon$ where the voting power of each voter is the magnitude of their pull compared to the sum of all pulls (see Section~\ref{sec:gm_def}).
By analyzing the Taylor's approximation of the gradient of the loss function for other voters, we show that the achievable set for the strategic voter (defined in \eqref{eq:achivable_set}) becomes approximately an ellipsoid as $\varepsilon \rightarrow 0$.
Now, as shown in Figure~\ref{fig:sketch_proof}, since the ellipsoid is convex, the best-possible achievable point for the strategic voter is the orthogonal projection $\pi_0$ of the target vector $\targetvector$ on the ellipsoid.
By comparing the distance between $\targetvector$ and $\pi_0$ to the distance between $\targetvector$ and the geometric median $\geometricmedian^\dagger$ obtained by a truthful reporting of $\targetvector$, we then obtain what the strategic voter can gain by behaving strategically, in the infinite-voter case where they have a very small voting power $\varepsilon$.
Intuitively, the more flattened the ellipsoid, the more the strategic voter can gain; conversely, 
for a quasi-hyperspherical ellipsoid,
the strategic voter cannot gain by misreporting.

\paragraph{\textbf{2. Deriving a Finite-voter Case from the Infinite One.}}
To obtain meaningful strategyproofness guarantees, we consider the  finite-voter case with a large (but not infinite) number of voters.
Unfortunately, the finite-voter case is trickier than the infinite-voter case.
To retrieve the strategyproofness bound,  we need in addition to  bound the divergence between the finite-voter case and the infinite-voter case.
Fortunately, for $\VOTER$ large enough, the voting power of a single strategic voter is small, which allows us to quasi-reduce the finite-voter case to the infinite-voter case. In fact, one important challenge of the proof is to leverage the well-behaved smoothness of the infinite-voter case to derive bounds for the finite-voter case, where singularities and approximation bounds make the analysis trickier. Indeed, while the infinite-voter loss function is smooth enough everywhere (under Assumption~\ref{ass:pdf}), the finite-voter loss function is not differentiable everywhere. 
At any point $\paramsub{\voter}$, it yields a nontrivial set of subgradients.
This complicates the analysis, as we exploit higher order derivatives.

To address this difficulty, we identify different regions around the infinite-voter geometric median where the finite-voter loss function is well-behaved enough as shown in Figure~\ref{fig:proof_strategy}.
Namely, in high dimensions, assuming a smooth distribution $\paramdistribution$,  the distances between any two randomly drawn vectors are large.
Concentration bounds allow us to guarantee that, with high probability, other voters' vectors $\paramsub{\voter}$  are all far away from the infinite geometric median $\geometricmedian_\infty$  (Lemma~\ref{lemma:no_voter} in the Appendix).
This has two important advantages. 
First, it guarantees the absence of singularities in a region around $\geometricmedian_\infty$.
Second, and more importantly, it allows us to control the variations of higher-order derivatives in this region (Lemma~\ref{lemma:close_hessian}).
This turns out to be sufficient to guarantee that the finite-voter geometric median is necessarily within this region.
\paragraph{\textbf{3. Controlling the Largeness of the Third Derivative Tensor.}}
Another challenge that we encountered was to guarantee that the achievable set in the finite-voter setting is convex.
This condition is indeed critical to provide an upper bound on $\strategyproofbound$, since it enables us to determine the strategic voter's optimal strategy by studying the orthogonal projection of the target vector onto the achievable set.
To prove this condition, we identify a sufficient condition, which involves the third derivative tensor of the finite-voter loss function (lemmas~\ref{lemma:convexity} and~\ref{lemma:convex_projection_strategyproof}).
Fortunately, just as we manage to guarantee that the finite-voter geometric median is necessarily close enough to the infinite-voter geometric median (Lemma~\ref{lemma:close_median}), using similar arguments based on concentration bounds, we successfully controlled the largeness of the third derivative tensor (Lemma~\ref{lemma:bounded_third_derivative}).
Therefore, for a large number of voters and with high probability, the achievable set is convex. Additionally, it is approximately an ellipsoid, which is characterized by the infinite-voter Hessian matrix $\hessian_\infty$. As a result, and since ``rounder'' ellipsoids yield better strategyproofness guarantees, when the number of voters is sufficiently large, the strategic gain of a strategic voter is upper-bounded by how skewed the infinite-voter Hessian matrix $\hessian_\infty$ is. 
}

%% file: asymptotic_skewness.tex
\subsection{Bounds on $\skewness$}
\label{sec:skewness}

As we saw, the asymptotic strategyproofness of the geometric median depends on the skewness of the Hessian matrix $\hessian_\infty$, defined in Equation (\ref{equ:skewness_definition}). 
In this section, we derive upper and lower bounds on the skewness function based on the ratio of the extreme eigenvalues of the matrix.
Intuitively, the more different the eigenvalues of $\sdpmatrix$ are, the more skewed $\sdpmatrix$ is. 
We formalize this intuition with upper and lower bounds, whose proofs are given in Appendix \ref{sec:proof_prop_skewness_lowerbound}.  
\begin{proposition}
\label{prop:skewness_lowerbound}
  Denote $\Lambda \triangleq \frac{\max \spectrum(\sdpmatrix)}{\min \spectrum(\sdpmatrix)}$ the ratio of extreme eigenvalues of $\sdpmatrix$. Then
    $\frac{1 + \Lambda}{2 \sqrt{\Lambda}} - 1
    \leq \skewness(\sdpmatrix) 
    \leq \Lambda - 1$.
  In dimension 2, the lower-bound inequality is an equality.
\end{proposition}

%% file: skewed_preferences.tex
\section{SKEWNESS GENERALIZATIONS}
\label{sec:skewed_preferences}
We generalize our main result in two aspects. First, we consider skewed preferences where users give different weights to different dimensions. Second, we study the skewed geometric median which can be derived by re-scaling the space before computing the geometric median.

\subsection{Skewed Preferences}

Our analysis so far rested on the assumption that voters have single-peaked preferences, which depend on the Euclidean distance between the geometric median and their preferred vectors. 
While this makes our analysis simpler, in practice, this assumption is not easy to justify.
In fact, it seems reasonable to assume that some dimensions have greater importance for voters than others.

This motivates us to introduce $\sdpmatrix$-skewed preferences, for a positive definite matrix $\sdpmatrix$. More precisely, we say that a voter $\voter$ has an $\sdpmatrix$-skewed preference if they aim to minimize $\norm{\geometricmedian - \paramsub{\voter}}{\sdpmatrix}$,
where $\geometricmedian$ is the result of the vector vote and $\norm{\parz}{\sdpmatrix} \triangleq \norm{\sdpmatrix \parz}{2}$ is  the $\sdpmatrix$-Mahalanobis norm. 
Intuitively, the matrix $\sdpmatrix$ allows us to highlight which directions of space matter more to voter $\voter$.
For instance, if $\sdpmatrix = \begin{pmatrix} Y & 0 \\ 0 & 1 \end{pmatrix}$, with $Y \gg 1$, it means that the voter gives a lot more importance to the first dimension than to the second dimension.

%% file: skewed_geometric_median.tex
\subsection{The Skewed Geometric Median}
\label{sec:skewed_geometric_median}

Intuitively, to counteract voters' strategic exaggeration incentives,
we could make it more costly to express strong preferences along the more polarized and more important dimensions.
In other words, voters would have a unit force along less polarized dimensions, and a less-than-unit force along more polarized dimensions.
We capture this intuition by introducing ``$\adaptivematrix$-skewed geometric median'' for a  positive definite matrix $\adaptivematrix \succ 0$.

\paragraph{\textbf{Skewed Loss.}}
We define the $\adaptivematrix$-skewed infinite loss  as $$\Loss^\adaptivematrix_\infty (\parz, \paramdistribution) \triangleq \expect_{\param\sim\paramdistribution}{\norm{\parz - \paramsub{} }{\adaptivematrix}},$$ using the $\adaptivematrix$-Mahalanobis norm ($\norm{\parz}{\adaptivematrix} \triangleq \norm{\adaptivematrix \parz}{2}$), and we call $\adaptivematrix$-skewed geometric median $\geometricmedian_\infty^\adaptivematrix$ its minimum.
We also introduce their finite-voter equivalents, for $1+\VOTER$ voters, by
\begin{equation*}
    \Loss^\adaptivematrix_{0:\VOTER} (\strategicvote{0}, \parz)
    \triangleq \frac{1}{1+\VOTER} \norm{ \strategicvote{0} - \parz }{\adaptivematrix} + \frac{1}{1+\VOTER} \sum_{\voter \in [\VOTER]} \norm{ \paramsub{\voter} - \parz }{\adaptivematrix},
\end{equation*}
and $\geometricmedian^\adaptivematrix_{0:\VOTER} = \argmin_{\parz} \Loss^\adaptivematrix_{0:\VOTER} (\strategicvote{0}, \parz)$.
Intuitively, this is equivalent to mapping the original space to a new space using the linear transformation $\adaptivematrix$, and computing the geometric median in this new space (Lemma~\ref{lemma:skewedGMComp} in Appendix~\ref{sec:proof_skewed_preferences}).

\begin{remark}
Interestingly, we also show that this skewed geometric median can be interpreted as modifying the way we measure the norm of voters' forces in the original space, thereby guaranteeing its consistency with the fairness principle \emph{``one voter, one vote with a unit force''}. The formal discussion is given in Appendix~\ref{sec:alternative}.
\end{remark}

\subsection{Strategyproofness of the Skewed Geometric Median for Skewed Preferences}

For any skewing positive definite matrix $\adaptivematrix$, we define 
$\hessian^\adaptivematrix_\infty \triangleq \nabla^2 \Loss^\adaptivematrix_\infty (\geometricmedian_\infty^\adaptivematrix)$ the Hessian matrix of the skewed loss at the skewed geometric median.
We then have the following asymptotic strategyproofness guarantee for an appropriately skewing matrix. The sketch of the proof is provided in Appendix \ref{sec:proof_skewed_preferences}.

\begin{theorem}
\label{th:skewed_geometric_median}
Under Assumption~\ref{ass:pdf}, the $\adaptivematrix$-skewed geometric median is asymptotically $\skewness(\sdpmatrix^{-1} \hessian^\adaptivematrix_\infty \sdpmatrix^{-1})$-strategyproof for a voter with $\sdpmatrix$-skewed preferences.
In particular, if $\hessian^\adaptivematrix_\infty = \sdpmatrix^{1/2}$, then the $\adaptivematrix$-skewed geometric median is asymptotically strategyproof for this voter.
\end{theorem}

\paragraph{\textbf{Interpretation.}}
Let us provide additional insights into what the theorem says.
Intuitively, the theorem asserts that the strategyproofness of the normal geometric median ($\adaptivematrix = I$) depends on how much an individual cares about polarized dimensions.
More precisely, the more the voter cares about polarized dimensions, the less strategyproof the geometric median is.

Indeed, suppose that the first dimension is both highly polarized and very important to voter $0$. 
The fact that it is polarized would typically correspond to a Hessian matrix of the form $\hessian_\infty = \begin{pmatrix} 1 & 0 \\ 0 & X^2 \end{pmatrix}$, with $X \gg 1$ 
(see the proof of Theorem \ref{th:geometric_median_not_strategyproof}).
The fact that voter $0$ cares a lot about the first dimension would typically correspond to a skewed preference matrix $\sdpmatrix = \begin{pmatrix} Y & 0 \\ 0 & 1 \end{pmatrix}$, with $Y \gg 1$.
We then have $\sdpmatrix^{-1} \hessian_\infty \sdpmatrix^{-1} = \begin{pmatrix} Y^{-2} & 0 \\ 0 & X^2 \end{pmatrix}$.
By Proposition \ref{prop:skewness_lowerbound}, we then have $\skewness(\sdpmatrix^{-1} \hessian_\infty \sdpmatrix^{-1}) = \frac{X^2+Y^{-2}}{2 \sqrt{X^2 Y^{-2}}} - 1 = \Theta(XY)$, which is very large for $X,Y \gg 1$.
In particular, this makes the normal geometric median unsuitable for voting problems where some dimensions are much more polarized and regarded as important by most voters. {Now, interestingly, if we find a skewing matrix $\adaptivematrix$ that weakens the voters' pulls in the first dimensions, making the Hessian matrix approximately  $\hessian^\adaptivematrix_\infty \approx  \begin{pmatrix} 1 & 0 \\ 0 & \frac{1}{Y^2} \end{pmatrix}$, then the resulting geometric median becomes asymptotically strategyproof.}
\paragraph{\textbf{Remarks on the Skewed Hessian Matrix.}}
\label{sec:skewed_hessian_matrix}
In general, $\geometricmedian_\infty^\adaptivematrix \neq \geometricmedian_\infty$ (Proposition~\ref{prop:invariance_linear_transformation} in Appendix~\ref{sec:GM}).
This makes identifying a skewing matrix $\adaptivematrix$ such that $\hessian^\adaptivematrix_\infty = \sdpmatrix^{1/2}$ challenging. In particular, it is hard to determine how such a matrix relates to the statistics of $\paramdistribution$.
We note however the following connection between the Hessian matrix $\nabla^2 \Loss^\adaptivematrix_\infty (\parz)$ of the $\adaptivematrix$-skewed loss and the Hessian matrix $\nabla^2 \Loss_\infty (\parz)$ of the Euclidean loss. The proof is given in Appendix \ref{sec:proof_proposition_skewed_hessian}.

\begin{proposition}
\label{proposition:skewed_hessian}
For any $\parz \in \setR^d$, we have $\nabla^2 \Loss^\adaptivematrix_\infty (\parz) = \adaptivematrix (\nabla^2 \Loss_\infty) (\adaptivematrix \parz, \adaptivematrix \paramdistribution) \adaptivematrix$.
\end{proposition}

Note that in particular, if $\geometricmedian_\infty^{\left(\hessian_\infty^{-1/2}\right)} = \geometricmedian_\infty$ and if $\nabla^2 \Loss_\infty (\hessian_\infty^{-1/2} \parz, \hessian_\infty^{-1/2} \paramdistribution) = \hessian_\infty$, then the $\hessian_\infty^{-1/2}$-skewed geometric median is asymptotically strategyproof.
This will be the case if the support of $\paramdistribution - \geometricmedian_\infty$ lies in the union of the eigenspaces of $\adaptivematrix$, as this implies that, when 
$\param{}$ is drawn from $\paramdistribution$, the vectors $\adaptivematrix \param{} - \adaptivematrix \geometricmedian_\infty$ and $\param{} - \geometricmedian_\infty$ are colinear and point in the same direction with probability 1.
But, in general, these assumptions do not hold. 
This makes the computation of the appropriate skewing challenging.
We thus leave open the problem of proving the existence and uniqueness (up to overall homothety) of such a matrix, as well as the design of algorithms to compute it.

{\color{red}
}

%% file: experiments.tex
\section{NUMERICAL EXPERIMENT}
\label{sec:exp}

{Strategyproofness is commonly studied purely theoretically,
as empirical strategyproofness evaluation is hard to perform in a meaningful and fair way.
Indeed, it requires identifying optimal attacks against a system,
which often amounts to solving an intractable optimization problem.
In particular, if such an empirical evaluation fails to find an effective attack,
it is unclear if this is because no such attack exists, or because no such attack has been found.
}
Nevertheless, here we provide  a simple experiment to evaluate the effect of the (skewness of the) underlying distribution on the strategic gain  $\alpha$ when using the geometric median to aggregate voters' vectors.
First, we sample $500000$ vectors from a $2$ dimensional Gaussian distribution $\paramdistribution$ with mean $0$ and covariance matrix of $\begin{pmatrix} c & 0 \\ 0 & \frac{1}{c} \end{pmatrix}$ for a parameter $c$. Note that as shown in Proposition~\ref{prop:skewness_lowerbound}, $c$ is closely related to the skewness of distribution $\paramdistribution$. We assume the strategic voters have a $1\%$ voting power, i.e., we simulate $5000$ strategic voters all with the same target vector $t$. Then, to find a vulnerable target vector, we use  a heuristic idea similar to that of Figure~\ref{fig:ellipsoid}. Essentially, in each dimension, we find the extreme achievable geometric median for the strategic voters. The target vector $t$ is then the combination of these extreme values of both dimensions.  Finally, We approximately find the maximum strategic gain by performing a grid search of the best reported vector $s$ in a neighborhood of $t$. Figure~\ref{fig:experiment} shows the dependence of the strategic gain $\alpha$ on parameter $c$ and  validates the intuition that the more skewed the distribution, the less strategyproof geometric median is. This experiment demonstrates that the skewness of the underlying distribution is a crucial factor to consider when assessing the strategyproofness of the geometric median. The code is available at [{\color{blue}\href{https://github.com/sadeghfarhadkhani/GM_Startegyproof}{this link}}].

\begin{figure}%
    \centering
    \includegraphics[width=.38\textwidth]{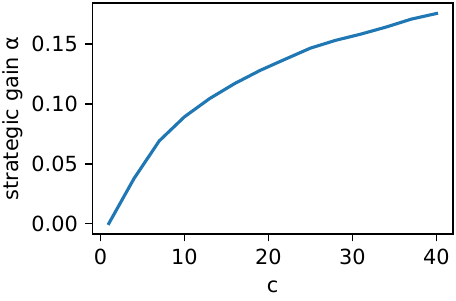}
    \caption{Dependence of the maximum strategic gain $\alpha$ on parameter $c$, where $c$ is the square root of condition number of the underlying distribution's covariance matrix.}
    \label{fig:experiment}
\end{figure}

\color{black}
\vspace{-3mm}

%% file: related_work2.tex
\section{RELATED WORK}
\label{sec:related_work}
\vspace{-2mm}
Strategyproofness in one dimension has been extensively studied~\citep{Moulin80,Procaccia13,AAAIW1510182}. 
It was shown~\citep{Moulin80}  that a generalized form of the median is group strategyproof,
and  that the randomized Condorcet voting system is also group strategyproof for single-peaked preferences~\citep{Hoang17}. 
The one-dimension median was also leveraged
for mechanism design without payment~\citep{Procaccia13}.

However, generalizing the median to higher dimensions is not straightforward~\citep{lopuhaa1989}. A common generalization known as the {\it coordinate-wise median}, was shown to be 
strategyproof, but not 
group strategyproof~\citep{Sui15}. 
The extent to which the generalized coordinate-wise median and the quantile mechanism are $\strategyproofbound$-(group)-strategyproof have been studied by~\citet{Sui15}, though their definition slightly diverges from ours (their error is additive, not multiplicative). 
Remarkably, it was shown by~\citet{KIM198429} that, in dimension 2, the only strategyproof, anonymous and continuous voting system is the (generalized) coordinate-wise median. 

Without restricting the dimension, but assuming the vectors to be taken from compact subsets of Euclidean spaces, strategyproof voting systems were characterized assuming   
all voters have generalized single-peaked preferences~\citep{barbera1998strategy}. This approach built upon~\citet{border1983straightforward} which characterized strategyproof voting systems for Cartesian product ranges. In both cases,  the set of strategyproof voting systems was defined as the class of  \emph{generalized (coordinate-wise) median voter schemes} which were shown in the case of~\cite{barbera1998strategy} to also satisfy the intersection property introduced by\footnote{This property roughly guarantees a certain level of coordination between the decisions taken on each coordinate.}~\cite{barbera1997voting}.

Overall, the coordinate-wise median has more desirable strategyproofness than the geometric median~\citep{FarhadkhaniGH21}. 
It is also important to notice that, as opposed to the coordinate-wise median, the geometric median guarantees that the output vector belongs to the convex hull of voters' vectors (Proposition \ref{prop:convex_hull}).
This makes the coordinate-wise median unsuitable for problems where the space of relevant vectors is the convex hull of the input vectors. This holds, for instance, for the budget allocation problem, whose decision vector $\parz$ must typically satisfy $\parz \geq 0$ and $\sum \parz[i] = 1$. 
In dimension 3, if three voters have preferences $(1,0,0)$, $(0,1,0)$ and $(0,0,1)$, then the coordinate-wise median would output $(0,0,0)$ which may be undesirable. 
On the other hand, the geometric median would output $(1/3,1/3,1/3)$, which seems more desirable.
Similarly, the coordinate-wise median is unfit to aggregate covariant matrices, which must be symmetric and semi-definite positive.

Another line of work focused on bounding the approximation ratio, which is the extent to which social cost is lost by using alternative aggregation rules like coordinate-wise median \citep{goel2020,pinyan09,walsh20}. Several papers also consider other variations of this problem, e.g., choosing $k$ facility locations instead of one \citep{esco11}, assigning different weights to different nodes \citep{zhang2014}, and assuming that the nodes lie on a network represented by a graph \citep{Alon2009}. Others have addressed the computational complexity of the geometric median~\citep{CohenLMPS16}. {Another work \citep{Brady17} shows that for three agents the geometric median is the only rule that satisfies anonymity, neutrality, and Maskin-Monotonicity. }

%% file: conclusion2.tex
\section{CONCLUSION}
\label{sec:conclusion}

We analyzed different flavors of strategyproofness for the geometric median, an instrumental component of the secure machine learning toolbox.
First, we showed that, in general, there can be no guarantee of approximate-strategyproofness, by exhibiting worst-case situations.
However, we proved that, assuming that voters' vectors follow some distribution $\paramdistribution$,   asymptotic $\strategyproofbound$-strategyproofness can be ensured.
We then generalized our  results to the case where some dimensions may matter more to the voters than other dimensions. 
In this setting, we proved that the geometric median becomes {\it less} strategyproof, when some dimensions are more polarized and more important than others.
Finally, we showed how the skewed geometric median can improve asymptotic strategyproofness, by providing more voting rights along more consensual dimensions.
Overall, our analysis helps better identify the settings where the geometric median can indeed be a suitable solution to high dimensional voting.  

%% file: Appendix.tex
\begin{center}
    \LARGE \bf {Appendix}
\end{center}

\section*{Organization}

The appendices are organized as follows:
\begin{itemize}
    \item Appendix~\ref{sec:GM} proves some useful preliminary results about the geometric median that are needed in this paper.
    \item Appendix~\ref{app:non_strategyproof} includes the proofs of the results presented in Section~\ref{sec:non_strategyproofness} (in particular, the proof of Theorem \ref{th:geometric_median_not_strategyproof}).
    \item Appendix~\ref{sec:proof_asymptotic_strategyproofness} includes some proofs and deferred results from Section~\ref{sec:asymptotic_strategyproofness} (in particular, the proof of Theorem \ref{th:asymptotic_strategyproofness}).
    \item Appendix~\ref{sec:proof_skewed_preferences} includes some proofs and deferred results from Section~\ref{sec:skewed_preferences} (in particular, the proof of Theorem \ref{th:skewed_geometric_median}).
    \item Appendix~\ref{sec:alternative} discusses the notion of alternative unit forces and proves auxiliary results on the equivalence between $\ell_p$ penality and $\ell_q$-unit force vote for $\frac{1}{p} + \frac{1}{q} = 1$ and the equivalence between  $\adaptivematrix$-skewed geometric median, and $\adaptivematrix^{-1}$-unit forces.
\end{itemize}

\input{properties}

\input{APP_nonstrategyproo}

\input{APP_approximate_strategyproof}

\input{App_skewed_generalizations}

\input{alternative_unit_force}

%% file: properties.tex
\section{GEOMETRIC MEDIAN: PRELIMINARIES}
\label{sec:GM}

In this section, we characterize a few properties of the geometric median, many of which are useful for our subsequent proofs.
For the sake of exposition, we consider 
in this section a geometric median restricted to the voters $\voter \in [\VOTER]$, in which case, the loss function would be
\begin{equation}
  \Loss (\paramfamily{}, \parz) \triangleq \frac{1}{\VOTER} \sum_{\voter \in [\VOTER]} \norm{\parz - \paramsub{\voter}}{2}.
\end{equation}
The generalization to $1+\VOTER$ voters is straightforward.

\subsection{Unit Forces}
\label{sec:unit_force}
\color{black}
We first show that the geometric median verifies the fairness principle {\it ``one voter, one vote with a unit force''}.
Consider a system in which each voter $\voter$ pulls the output of voting $\parz$ towards their location $\paramsub{\voter}$ with a unit force. 
Voter $\voter$'s force is then given by the unit vector $\unitvector{\parz - \paramsub{\voter}}$ in the direction of $\parz - \paramsub{\voter}$. 
Any equilibrium of this process must then be a point $z$ where all the forces cancel out, i.e., we must essentially have $\sum_{\voter\in\VOTER} \unitvector{\parz - \paramsub{\voter}}=0$. 
Lemma~\ref{lemma:geometric-median_unit-forces} shows that this condition is equivalent to the computation of a geometric median.
But first, let us characterize the gradient of the $\ell_2$-norm.
\color{black}

\begin{lemma}
\label{lemma:unit-force}
The gradient of the Euclidean norm is a unit vector.
More precisely, for all $\parz \in \setR^d$, we have $\nabla \norm{\parz}{2} = \unitvector{\parz}$, where $\unitvector{\parz} \triangleq \parz / \norm{\parz}{2}$ if $\parz \neq 0$, and otherwise $\unitvector{0} \triangleq \ball(0,1)$ is the unit ball centered at the origin.
\end{lemma}

In the latter case, the $\ell_2$ norm thus actually has a large set of subgradients.

\begin{proof}
Assume $\parz \neq 0$. We have $\nabla \norm{\parz}{2}^2 = 2 \parz$.
As a result, $\nabla \norm{\parz}{2} = \nabla \sqrt{\norm{\parz}{2}^2} = \nabla \norm{\parz}{2}^2 / 2 \sqrt{\norm{\parz}{2}^2} = \parz / \norm{\parz}{2} = \unitvector{\parz}$.

Now consider the case $\parz = 0$. Then note that for all $\parx \in \setR^d$, we have $\norm{\parx}{2} - \norm{\parz}{2} = \parx^T \unitvector{\parx} \geq \parx^T h$ for any vector $h$ of Euclidean norm at most 1.
This proves that $\nabla_{\parz = 0} \norm{\parz}{2} \supset \ball(0,1)$.
On the other hand, if $\norm{h}{2} > 1$, then we have $\norm{\varepsilon \unitvector{h}}{2} = \varepsilon < \varepsilon \norm{h}{2} = (\varepsilon \unitvector{h})^T h$.
Thus $h$ cannot be a subgradient, and thus $\nabla_{\parz = 0} \norm{\parz}{2} = \ball(0,1) = \unitvector{0}$.
\end{proof}

As an immediate corollary, the following condition characterizes the geometric medians.

\begin{lemma}
\label{lemma:geometric-median_unit-forces}
The sum of voters' unit pulls cancel out on $\geometricmedian \triangleq \GeometricMedian (\paramfamily)$, i.e., $0 \in \sum_{\voter \in [\VOTER]} \unitvector{\geometricmedian - \paramsub{\voter}}$.
\end{lemma}

\begin{proof}
By Lemma~\ref{lemma:unit-force}, $\nabla_\parz \norm{\parz - \paramsub{\voter}}{2} = \unitvector{\parz - \paramsub{\voter}}$.
Therefore, $\VOTER \nabla_\parz \Loss(\paramfamily, \parz) = \sum_{\voter \in [\VOTER]} \unitvector{\parz - \paramsub{\voter}}$.
The optimality condition of $\geometricmedian$ then implies $0 \in \nabla_\parz \Loss (\paramfamily, \geometricmedian)$ and hence $0 \in \sum_{\voter \in [\VOTER]} \unitvector{\geometricmedian - \paramsub{\voter}}$.
\end{proof}

Before moving on, we make one last observation about the second derivative of the Euclidean norm, which is very useful for the rest of the paper.

\begin{lemma}
\label{lemma:hessian_Euclidean_norm}
Suppose $\parz \neq 0$. Then $\nabla^2 \norm{\parz}{2} = \frac{1}{\norm{\parz}{2}} \left( I - \unitvector{\parz} \unitvector{\parz}^T \right)$ is a positive semi-definite matrix.
The vector $\parz$ is an eigenvector of the matrix associated with eigenvalue $0$, while the hyperplane orthogonal to $\parz$ is the $(d-1)$-dimensional eigenspace of $\nabla^2 \norm{\parz}{2}$ associated with eigenvalue $1/\norm{\parz}{2}$.
\end{lemma}
\textbf{Notation:} We denote by $z[i]$, the $i$-th coordinate of vector $z$.
\begin{proof}
 For clarity, let us denote $\ell_2 (\parz) \triangleq \norm{\parz}{2}$. 
  By Lemma~\ref{lemma:unit-force}, we know that $\nabla \ell_2(\parz) = \unitvector{\parz} = \parz / \ell_2(\parz)$.
  We then have
  \begin{equation}
    \partial_{ij}^2 \ell_2(\parz) 
    = \frac{1}{\ell_2(\parz)^2} \left( \ell_2(\parz) \partial_j \parz[i] - \parz[i] \partial_j \ell_2(\parz) \right)
    = \frac{1}{\ell_2(\parz)} \left( \delta_i^j - \frac{\parz[i]}{\ell_2(\parz)} \frac{\parz[j]}{\ell_2(\parz)} \right),
  \end{equation}
  where $\delta_i^j = 1$ if $i=j$, and $0$ if $i \neq j$.
  Combining all coordinates then yields
  \begin{equation}
    \nabla^2 \ell_2(\parz) 
    = \frac{1}{\norm{\parz}{2}} \left( I - \frac{\parz}{\norm{\parz}{2}} \frac{\parz^T}{\norm{\parz}{2}} \right)
    = \frac{1}{\norm{\parz}{2}} \left( I - \unitvector{\parz} \unitvector{\parz}^T \right).
  \end{equation}
  It is then clear that $\nabla^2\ell_2(\parz) \parz = \frac{1}{\norm{\parz}{2}} \left( \parz - \unitvector{\parz} \unitvector{\parz}^T \parz \right) = 0$.
  Meanwhile, if $\parx \perp \parz$, then $\unitvector{\parz}^T \parx= 0$, which then results in $\nabla^2\ell_2(\parz) \parx = \parx / \norm{\parz}{2}$.
  This proves the lemma.
\end{proof}

Intuitively, the lemma says that the pull of $\parz$ on $0$ does not change if we slightly move $\parz$ along the direction $\parz$.
However, this pull is indeed changed if we move $\parz$ in a direction orthogonal to $\parz$.
Moreover, the further away $\parz$ is from $0$, the weaker is this change in direction.

\subsection{Existence and Uniqueness}
\label{sec:GM_uinique}

In dimension one, the definition of the geometric median coincides with the definition of the median. As a result, the geometric median may not be uniquely defined. Fortunately, in higher dimensions, the uniqueness can be guaranteed, under reasonable assumptions. We first prove a few useful lemmas about the strict convexity of convex and  piecewise strictly convex functions.
\input{strict_convexity}

In what follows, we define the dimension of the tuple $\paramfamily$ of preferred vectors as the dimension of the affine space spanned by these vectors, i.e., $\dim \paramfamily \triangleq \dim \set{\paramsub{\voter} - \paramsub{\voterbis} \st \voter, \voterbis \in [\VOTER]}$. 
We then have the following result.

\begin{proposition}
\label{prop:uniqueness}
  $\parz \mapsto \Loss (\paramfamily{}, \parz)$ is infinitely differentiable for all $\parz \notin \set{\paramsub{\voter} \st \voter \in [\VOTER]}$. 
  Moreover, if $\dim \paramfamily{} \geq 2$, then for all such $\parz$, the Hessian matrix of the sum of distances is positive definite, i.e., $\nabla_\parz^2 \Loss (\paramfamily{}, \parz) \succ 0$. 
  In particular, $\Loss$ is then strictly convex on $\setR^d$, and has a unique minimum.
\end{proposition}

\begin{proof} 
  Define $\ell_2(\parz) \triangleq \norm{\parz}{2} = \sqrt{\parz^T \parz} = \sqrt{\sum_{i \in [d]} \parz_i^2}$. 
  This function is clearly infinitely differentiable for all points $\parz \neq 0$. 
  Since $\Loss(\paramfamily{}, \parz) = \frac{1}{\VOTER} \sum \ell_2(\parz - \paramsub{\voter})$, it is also infinitely differentiable for $\parz \notin \set{\paramsub{\voter} \st \voter \in [\VOTER]}$.
  
  Moreover, by using triangle inequality and absolute homogeneity, we know that, for any $\lambda \in [0,1]$ and any $\paramsub{\voter} \in \setR^d$, we have 
  \begin{align}
      \ell_2 &\left( \left(\lambda \parz + (1-\lambda) \parz'\right) - \paramsub{\voter} \right) 
      = \ell_2 \left( \lambda (\parz - \paramsub{\voter}) + (1-\lambda) (\parz' - \paramsub{\voter}) \right) \\
      &\leq \ell_2(\lambda (\parz - \paramsub{\voter})) + \ell_2((1-\lambda) (\parz' - \paramsub{\voter})) 
      = \lambda \ell_2(\parz - \paramsub{\voter}) + (1-\lambda) \ell_2 (\parz' - \paramsub{\voter}),
  \end{align}
  which proves the convexity of $\parz \mapsto \ell_2(\parz - \paramsub{\voter})$.
  Since the sum of convex functions is convex, we also know that $\parz \mapsto \Loss(\paramfamily{}, \parz)$ is convex too.

  Now, we know that $\dim(\parz, \paramfamily{}) \geq \dim(\paramfamily{}) \geq 2$. 
  Therefore, there exists $\voter, \voterbis \in [\VOTER]$ such that $a \triangleq \parz - \paramsub{\voter}$ and $b \triangleq \parz - \paramsub{\voterbis}$ are not colinear. 
  This implies that $-1 < \unitvector{a}^T \unitvector{b} < 1$.
  By Lemma~\ref{lemma:hessian_Euclidean_norm}, we then have
  \begin{align}
    \nabla_\parz^2 &\Loss (\paramfamily{}, \parz)
    \succeq \frac{1}{\VOTER \norm{a}{2}} \left( I - \unitvector{a} \unitvector{a}^T \right)
    + \frac{1}{\VOTER \norm{b}{2}} \left( I - \unitvector{b} \unitvector{b}^T \right) \\
    &\succeq \frac{1}{\VOTER \max \set{\norm{a}{2}, \norm{b}{2}} } \left( 2 I - \unitvector{a} \unitvector{a}^T - \unitvector{b} \unitvector{b}^T \right) \\
    &\succeq \frac{1}{\VOTER \max \set{\norm{a}{2}, \norm{b}{2}}} \left( 2 I - \frac{1}{2} (\unitvector{a} + \unitvector{b}) (\unitvector{a} + \unitvector{b})^T - \frac{1}{2} (\unitvector{a} - \unitvector{b}) (\unitvector{a} - \unitvector{b})^T \right) \\
    &= \frac{2}{\VOTER \max \set{\norm{a}{2}, \norm{b}{2}}} \left( I - \frac{1+\unitvector{a}^T\unitvector{b}}{2} \frac{(\unitvector{a} + \unitvector{b}) (\unitvector{a} + \unitvector{b})^T}{\norm{\unitvector{a} + \unitvector{b}}{2}^2} 
    - \frac{1-\unitvector{a}^T\unitvector{b}}{2} \frac{(\unitvector{a} - \unitvector{b}) (\unitvector{a} - \unitvector{b})^T}{\norm{\unitvector{a} - \unitvector{b}}{2}^2} \right),
  \end{align}
  where we used $\norm{\unitvector{a} + \unitvector{b}}{2}^2 = 2+ 2\unitvector{a}^T \unitvector{b}$ and $\norm{\unitvector{a} - \unitvector{b}}{2}^2 = 2 - 2\unitvector{a}^T \unitvector{b}$. 
  This last matrix turns out to have eigenvalues equal to $\frac{1-\unitvector{a}^T \unitvector{b}}{\VOTER \max \set{\norm{a}{2}, \norm{b}{2}}}$ in the direction $\unitvector{a} + \unitvector{b}$, 
  $\frac{1+\unitvector{a}^T \unitvector{b}}{\VOTER \max \set{\norm{a}{2}, \norm{b}{2}}}$ in the direction $\unitvector{a} - \unitvector{b}$, 
  and $\frac{1}{\VOTER \max \set{\norm{a}{2}, \norm{b}{2}}}$ in directions orthogonal to $a$ and $b$. 
  Since $-1 < \unitvector{a}^T \unitvector{b} < 1$, all such quantities are strictly positive.
  Thus all eigenvalues of $\nabla_\parz^2 \Loss (\paramfamily{}, \parz)$ are strictly positive.
  This implies that along any segment $(\parx,\pary)$ that contains no $\paramsub{\voter}$, then $\parz \mapsto \Loss (\paramfamily{}, \parz)$ is strictly convex.
  Given that $\parz \mapsto \Loss (\paramfamily{}, \parz)$ is convex everywhere, and that there is only a finite number of points $\paramsub{\voter}$, 
  Lemma \ref{lemma:strict_convexity} applies, and proves the strict convexity of $\parz \mapsto \Loss (\paramfamily{}, \parz)$ everywhere and along all directions.
  Uniqueness follows immediately from this.
  
  To prove the existence of the geometric median, we observe that $\Loss(\parz) \geq \Loss(0)$, for $\parz$ large enough. More precisely, denote $\Delta \triangleq \max \set{\norm{\paramsub{\voter}}{2} \st \voter \in [\VOTER]}$. Then $\Loss(0) \leq \Delta$. 
  Yet if $\norm{\parz}{2} \geq 3\Delta$, then $\norm{\parz - \paramsub{\voter}}{2} \geq \norm{\parz}{2} - \norm{\paramsub{\voter}}{2} \geq 3\Delta - \Delta = 2\Delta$, which implies $\Loss(\parz) \geq 2\Delta$. 
  Thus $\inf_{\parz \in \setR^d} \Loss(\parz) = \inf_{\parz \in \ball(0,\Delta)} \Loss(\parz)$, where $\ball(0, \Delta)$ is a ball centered on $0$, and of radius $\Delta$. By continuity of $\Loss$ and compactness of $\ball(0,\Delta)$, we know that this infimum is reached by some point in $\ball(0,\Delta)$.
\end{proof}

\subsection{Symmetries}

We contrast here the symmetry properties of the average, the geometric median and the coordinate-wise median. 

\begin{proposition}
\label{prop:anonymity}
  Assuming uniqueness, the ordering of voters does not impact the average, the geometric median and the coordinate-wise median of their votes.
  This is known as the \emph{anonymity} property.
\end{proposition}

\begin{proof}
  All three operators can be regarded as minimizing an anonymous function, namely, the sum of square distances, the sum of distances and the sum of $\ell_1$ distances (see Section~\ref{sec:alternative}). All such functions are clearly invariant under re-ordering of voters' labels.
\end{proof}

\begin{proposition}
\label{prop:orthogonal_invariance}
  Assuming uniqueness, the average, the geometric median and the coordinate-wise median are invariant under translation and homothety.
  The average and the geometric median are also invariant under any orthogonal transformation, but, in general, the coordinate-wise median is not.
\end{proposition}

\begin{proof}
  The average and the geometric median can both be regarded as minimizing a function that only depends on Euclidean distances. Since any Euclidean isometry $M$ is distance-preserving, if $\average$ and $\GeometricMedian$ are the average and the geometric median of $\paramfamily$, and if $\tau \in \setR^d$ and $\lambda > 0$, it is clear that $\lambda M \average(\paramfamily) + \tau$ and $\lambda M \GeometricMedian(\paramfamily) + \tau$ is the average and the geometric median of the family $\lambda M \paramfamily + \tau$. 
  
  In Section \ref{sec:alternative}, we show that the coordinate-wise median $\CWMedian(.)$ minimizes a function that depends on $\ell_1$ distances. 
  By the same argument as above, this guarantees that the coordinate-wise median of $\lambda \paramfamily + \tau$ is $\lambda \CWMedian(\paramfamily) + \tau$.
  Now consider the vectors $\paramsub{1} \triangleq (0,0)$, $\paramsub{2} \triangleq (1,2)$ and $\paramsub{3} \triangleq (2,1)$. 
  The coordinate-wise median of these vectors is $\CWMedian(\paramfamily) = (1,1)$.
  Now consider the rotation $R = \frac{\sqrt{2}}{2} \begin{pmatrix} 1 & -1 \\ 1 & 1\end{pmatrix}$ of these vectors around $(0,0)$ by an anti-clockwise eighth of a turn.
  We then obtain the vectors $R \paramsub{1} \triangleq (0,0)$, $R \paramsub{2} \triangleq \frac{\sqrt{2}}{2} (-1,3)$ and $R \paramsub{3} \triangleq \frac{\sqrt{2}}{2} (1,3)$. 
  Thus the coordinate-wise median of $R \paramfamily$ is $\CWMedian(R \paramfamily) = \frac{\sqrt{2}}{2} (0,3)$.
  However, $R \CWMedian(\paramfamily) = \frac{\sqrt{2}}{2} (0,2)$. 
  Thus $\CWMedian(R \paramfamily) \neq R \CWMedian(\paramfamily)$.
\end{proof}

\begin{proposition}
\label{prop:center_of_symmetry}
  Assuming uniqueness, if $\parz$ is a center of symmetry of $\paramfamily$, then it is the average, the geometric median and the coordinate-wise median.
\end{proposition}

\begin{proof}
  We can pair all vectors of $\paramfamily$ different from $\parz$ by their symmetry with respect to $\parz$. For any vote, the pull of each pair on $\parz$ cancels out. Thus the sum of pulls vanishes.
\end{proof}

\begin{proposition}
\label{prop:invariance_linear_transformation}
  The average is invariant under any invertible linear transformation, but, even assuming uniqueness, in general, the geometric median and the coordinate-wise median are not.
\end{proposition}

This proposition might appear to be a weakness of the geometric median.
Note that Section~\ref{sec:skewed_geometric_median} actually leverages this to define the {\it skewed geometric median} and improve strategyproofness.

\begin{proof}
  The average is linear. Thus, for any matrix $M \in \setR^{d \times d}$, we have $\average(M \paramfamily) = M \average(\paramfamily)$.
  Moreover, the case of the coordinate-wise median follows from Proposition \ref{prop:orthogonal_invariance}.
  
  To see that the geometric median is not invariant under invertible linear transformation, consider $\paramsub{1} \triangleq (1,0)$, $\paramsub{2} \triangleq (\cos(\tau/3), \sin(\tau/3)) = (-1/2, \sqrt{3}/2)$ and $\paramsub{3} \triangleq (\cos(2\tau/3), \sin(2\tau/3)) = (-1/2, -\sqrt{3}/2)$, where $\tau \approx 6.28$ corresponds to a full turn angle. 
  Then $\GeometricMedian(\paramfamily) = 0$, since the sum of pulls at $0$ cancel out.
  Now let us stretch space in the $y$-axis, using the matrix $M = \begin{pmatrix} 1 & 0 \\ 0 & 2/\sqrt{3} \end{pmatrix}$. Clearly $0$ is invariant under this stretch, as $M0 = 0$. 
  Moreover, we have $M\paramsub{1} = (1,0)$, $M\paramsub{2} = (-1/2, 1)$ and $M\paramsub{3} = (-1/2,-1)$. 
  The unit-force pull on $0$ by voter 2 is then $M \paramsub{2} / \norm{M \paramsub{2}}{2} = 2/\sqrt{5} (-1/2,1)$, while that of voter 3 is $M \paramsub{3} / \norm{M \paramsub{3}}{2} = 2/\sqrt{5} (-1/2,-1)$. Finally, voter 1 still pulls with a unit force towards the right. The sum of forces along the horizontal axis is then equal to $1-2/\sqrt{5} >0$. Thus despite being invariant, $0$ is no longer the geometric median.
  
  The case of the coordinate-wise median follows from Proposition~\ref{prop:orthogonal_invariance}.
\end{proof}

\begin{proposition}
\label{prop:convex_hull}
  The average and the geometric median of a tuple of vectors belong to the convex hull of the vectors. In general, the coordinate-wise median does not.
\end{proposition}

\begin{proof}
  Consider $\parz$ not in the convex hull. Then there must exist a separating hyperplane, with a normal vector $h$, which goes from the convex hull to $\parz$. But then all vectors pull $\parz$ in the direction of $-h$. The projection of the sum of forces on $h$ thus cannot be nil, which shows that $\parz$ cannot be an equilibrium.
  
  Now, to show that the coordinate-wise median may not lie within the convex hull of voters' vote, consider $\paramsub{1} = (1,0,0)$, $\paramsub{2} = (0,1,0)$ and $\paramsub{3} = (0,0,1)$. 
  Then the coordinate-wise median is $(0,0,0)$. This clearly does not belong to the convex hull of $\paramfamily{}$.
\end{proof}

\begin{proposition}
\label{prop:continuity}
  The geometric median is continuous on all points of $\paramfamily$, if $\dim \paramfamily \geq 2$. 
\end{proposition}

\begin{proof}
  Consider $\paramfamily \in \setR^{d\times{\VOTER}}$ with $\dim \paramfamily \geq 2$.
  By Proposition \ref{prop:uniqueness}, there is a unique geometric median $\geometricmedian \triangleq \GeometricMedian(\paramfamily)$.
  
  To prove the continuity of $\GeometricMedian$, let us consider a sequence of families $\paramfamily^{(n)}$ such that $\paramfamily^{(n)} \rightarrow \paramfamily$, 
  and let us prove that this family eventually has a unique geometric median $\geometricmedian^{(n)}$, which converges to $\geometricmedian$ as $n \rightarrow \infty$.
  
  First note that the set of families $\parxfamily \in \setR^{d\times{\VOTER}}$ for which $\dim \parxfamily \leq 1$ is isomorphic to the set of matrices of $\setR^{d \times \VOTER}$ of rank at most 1. 
  It is well-known that this set is closed for all norms in $\setR^{d \times \VOTER}$ (this can be verified by considering the determinants of all $2 \times 2$ submatrices, which are all continuous functions).
  Thus the set of families $\parxfamily$ such that $\dim \parxfamily \geq 2$ is open.
  In particular, there is a ball centered on $\paramfamily$ whose points $\parxfamily$ all satisfy $\dim \parxfamily \geq 2$.
  Since, for $n \geq N_0$ large enough, $\paramfamily^{(n)}$ must belong to this ball, it must eventually satisfy $\dim \paramfamily^{(n)} \geq 2$.
  This guarantees the uniqueness of $\geometricmedian^{(n)} \triangleq \GeometricMedian(\paramfamily^{(n)})$ for $n \geq N_0$.

  Now consider any convergent subsequence $\geometricmedian^{(n_k)} \rightarrow g^*$.
  Since the geometric median minimizes the loss $\Loss$, for any $n_k\in\setN$, we know that
  $\Loss(\geometricmedian^{(n_k)}, \paramfamily^{(n_k)}) \leq \Loss(\geometricmedian, \paramfamily^{(n_k)})$.
  Taking the limit then yields $\Loss(\geometricmedian^{*}, \paramfamily) \leq \Loss(\geometricmedian, \paramfamily)$.
  Since $\geometricmedian$ is the geometric median of $\paramfamily$, we thus actually have $\Loss(\geometricmedian^{*}, \paramfamily) = \Loss(\geometricmedian, \paramfamily)$. 
  But Proposition \ref{prop:uniqueness} guarantees the uniqueness of the geometric median. 
  Therefore, we actually have $\geometricmedian^* = \geometricmedian$. 
  Put differently, any convergent subsequence of $\geometricmedian^{(n)}$ converges to $\geometricmedian$. 
    
  Now by contradiction, assume $\geometricmedian^{(n)}$ does not converge to $\geometricmedian$. 
  Thus, for any $\varepsilon > 0$, there is an infinite subsequence $\geometricmedian^{(n_i)}$ of $g^{(n)}$ lies outside the open ball $\ball(\geometricmedian,\varepsilon)$. 
  But since the geometric median belongs to the convex hull of the vectors (Proposition \ref{prop:convex_hull}), for $n \geq N_0$, $\geometricmedian^{(n_i)}$ is clearly also bounded. 
  Thus, by the Bolzano-Weierstrass theorem, the subsequence $\geometricmedian^{(n_i)}$ must have at least one converging subsequence, whose limit $\geometricmedian^\dagger$ lies outside the open ball $\ball(\geometricmedian,\varepsilon)$.
  But this contradicts the fact that every convergent subsequence of $\geometricmedian^{(n)}$ converges to $\geometricmedian$. 
  Therefore, $\geometricmedian^{(n)}$ must converge to $\geometricmedian$. 
  This proves that the geometric median is continuous with respect to $\paramfamily$. 
\end{proof}

\subsection{Approximation of the Average}

One interesting feature of the geometric median and of the coordinate-wise median is that they are provably a good approximation of the average. 
Note that the uniqueness of the geometric median or of the coordinate-wise median is not needed for the following well-known proposition.

\begin{proposition}[\cite{STANISLAV15}]
Denote by $\Sigma(\paramfamily)$ the covariance matrix of $\paramfamily$ defined by 
\begin{equation}
  \Sigma_{ij}(\paramfamily) \triangleq \frac{1}{\VOTER} \sum_{\voter \in [\VOTER]} (\paramsub{\voter}[i] - \average(\paramfamily)[i]) (\paramsub{\voter}[j] - \average(\paramfamily)[j]).
\end{equation} 
Then $\norm{\average(\paramfamily) - \GeometricMedian(\paramfamily)}{2} \leq \sqrt{\trace{\Sigma(\paramfamily)}}$ and $\norm{\average(\paramfamily) - \CWMedian(\paramfamily)}{2} \leq \sqrt{\trace{\Sigma(\paramfamily)}}$.
\end{proposition}

\begin{proof}
We start with the geometric median. Recall that $\GeometricMedian(\paramfamily)$ minimizes $\parz \mapsto \expectVariable{\voter}{\norm{\paramsub{\voter} - \parz}{2}}$, where $\voter$ is drawn uniformly randomly from $[\VOTER$].
It thus does better to minimize this term than $\average(\paramfamily)$. 
We then have 
\begin{align}
  \norm{\average(\paramfamily) - \GeometricMedian(\paramfamily)}{2}
  &= \norm{\expectVariable{\voter}{\paramsub{\voter}} - \GeometricMedian(\paramfamily{})}{2} 
  \leq \expectVariable{\voter}{\norm{\paramsub{\voter} - \GeometricMedian(\paramfamily{})}{2}} \\
  &\leq \expectVariable{\voter}{\norm{\paramsub{\voter} - \average(\paramfamily)}{2}} 
  \leq \sqrt{\expectVariable{\voter}{\norm{\paramsub{\voter} - \average(\paramfamily)}{2}^2}}
  = \sqrt{\trace{\Sigma(\paramfamily)}},
\end{align}
where we also used Jensen's inequality twice for the function $x \mapsto \norm{x}{2}$ and $t \mapsto t^2$.

We now address the case of the coordinate-wise median. 
On dimension $i$, using similar arguments as in the proof above, this square of the discrepancy can be upper-bounded by the variance of $\param$ along dimension $i$.
In other words, we have $\absv{\average(\paramfamily)[i] - \CWMedian(\paramfamily)[i]} \leq \sqrt{\Sigma_{ii}(\paramfamily{})}$.
Squaring this inequality, and summing over all coordinates then yields $\norm{\average(\paramfamily) - \CWMedian(\paramfamily)}{2}^2 \leq \sum \Sigma_{ii}(\paramfamily{}) = \trace{\Sigma(\paramfamily)}$.
Taking the square root yields the second inequality of the proposition.
\end{proof}

%% file: strict_convexity.tex
\label{app:strict_convexity}

\begin{lemma}
\label{lemma:strict_convexity_closed_interval}
  If $f$ is convex on $[0,1]$, and strictly convex on $(0,1)$, then it is strictly convex on $[0,1]$.
\end{lemma}
\begin{proof}
  Consider $\parx, \pary \in [0,1]$, with $\parx < \pary$, $\lambda \in (0,1)$ and $\mu \triangleq 1-\lambda$.
  Denote $\parz = \lambda \parx + \mu \pary$. 
  It is straightforward to verify that $\parz \in (0,1)$.
  Define $\parx' \triangleq \frac{\parx + \parz}{2}$ and $\pary' \triangleq \frac{\parz + \pary}{2}$.
  Clearly, we have $\parx', \pary' \in (0,1)$.
  Moreover, $\lambda \parx' + \mu \pary' = \frac{1}{2} (\lambda \parx + \mu \pary) + \frac{1}{2} (\lambda + \mu) \parz = \parz$. 
  By strict convexity of $f$ in $(0,1)$, we then have $f(\parz) < \lambda f(\parx') + \mu f(\pary')$.
  Moreover, by convexity of $f$ in $[0,1]$, we also have $f(\parx') \leq \frac{1}{2} f(\parx) + \frac{1}{2} f(\parz)$ and $f(\pary') \leq \frac{1}{2} f(\pary) + \frac{1}{2} f(\parz)$.
  Combining the three inequalities yields $f(\parz) < \frac{1}{2} \left( \lambda f(\parx) + \mu f(\pary) \right) + \frac{1}{2} f(\parz)$, 
  from which we derive $f(\parz) < \lambda f(\parx) + \mu f(\pary)$.
  This allows to conclude.
\end{proof}

\begin{lemma}
\label{lemma:strict_convexity_singularity}
  If $f : [0,1] \rightarrow \setR$ is convex, and if there is $\parw \in (0,1)$ such that $f$ is strictly convex on $(0,\parw)$ and strictly convex on $(\parw,1)$.
  Then, for any $\parx < \parw < \pary$, we have $f(\parw) < \frac{\pary - \parw}{\pary - \parx} f(\parx) + \frac{\parw - \parx}{\pary - \parx} f(\pary)$.
\end{lemma}

\begin{proof}
  Define $\parx' \triangleq \frac{\parx + \parw}{2}$ and $\pary' \triangleq \frac{\parw + \pary}{2}$.
  Since $f$ is strictly convex on $(0,\parw)$, by Lemma \ref{lemma:strict_convexity_closed_interval}, 
  we know that it is strictly convex on $[0,\parw]$.
  As a result, we have $f(\parx') < \frac{1}{2} f(\parx) + \frac{1}{2} f(\parw)$.
  Similarly, we show that $f(\pary') < \frac{1}{2} f(\pary) + \frac{1}{2} f(\parw)$.
  Note now that $\frac{\pary - \parw}{\pary - \parx} \parx' + \frac{\parw - \parx}{\pary - \parx} \pary' 
  = \frac{1}{2} \frac{(\pary - \parw) \parx + (\parw - \parx) \pary}{\pary - \parx}  + \frac{\parw}{2} = \parw$.
  Using the convexity of $f$ over $[0,1]$, we then have 
  $f(\parw) \leq \frac{\pary - \parw}{\pary - \parx} f(\parx') + \frac{\parw - \parx}{\pary - \parx} f(\pary') 
  < \frac{\pary - \parw}{\pary - \parx} \left( \frac{1}{2} f(\parx) + \frac{1}{2} f(\parw) \right) + \frac{\parw - \parx}{\pary - \parx} \left( \frac{1}{2} f(\pary) + \frac{1}{2} f(\parw) \right) 
  = \frac{1}{2} \left( \frac{\pary - \parw}{\pary - \parx} f(\parx) + \frac{\parw - \parx}{\pary - \parx} f(\pary) \right) + \frac{1}{2} f(\parw)$. 
  Rearranging the terms yields the lemma.
\end{proof}

\begin{lemma}
\label{lemma:dich_convexity}
  If $f : [0,1] \rightarrow \setR$ is convex, and if there is $\parw \in [0,1]$ such that $f$ is strictly convex on $(0,\parw)$ and strictly convex on $(\parw,1)$, then $f$ is strictly convex on $[0,1]$.
\end{lemma}

\begin{proof}
  Consider $\parx, \parz, \pary \in [0,1]$, with $\parx < \parz < \pary$. 
  We denote $\lambda \triangleq \frac{\parz-\parx}{\pary-\parx} \in (0,1)$ and $\mu \triangleq 1-\lambda$.
  We then have $\parz = \lambda \parx + \mu \pary$.   
  If $\parx \geq \parw$ or $\pary \leq \parw$, then by Lemma \ref{lemma:strict_convexity_closed_interval}, we know that $f(\parz) < \lambda f(\parx) + \mu f(\pary)$.
  Moreover, Lemma \ref{lemma:strict_convexity_singularity} yields the same equation for the case $\parx < \parz = \parw < \pary$.
  
  Now assume $\parx < \parz < \parw < \pary$. 
  By Lemma \ref{lemma:strict_convexity_singularity}, we have
  $f(\parw) < \frac{\pary-\parw}{\pary-\parx} f(\parx) + \frac{\parw-\parx}{\pary-\parx}  f(\pary)$.
  We also know that $\parz = \frac{\parw - \parz}{\parw-\parx} \parx + \frac{\parz - \parx}{\parw-\parx} \parw$.
  By strict convexity, we thus have 
  $f(\parz) < \frac{\parw - \parz}{\parw-\parx} f(\parx) + \frac{\parz - \parx}{\parw-\parx} f(\parw) 
  < \frac{\parw - \parz}{\parw-\parx} f(\parx) + \frac{\parz - \parx}{\parw-\parx} \left( \frac{\pary-\parw}{\pary-\parx} f(\parx) + \frac{\parw-\parx}{\pary-\parx} f(\pary) \right)
  = \lambda f(\parx) + \mu f(\pary)$.
  
  The last case $\parx < \parw < \parz < \pary$ is dealt with similarly. 
\end{proof}

\begin{lemma}
\label{lemma:strict_convexity}
  Assume that $f : [0,1] \rightarrow \setR$ is convex, and that there is a finite number of points $\parwsub{0} \triangleq 0 <\parwsub{1} < \ldots < \parwsub{K-1} < \parwsub{K} \triangleq 1$ such that $f$ is strictly convex on $(\parwsub{k-1}, \parwsub{k})$ for $k \in [K]$.
  Then $f$ is strictly convex on $[0,1]$.
\end{lemma}

\begin{proof}
  We prove this result by induction on $K$.
  For $K=1$, we simply invoke Lemma \ref{lemma:strict_convexity_closed_interval}.
  Now assume that it holds for $K-1$, and let us use this to derive it for $K$.
  By induction, we know that $f$ is strictly convex on $(0, \parwsub{K-1})$ (we can use the induction hypothesis more rigorously by defining $g(\parx) \triangleq f(\parx \parwsub{K-1})$).
  Yet, by assumption, $f$ is also known to be convex on $(\parwsub{K-1},1)$.
  Lemma \ref{lemma:strict_convexity_singularity} thus applies, and implies the strict convexity of $f$ on $[0,1]$.
\end{proof}

%% file: APP_nonstrategyproo.tex
\section{PROOFS OF SECTION \ref{sec:non_strategyproofness}}
\label{app:non_strategyproof}

\begin{figure}%
    \centering
    \includegraphics[width=.39\textwidth]{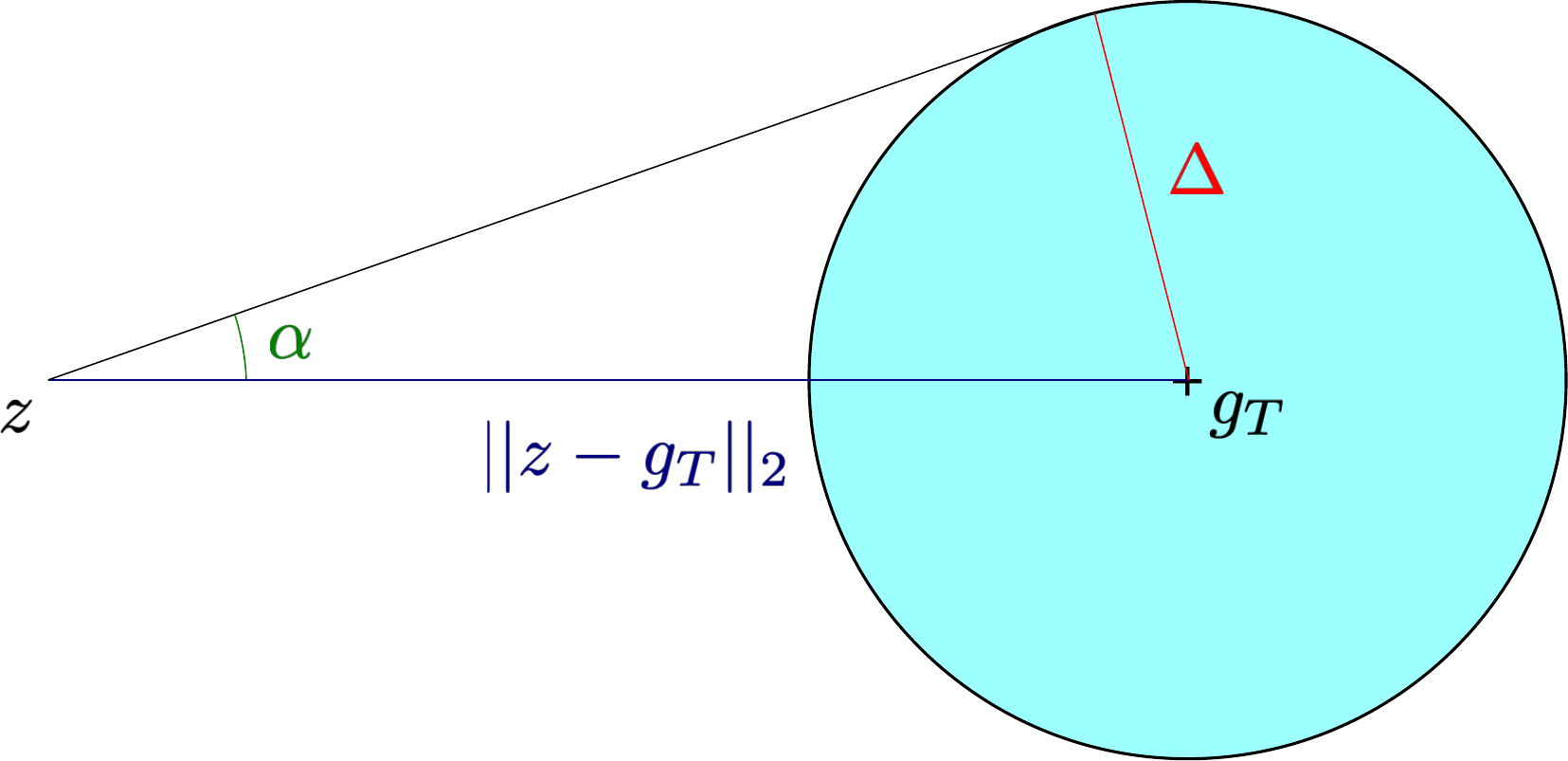}
    \caption{Resilience of the geometric median against coordinated attacks by a minority of strategic voters $\strategists$, who pull on $\parz$ in the opposite direction from a strict majority of truthful voters $\truthful$, whose vectors are all in the ball centered on $\geometricmedian_\truthful$ and of radius $\Delta$.}
    \label{fig:byzantine}
\end{figure}

\subsection{Proof of Proposition \ref{th:byzantine}}
\label{sec:proof_nonmanipulabe}
\begin{proof}
  Let us denote $[\VOTER] = \truthful \cup \strategists$ a decomposition of the voters into two disjoint subsets of truthful and strategic voters. We assume a strict majority of truthful voters, i.e., $\card{\truthful} > \card{\strategists}$. 
  Denote $\geometricmedian_\truthful \triangleq \GeometricMedian(\paramfamily_\truthful)$ the geometric median of truthful voters' preferred vectors, and $\Delta \triangleq \max \set{\norm{\paramsub{t} - \geometricmedian_\truthful}{2} \st t \in \truthful}$ the maximum distance between a truthful voters' preferred vector and the geometric median $\geometricmedian_\truthful$.
  
  Now consider any point $\parz \notin \ball(\geometricmedian_\truthful, \Delta)$. The sum of forces on $\parz$ by truthful voters has a norm equal to
  \begin{align}
    \norm{\sum_{t \in \truthful} \unitvector{\paramsub{t} - \parz}}{2} 
    &\geq \left(\sum_{t \in \truthful} \unitvector{\paramsub{t} - \parz} \right)^T \unitvector{\geometricmedian_\truthful - \parz} 
    = \sum_{t \in \truthful} \unitvector{\paramsub{t} - \parz}^T \unitvector{\geometricmedian_\truthful - \parz} \\
    &\geq \card{\truthful} \cos \alpha 
    = \card{\truthful} \sqrt{1-\sin^2 \alpha} 
    = \card{\truthful} \sqrt{1 - \frac{\Delta^2}{\norm{\parz - \geometricmedian_\truthful}{2}^2}},
  \end{align}
  where $\alpha$ is defined in Figure \ref{fig:byzantine} as the angle between $\geometricmedian_\truthful - \parz$ and a tangent to $\ball(\geometricmedian_\truthful)$ that goes through $\parz$. But then the sum of all forces at $\parz$ must be at least
  \begin{align}
    &\norm{\sum_{t \in \truthful} \unitvector{\paramsub{t} - \parz} + \sum_{\voter \in \strategists} \unitvector{\strategicvote{\voter} - \parz}}{2}
    \geq \norm{\sum_{t \in \truthful} \unitvector{\paramsub{t} - \parz}}{2} - \norm{\sum_{\voter \in \strategists} \unitvector{\strategicvote{\voter} - \parz}}{2} \\
    &\qquad \geq \card{\truthful} \sqrt{1 - \frac{\Delta^2}{\norm{\parz}{2}^2}} - \sum_{\voter \in \strategists} \norm{\unitvector{\strategicvote{\voter} - \parz}}{2}
    = \card{\truthful} \sqrt{1 - \frac{\Delta^2}{\norm{\parz - \geometricmedian_\truthful}{2}^2}} - \card{\strategists} > 0,
  \end{align}
  as long as we have $\norm{\parz - \geometricmedian_\truthful}{2} > \frac{\card{\truthful} \Delta}{ \sqrt{\card{\truthful}^2 - \card{\strategists}^2}}$. A value of $\parz$ that satisfies this strict inequality can thus not be a geometric median. Put differently, no matter what strategic voters do, we have 
  \begin{equation}
    \GeometricMedian(\paramfamily_{\truthful}, \strategicvote{\strategists}) \in 
    \ball\left( \GeometricMedian(\paramfamily_\truthful), \left( 1- \frac{\card{\strategists}^2}{\card{\truthful}^2} \right)^{-1/2} \max_{t \in \truthful} \norm{\paramsub{t} - \GeometricMedian(\paramfamily_\truthful)}{2} \right).
  \end{equation}
  This concludes the proof.
\end{proof}

\subsection{Proof of Theorem \ref{th:geometric_median_not_strategyproof}}
\label{sec:proof_th_geometric_median_not_strategyproof}

To obtain Theorem~\ref{th:geometric_median_not_strategyproof}, we make use of a technical lemma that characterizes the achievable set for the the strategic voter.
Consider the set 
    $\setAchieve \triangleq \set{\parz \in \setR^d \st \exists \subgradient \in \nabla_\parz \Loss_{}(\paramfamily,\parz) \mathsep \norm{\subgradient}{2}\leq 1/\VOTER }$,
of points $\parz$ where the loss restricted to other voters $\voter \in [\VOTER]$ has a subgradient of norm at most $1/\VOTER$.
We now observe that, by behaving strategically, voter $0$ can choose any value for the geometric median within $\setAchieve$.

\begin{lemma}
\label{lemma:achievable_set}
  For any $\strategicvote{0} \in \setR^d$, 
  $\GeometricMedian (\strategicvote{0}, \paramfamily) \in \setAchieve$.
  Moreover, for $\dim \paramfamily \geq 2$ and $\strategicvote{0} \in \setAchieve$,  we have
  $\GeometricMedian (\strategicvote{0}, \paramfamily) = \strategicvote{0}$.
\end{lemma}
\begin{proof}
   Define $\ell_2(\parz) \triangleq \norm{\parz}{2}$. Now note that $(1+\VOTER) \nabla_\parz \Loss_{} \left( \strategicvote{0}, \paramfamily,\parz \right) = \nabla_\parz \ell_2 (\parz - \strategicvote{0}) + \VOTER \nabla_\parz \Loss_{} (\paramfamily, \parz)$.
  In other words, for any subgradient $\subgradient_{0:\VOTER} \in \nabla \Loss_{} \left( \strategicvote{0}, \paramfamily,\parz \right)$, 
  there exists $\subgradient_{0} \in \nabla \ell_2(\parz - \strategicvote{0})$ and $\subgradient_{1:\VOTER} \in \nabla \Loss_{} \left(\paramfamily,\parz \right) $
  such that $(1+\VOTER) \subgradient_{0:\VOTER} = \subgradient_{0} + \VOTER \subgradient_{1:\VOTER}$.
  Note that any subgradient of $\ell_2$ has at most a unit $\ell_2$-norm (Lemma \ref{lemma:unit-force}).
  Thus, $\norm{\subgradient_{0}}{2} \leq 1$.

  Now, assume $\parz \notin \setAchieve$. 
  Then for any $\subgradient_{1:\VOTER} \in \nabla \Loss_{} \left(\paramfamily,\parz \right) $, we must have $\norm{\subgradient_{1:\VOTER}}{2} > 1 / \VOTER$. 
  As a result, 
  \begin{align}
    (1+\VOTER) \norm{\subgradient_{0:\VOTER}}{2} 
    \geq \norm{ \VOTER \subgradient_{1:\VOTER} + \subgradient_{0} }{2}
    \geq \VOTER \norm{\subgradient_{1:\VOTER}}{2} - 1 > 0.
  \end{align}
  Thus, $0 \notin \nabla \Loss_{} \left( \strategicvote{0}, \paramfamily,\parz \right)$, which means that $\parz$ cannot be a geometric median.
  For any $\strategicvote{0} \in \setR^d$, we thus necessarily have $\GeometricMedian (\strategicvote{0}, \paramfamily) \in \setAchieve$.

  Now assume that $\strategicvote{0} \in \setAchieve$.
  Then there must exist $\subgradient_{1:\VOTER} \in \nabla \Loss_{} \left(\paramfamily,\strategicvote{0} \right) $ such that $\VOTER \norm{\subgradient_{1:\VOTER}}{2} \leq 1$. 
  Thus $\subgradient_{0} \triangleq - \VOTER \subgradient_{1:\VOTER} \in \nabla \ell_2 (\parz - \strategicvote{0})$, for $\parz = \strategicvote{0}$, since the set of subgradients of $\ell_2$ at $0$ is the unit closed ball.
  We then have $\subgradient_{0} + \VOTER \subgradient_{1:\VOTER} = 0 \in \nabla \Loss_{} \left( \strategicvote{0}, \paramfamily,\strategicvote{0} \right)$.
  Thus $\strategicvote{0}$ minimizes $\Loss_{} \left( \strategicvote{0}, \paramfamily,\cdot \right)$.
  The uniqueness of the geometric median for $\dim (\strategicvote{0}, \paramfamily) \geq \dim \paramfamily \geq 2$ (Proposition~\ref{prop:uniqueness}) then implies that $\GeometricMedian (\strategicvote{0}, \paramfamily) = \strategicvote{0}$.
\end{proof}
We now provide the detailed proof of Theorem \ref{th:geometric_median_not_strategyproof} by formalizing the example of Figure \ref{fig:ellipsoid}.

\input{proof_not_strategyproof}

%% file: proof_not_strategyproof.tex
\begin{proof}[Proof of Theorem \ref{th:geometric_median_not_strategyproof}]
  Define $\paramsub{1} = (-X,-1)$, $\paramsub{2} = (-X,1)$, $\paramsub{3} = (X,-1)$ and $\paramsub{4} = (X,1)$, with $X \geq 8$. We define the sum of distance restricted to these four inputs as
  \begin{equation}
    \Loss_0 (\parz) \triangleq \frac{1}{4}\sum_{\voter = 1}^4 \norm{\paramsub{\voter} - \parz}{2}.
  \end{equation}
  Since $0$ is a center of symmetry of the four inputs, it is the geometric median. Moreover, it can then be shown that the Hessian matrix at this optimum is
  \begin{equation}
    \hessian \triangleq \nabla_\parz^2 \Loss_0 (0) 
    = \frac{1}{4} (1+X^2)^{-3/2} 
    \begin{pmatrix}
      1 & 0 \\
      0 & X^2
    \end{pmatrix}.
  \end{equation}
  Note that the ratio between the largest and smallest eigenvalues of this Hessian matrix $\hessian$ is $X^2$, which can take arbitrarily large values. This observation turns out to be at the core of our proof. The eigenvalues also yield bounds on the norm of a vector to which $\hessian$ was applied. Using the inequality $X \geq 1$, 
  \begin{equation}
  \label{eq:H_inequality}
    \frac{1}{32 X^3} \norm{\parz}{2} \leq \norm{\hessian \parz}{2} \leq \frac{1}{4X} \norm{\parz}{2} \leq \norm{\parz}{2}.
  \end{equation}
  In the vicinity of $0$, since $\nabla_\parz \Loss_0 (0) = 0$ and since $\Loss_0$ is infinitely differentiable in $0$, we then have
  \begin{align}
    \nabla_\parz \Loss_0 (\parz)
    &= \hessian \parz + \varepsilon(\parz), 
  \end{align}
  where $\norm{\varepsilon(\parz)}{2} = O(\norm{\parz}{2}^2)$ when $\parz \rightarrow 0$. 
  In fact, for $X \geq 1$, we know that there exists $A$ such that, for all $\parz \in \ball(0,1)$, where $\ball(0,1)$ is the unit Euclidean ball centered on $0$, we have $\norm{\varepsilon(\parz)}{2} \leq A\norm{\parz}{2}^2$. 
  We also define
  \begin{equation}
    \lambda \triangleq \inf_{\parz \in \ball(0,1)} \min \spectrum(\nabla_\parz^2 \Loss_0(\parz)) \quad\text{and}\quad \mu \triangleq \sup_{\parz \in \ball(0,1)} \max \spectrum(\nabla_\parz^2 \Loss_0(\parz))
  \end{equation}
  the minimal and maximal eigenvalues of the Hessian matrix of $\Loss_0$ over the ball $\ball(0,1)$. By continuity (Lemma~\ref{lemma:continuous_min_spectrum}) and strong convexity, we know that $\mu \geq \lambda > 0$.
  We then have $\lambda I \preceq \nabla_\parz^2 \Loss_0(\parz) \preceq \mu I$ over $\ball(0,1)$.  
  From this, it follows that
  \begin{equation}
    \lambda \norm{\parz}{2} \leq \norm{\nabla_\parz \Loss_0(\parz)}{2} \leq \mu \norm{\parz}{2},
  \end{equation}
  for all $\parz \in \ball(0,1)$. Now, since $\nabla_\parz^2 \Loss_0(\parz) \succeq 0$ for all $\parz \in \setR^d$, from this we also deduce that $\norm{\nabla_\parz \Loss_0(\parz)}{2} \geq \lambda$ if $\parz \notin \ball(0,1)$.
  
  Now consider $\VOTER$ honest voters such that $\VOTER/4 \in \setN$ and $\paramsub{4k+j} = \paramsub{j}$, for $j \in [4]$ and $k \in [\VOTER/4-1]$. We denote by $\paramfamily_\VOTER$ this vector family. For any voter $0$'s strategic vote $\strategicvote{0}$, we then have
  \begin{equation}
    (1+\VOTER) \Loss (\strategicvote{0}, \paramfamily_\VOTER, \parz) = \norm{\strategicvote{0} - \parz}{2} + \VOTER \Loss_0(\parz).
  \end{equation}
  Note that we then have 
  \begin{equation}
   (1+\VOTER) \nabla_\parz \Loss (\strategicvote{0}, \paramfamily_\VOTER, \parz) = \unitvector{\parz-\strategicvote{0}} + \VOTER \nabla \Loss_0 (\parz),
   \label{eq:not_strategyproof_gradient}
  \end{equation}
  where $\unitvector{\parx} \triangleq \frac{\parx}{\norm{\parx}{2}}$ is the unit vector in the same direction as $\parx$. For all $\parz \notin \ball(0,1)$, we then have 
  \begin{equation}
    \norm{\nabla_\parz \Loss (\strategicvote{0}, \paramfamily_\VOTER, \parz)}{2} \geq \frac{\VOTER \norm{\nabla \Loss_0 (\parz)}{2} - \norm{\unitvector{\parz-\strategicvote{0}}}{2}}{1+\VOTER} \geq \frac{\VOTER \lambda - 1}{1+\VOTER} > 0,
  \end{equation}
  for $\VOTER > 1/\lambda$.
  Thus, for $\VOTER > 1/\lambda$, we know that, for any $\strategicvote{0}$, we have $\GeometricMedian(\strategicvote{0}, \paramfamily_\VOTER) \in \ball(0,1)$, where the inequality $\norm{\varepsilon(\parz)}{2} \leq A\norm{\parz}{2}^2$ holds.
  
  Since $\Loss_0$ is strictly convex, there exists a unique $\alpha_\VOTER >0$ such that $\norm{\nabla_\parz \Loss_0 (\alpha_\VOTER (X^3,1))}{2} =1/\VOTER$. Denote $\geometricmedian_\VOTER \triangleq \alpha_\VOTER (X^3,1)$. 
  Now define 
  \begin{equation}
    \targetvector = \targetvector (\VOTER)\triangleq \geometricmedian_\VOTER + \frac{1}{\sqrt{\VOTER}} \nabla_\parz \Loss_0 (\geometricmedian_\VOTER).
  \end{equation}
  The force of $\targetvector$ on $\geometricmedian_\VOTER$ is then the unit force with direction $\targetvector - \geometricmedian_\VOTER = \frac{1}{\sqrt{\VOTER}} \nabla_\parz \Loss_0(\geometricmedian_\VOTER)$.
  Since $\norm{\nabla_\parz \Loss_0(\geometricmedian_\VOTER)}{2} = 1/\VOTER$, this unit vector must be $\VOTER \nabla_\parz \Loss_0(\geometricmedian_\VOTER)$. 
  Plugging this into the gradient of $\Loss$ (Equation~(\ref{eq:not_strategyproof_gradient})) shows that $\nabla_\parz \Loss (\targetvector, \paramfamily_\VOTER, \geometricmedian_\VOTER) = 0$. 
  Therefore, Lemma \ref{lemma:geometric-median_unit-forces} and the uniqueness of the geometric median (Proposition \ref{prop:uniqueness}) allow us to conclude that $\geometricmedian_\VOTER$ is the geometric median of the true preferred vectors, i.e., $\geometricmedian_\VOTER = \GeometricMedian (\targetvector, \paramfamily_\VOTER)$. 
  Also, we have 
  \begin{equation}
    \norm{\targetvector - \GeometricMedian(\targetvector, \paramfamily_\VOTER)}{2} = \frac{1}{\sqrt{\VOTER}} \norm{\nabla_\parz \Loss_0(\geometricmedian_\VOTER)}{2} = \VOTER^{-3/2}.
  \end{equation}
  Since $\geometricmedian_\VOTER$ is a geometric median of $\targetvector$ and $\paramfamily_\VOTER$, we know that, for $\VOTER > 1/\lambda$, we have $\geometricmedian_\VOTER \in \ball(0,1)$. 
  As a result, we have $\lambda \norm{\geometricmedian_\VOTER}{2} \leq 1/\VOTER = \norm{\nabla_\parz \Loss_0 (\geometricmedian_\VOTER)}{2} \leq \mu \norm{\geometricmedian_\VOTER}{2}$, and thus
  \begin{equation}
  \label{eq:g_bound}
      1/\mu \VOTER \leq \norm{\geometricmedian_\VOTER}{2} \leq 1/\lambda \VOTER.
  \end{equation}
  
  Now, suppose that, instead of reporting $\targetvector$, voter $0$ reports $\strategicvote{0}$, which is approximately the orthogonal projection of $\targetvector$ on the ellipsoid $\set{\parz \st \norm{\hessian \parz}{2} \leq 1/\VOTER}$. More precisely, voter $0$'s strategic vote is defined as
  \begin{align}
    \strategicvote{0} = \strategicvote{0}(\VOTER) &\triangleq \targetvector - \frac{2}{\sqrt{\VOTER}} \frac{\geometricmedian_\VOTER^T \hessian \hessian \hessian \geometricmedian_\VOTER}{\norm{\hessian \hessian \geometricmedian_\VOTER}{2}^2} \hessian \hessian \geometricmedian_\VOTER \\
    &= \geometricmedian_\VOTER + \frac{1}{\sqrt{\VOTER}} \nabla_\parz \Loss_0(\geometricmedian_\VOTER) - \frac{2}{\sqrt{\VOTER}} \frac{\geometricmedian_\VOTER^T \hessian \hessian \hessian \geometricmedian_\VOTER}{\norm{\hessian \hessian \geometricmedian_\VOTER}{2}^2} \hessian \hessian \geometricmedian_\VOTER.
    \label{eq:strategic_vote_def}
  \end{align}
  Given the inequalities $\norm{\hessian \parz}{2} \leq \norm{\parz}{2}$ (Equation (\ref{eq:H_inequality})) and $\norm{\geometricmedian_\VOTER}{2} \leq 1/\lambda \VOTER$, the norm of $\strategicvote{0}$ can be upper-bounded by
  \begin{align}
    \norm{\strategicvote{0}}{2} &\leq \norm{\geometricmedian_\VOTER}{2} + \frac{1}{\sqrt{\VOTER}} \norm{\nabla_\parz \Loss_0(\geometricmedian_\VOTER)}{2} + \frac{2}{\sqrt{\VOTER}} \frac{\norm{\hessian \geometricmedian_\VOTER}{2} \norm{\hessian \hessian \geometricmedian_\VOTER}{2}}{\norm{\hessian \hessian \geometricmedian_\VOTER}{2}} 
    \leq \frac{1 + (2+\lambda) \VOTER^{-1/2}}{\lambda \VOTER}.
  \end{align}
  Assuming $\VOTER \geq 1+3/\lambda$ then implies $\norm{\strategicvote{0}}{2} \leq (3+\lambda)/\lambda \VOTER \leq 1$ and $\norm{\strategicvote{0}}{2} = \mathcal O(1/\VOTER)$.
  As a result, $\norm{\varepsilon(\strategicvote{0})}{2} \leq A \norm{\strategicvote{0}}{2}^2 = \mathcal O(1/\VOTER^2)$. Thus
  \begin{align}
    \norm{\nabla_\parz \Loss_0 (\strategicvote{0})}{2}^2 
    &= \norm{\hessian \strategicvote{0} + \varepsilon(\strategicvote{0})}{2}^2 \\
    &\leq \norm{\hessian \strategicvote{0}}{2}^2 + 2 \norm{\strategicvote{0}}{2} \norm{\varepsilon(\strategicvote{0})}{2} + \norm{\varepsilon(\strategicvote{0})}{2}^2 \\ 
    &\leq \norm{\hessian \strategicvote{0}}{2}^2 + \mathcal O(1/\VOTER^3),
  \end{align}
  where the hidden constant in $\mathcal O(1/\VOTER^3)$ depends on $\lambda$ and $A$. 
  Moreover, given that $\norm{\geometricmedian_\VOTER}{2} = \mathcal{O}(1/\VOTER)$ (Equation (\ref{eq:g_bound})) and $\nabla_\parz \Loss_0(\geometricmedian_\VOTER) = \hessian \geometricmedian_\VOTER + \varepsilon(\geometricmedian_\VOTER)$, by Equation (\ref{eq:strategic_vote_def}), we have
  \begin{align}
    \norm{\hessian \strategicvote{0}}{2}^2 
    &= \norm{\hessian \geometricmedian_\VOTER}{2}^2 + \frac{2}{\sqrt{\VOTER}} (\hessian \geometricmedian_\VOTER)^T \hessian \left( \hessian \geometricmedian_\VOTER + \varepsilon(\geometricmedian_\VOTER) - 2 \frac{\geometricmedian_\VOTER^T \hessian \hessian \hessian \geometricmedian_\VOTER}{\norm{\hessian \hessian \geometricmedian_\VOTER}{2}^2} \hessian \hessian \geometricmedian_\VOTER \right) + \mathcal O(1/\VOTER^3) \\
    &= \norm{\nabla_\parz \Loss_0(\geometricmedian_\VOTER) - \varepsilon(\geometricmedian_\VOTER)}{2}^2 + \frac{2}{\sqrt{\VOTER}} \geometricmedian_\VOTER^T \hessian \hessian \hessian \geometricmedian_\VOTER - \frac{4}{\sqrt{\VOTER}} \geometricmedian_\VOTER^T \hessian \hessian \hessian \geometricmedian_\VOTER + \mathcal O(1/\VOTER^3) \\
    &\leq \norm{\nabla_\parz \Loss_0(\geometricmedian_\VOTER)}{2}^2 - \frac{2}{\sqrt{\VOTER}} \geometricmedian_\VOTER^T \hessian \hessian \hessian \geometricmedian_\VOTER + \mathcal O(1/\VOTER^3) \\
    &\leq \frac{1}{\VOTER^2} - \frac{2}{\sqrt{\VOTER}} \geometricmedian_\VOTER^T \hessian \hessian \hessian \geometricmedian_\VOTER + \mathcal O(1/\VOTER^3).
  \end{align}
  The hidden constants in $\mathcal O(1/\VOTER^3)$ depend on $\lambda$, $A$, $\hessian$ and $X$. 
  Since $\hessian$ has strictly positive eigenvalues and does not depend on $\VOTER$, we know that $\geometricmedian_\VOTER^T \hessian \hessian \hessian \geometricmedian_\VOTER = \Theta(\norm{\geometricmedian_\VOTER}{2}^2) = \Theta(1/\VOTER^2)$. 
  In particular, for $\VOTER$ large enough $\frac{2}{\sqrt{\VOTER}} \geometricmedian_\VOTER^T \hessian \hessian \hessian \geometricmedian_\VOTER = \Theta(1/\VOTER^{2.5})$ takes larger values than $\mathcal O(1/\VOTER^3)$.
  We then have $\norm{\nabla_\parz \Loss_0 (\strategicvote{0})}{2} < 1/\VOTER$, which means that $\strategicvote{0}$ lies inside the achievable set $\setAchieve$.
  Therefore, Lemma \ref{lemma:achievable_set} implies that for $\VOTER$ large enough, by reporting $\strategicvote{0}$ instead of $\targetvector$, voter $0$ can move the geometric median from $\geometricmedian_\VOTER$ to $\strategicvote{0}$, i.e., we have $\GeometricMedian(\strategicvote{0}, \paramfamily_\VOTER) = \strategicvote{0}$.
  But then, the distance between voter $0$'s preferred vector $\targetvector$ and the manipulated geometric median is given by
  \begin{align}
    \norm{\GeometricMedian(\strategicvote{0}, \paramfamily_\VOTER) - \targetvector}{2} 
    = \norm{\frac{2}{\sqrt{\VOTER}} \frac{\geometricmedian_\VOTER^T \hessian \hessian \hessian \geometricmedian_\VOTER}{\norm{\hessian \hessian \geometricmedian_\VOTER}{2}^2} \hessian \hessian \geometricmedian_\VOTER}{2} 
    = \frac{2}{\sqrt{\VOTER}} \frac{(\hessian \geometricmedian_\VOTER)^T (\hessian \hessian \geometricmedian_\VOTER)}{\norm{\hessian \hessian \geometricmedian_\VOTER}{2}}.
  \end{align}
  Now recall that $\geometricmedian_\VOTER = \alpha_\VOTER (X^3, 1)$. Moreover, $\alpha_\VOTER \norm{\hessian (X^3,1)}{2} = \norm{\hessian \geometricmedian_\VOTER}{2} = \norm{\nabla_\parz \Loss_0 (\geometricmedian_\VOTER) - \varepsilon(\geometricmedian_\VOTER)}{2} = 1/\VOTER + \mathcal O(1/\VOTER^2)$. 
  Since $\hessian(X^3,1) = \frac{1}{4}(1+X^2)^{-3/2} (X^3, X^2) = \frac{1}{4} X^2 (1+X^2)^{-3/2} (X,1)$, we have $\norm{\hessian(X^3,1)}{2} =\frac{1}{4} X^2 (1+X^2)^{-1}$. Thus, $\alpha_\VOTER = 4 X^{-2} (1+X^2)/\VOTER + \mathcal O(1/\VOTER^2)$.  
  As a result, we have $\hessian \hessian \geometricmedian_\VOTER = \frac{1}{16}\alpha_\VOTER X^3 (1+X^2)^{-3} (1,X)$.
  The norm of this vector is then $\norm{\hessian \hessian \geometricmedian_\VOTER}{2} = \frac{1}{16} \alpha_\VOTER X^3 (1+X^2)^{-5/2}$. 
  Moreover, its scalar product with $\hessian \geometricmedian_\VOTER$ yields $(\hessian \geometricmedian_\VOTER)^T (\hessian \hessian \geometricmedian_\VOTER) = \frac{1}{32} \alpha_\VOTER^2 X^6 (1+X^2)^{-9/2}$. We thus have 
  \begin{align}
    \norm{\GeometricMedian(\strategicvote{0}, \paramfamily_\VOTER) - \targetvector}{2} 
    &= \frac{ \alpha_\VOTER}{\sqrt{\VOTER}} \frac{X^6 (1+X^2)^{-9/2}}{X^3 (1+X^2)^{-5/2}} \\
    &= \frac{4X}{(1+X^2) \VOTER^{3/2}} + \mathcal O(\VOTER^{-5/2}) \\
    &= \frac{4X}{1+X^2} \norm{\targetvector - \GeometricMedian(\targetvector, \paramfamily_\VOTER)}{2} + \mathcal O(\VOTER^{-5/2}).
  \end{align}
  In particular, for $\VOTER$ large enough, we can then guarantee that
  \begin{align}
    \norm{\targetvector - \GeometricMedian(\targetvector, \paramfamily_\VOTER)}{2} 
    &> \frac{1+X^2}{8X} \norm{\GeometricMedian(\strategicvote{0}, \paramfamily_\VOTER) - \targetvector}{2} \\
    &= \left( 1+\frac{X^2-8X+1}{8X} \right) \norm{\GeometricMedian(\strategicvote{0}, \paramfamily_\VOTER) - \targetvector}{2}.
  \end{align}
  This proves that the geometric median fails to be $\frac{X^2-8X+1}{8X}$-strategyproof. But our proof holds for any value of $X$, and $\frac{X^2-8X+1}{8X} \rightarrow \infty$ as $X \rightarrow \infty$. Thus, there is no value of $\strategyproofbound$ such that the geometric median is $\strategyproofbound$-strategyproof.
\end{proof}

%% file: APP_approximate_strategyproof.tex
\section{PROOFS AND DIFFERENT RESULTS FROM SECTION \ref{sec:asymptotic_strategyproofness}}
\label{sec:proof_asymptotic_strategyproofness}
In this section we provide a formal proof for the main result of our paper which is Theorem~\ref{th:asymptotic_strategyproofness}. We start by proving a few useful facts about the infinite geometric median $\geometricmedian_\infty$ defined on the distribution of the reported vectors.
\subsection{Preliminary Results for the Infinite Limit Case}
\begin{lemma}
\label{lemma:dimension}
Under Assumption \ref{ass:pdf}, with probability 1, we have $\dim \paramfamily_\VOTER = \min \set{\VOTER-1, d}$.
\end{lemma}

\begin{proof}
  We prove this by induction over $\VOTER$. 
  For $\VOTER = 1$, the lemma is obvious.
  
  Assume now that the lemma holds for $\VOTER \leq d$. 
  Then $\dim \paramfamily_\VOTER = \VOTER - 1$ with probability 1.
  The affine space generated by $\paramfamily_\VOTER$ is thus a hyperplane, whose Lebesgue measure is zero.
  Assumption~\ref{ass:pdf} then implies that the probability of drawing a point on this hyperplane is zero.
  In other words, with probability 1, $\paramsub{\VOTER +1}$ does not belong to the hyperplane, which implies that $\dim \paramfamily_{\VOTER + 1} = \dim \paramfamily_\VOTER + 1 = (\VOTER +1) -1 \leq d$, which proves the induction.
  
  Now assume that the lemma holds for $\VOTER \geq d+1$. 
  Then $\dim \paramfamily_\VOTER = d$ with probability 1.
  We then have $\dim \paramfamily_{\VOTER+1} \geq \dim \paramfamily_\VOTER = d$.
  Since this dimension cannot be strictly larger than $d$, we must then have $\dim \paramfamily_{\VOTER+1} = d$.
  This concludes the proof.
\end{proof}
Combing Lemma~\ref{lemma:dimension} with Proposition~\ref{prop:uniqueness} guarantees the uniqueness of the geometric median for $\VOTER \geq 3$ under Assumption~\ref{ass:pdf}.

\begin{lemma}
  \label{lemma: bounded integral}
  If $d \geq k+1$, then $\parx \mapsto \norm{\parx}{2}^{-k}$ is integrable in $\ball(0,1)$, and $\int_{\ball(0, \varepsilon)} \norm{\parx}{2}^{-k} dx = \mathcal O(\varepsilon)$ as $\varepsilon \rightarrow 0$.
\end{lemma}
\begin{proof}
  Consider the hyperspherical coordinates $(r, \varphi_1, \ldots, \varphi_{d-1})$, where $\parx_j = r \left( \prod_{i=1}^{j-1} \cos \varphi_i \right) \sin \varphi_j$. We then have $d \parx = r^{d-1} dr \left( \prod_{i=1}^{d-1} \cos^{d-i-1} \varphi_i d\varphi_i \right)$. The integral becomes
  \begin{equation}
    \int_{\ball(0,1)} \norm{\parx}{2}^{-k} = C(d) \int_{0}^1 r^{-k} r^{d-1} dr = C(d) \int_0^1 r^{d-1-k} dr,
  \end{equation}
  where $C(d)$ is obtained by integrating appropriately all the angles of the hyperspherical coordinates, which are clearly integrable. But $\int_0^1 r^{d-1-k} dr$ is also integrable when $d-1-k \geq 0$. We conclude by noting that we then have $\int_0^\varepsilon r^{d-1-k} dr \propto \varepsilon^{d-k} = \mathcal O(\varepsilon)$ for $d-k \geq 1$.
\end{proof}

\label{sec:proof_positive_difinite}
\begin{proposition}
\label{porp:positive difinite}
  Under Assumption \ref{ass:pdf}, $\Loss_\infty$ is five-times continuously differentiable with a strictly positive definite Hessian matrix on $\PARAM$. 
  As a corollary, the geometric median $\geometricmedian_\infty$ is unique and lies in $\PARAM$.
\end{proposition}

\begin{proof}
  Let $\parz \in \PARAM$ and $\delta >0$ such that $\ball(\parz, \delta) \subset \PARAM$. By Leibniz's integral rule, we obtain
  \begin{equation}
    \nabla \Loss_\infty (\parz)
    = \int_\PARAM \nabla_\parz \norm{\parz - \param}{2} p(\param) d\param
    = \int_\PARAM \unitvector{\parz-\param} p(\param) d\param.
  \end{equation}
  To deal with the singularity at $\param = \parz$, 
  we first isolate the integral in the ball $\ball(\parz, \varepsilon)$, for some $0 < \varepsilon \leq \delta$. 
  On this compact set, $p$ is continuous and thus upper-bounded. We can then apply the previous lemma for $k=0$ to show that this singularity is negligible as $\varepsilon \rightarrow 0$.
  Moreover, Leibniz's integral rule does apply, since $\unitvector{\parz-\param} p(\param)$ can be upper-bounded by $p(\param)$ outside of $\ball(\parz,\delta)$, which is integrable by Assumption \ref{ass:pdf}.
  This shows that $\Loss_\infty$ is continuously differentiable. To prove that it is twice-differentiable, we note that Leibniz's integral rule applies again. Indeed, we have
  \begin{equation}
    \nabla^2 \Loss_\infty (\parz)
    = \int_\PARAM \nabla_\parz^2 \norm{\parz - \param}{2} p(\param) d\param
    = \int_\PARAM \frac{I-\unitvector{\parz-\param} \unitvector{\parz-\param}^T}{\norm{{\parz-\param}}{2}} p(\param) d\param,
  \end{equation}
  But note that each coordinate of the matrix $\frac{I-\unitvector{\parz-\param} \unitvector{\parz-\param}^T}{\norm{{\parz-\param}}{2}}$ is at most
  $\frac{1}{\norm{\parz-\param}{2}}$.
  By virtue of the previous lemma, for $d \geq 2$, this is integrable in $\parz$.
  Moreover, by isolating the integration in the ball $\ball(\parz, \varepsilon)$, we show that the impact of the integration in this ball is negligible as $\varepsilon \rightarrow 0$.
  Finally, the rest of the integration is integrable, as $\frac{1}{\norm{\parz-\param}{2}} p(\param)$ can be upper-bounded by $\frac{1}{\delta} p(\param)$ outside of $\ball(\parz,\delta)$, which is integrable by Assumption \ref{ass:pdf}.

  The cases of the third, fourth, and fifth derivatives are handled similarly, 
  with now the bounds $\absv{\partial_{ijk}^3 \norm{\parz - \param}{2}} \leq 6/\norm{\parz-\param}{2}^2$, 
  $\absv{\partial_{ijkl}^4 \norm{\parz - \param}{2}} \leq 36/\norm{\parz-\param}{2}^3$ and
  $\absv{\partial_{ijklm}^5 \norm{\parz - \param}{2}} \leq 300/\norm{\parz-\param}{2}^4$,
  and using $d \geq 5$.

  To prove the strict convexity, consider a point $\parz \in \PARAM$ such that $p(\parz) > 0$. By continuity of $p$, for any two orthogonal unit vectors $\unitvector{1}$ and $\unitvector{d}$ and $\eta >0$ small enough, we must have $p(\parz + \eta \unitvector{1}) >0$ and $p(\parz + \eta \unitvector{d}) >0$.
  For any $\varepsilon >0$, there must then be a strictly positive probability to draw a point in $\ball(\parz, \varepsilon)$, a point in $\ball(\parz + \eta \unitvector{1}, \varepsilon)$, and a point in $\ball(\parz + \eta \unitvector{d}, \varepsilon)$.
  Moreover, for $\varepsilon$ much smaller than $\eta$, then the three points thereby drawn cannot be colinear.
  We then obtain a situation akin to the proof of Proposition \ref{prop:uniqueness}.
  By the same argument, this suffices to prove that the Hessian matrix must be positive definite.
  Therefore, $\Loss_\infty$ is strictly convex.

  It follows straightforwardly from this that the geometric median is unique. Its existence can be derived by considering a ball $\ball(0,A)$ of probability at least $1/2$ according to $\paramdistribution$. If $\norm{\parz}{2} \geq A + 2 \expect{\norm{\param}{2}}$, then
  \begin{equation}
    \Loss_\infty(\parz) \geq \frac{1}{2} \left( A + 2 \expect{\norm{\param}{2}} - A \right) \geq \expect{\norm{\param}{2}} = \Loss_\infty(0).
  \end{equation}
  Thus $\Loss_\infty$ must reach a minimum in $\ball(0, A + 2 \expect{\norm{\param}{2}})$.
  Finally, we conclude that the geometric median must belong to $\PARAM$, by re-using the argument of Proposition \ref{prop:convex_hull}.
\end{proof}

\subsection{Proof Steps for Theorem~\ref{th:asymptotic_strategyproofness}}
\label{sec:proof_steps_main}
In this section, we provide the full proof of Theorem~\ref{th:asymptotic_strategyproofness} that consists of the following steps. First, in Section~\ref{sec:global_lemmas}, we find the sufficient conditions under which for a given function $F$ the set $\set{z : \norm{\nabla F(\parz)}{2} \leq 1}$ is convex. We then use this result to find sufficient conditions for the geometric median to become $\strategyproofbound$-strategyproof in Section~\ref{sec:sufficient_strategyproofness}. Then in Section~\ref{sec:finite} we show that these conditions are satisfied with high probability when the number of voters is large enough. Next, Section~\ref{sec:proof_skewness_continuous} proves that the $\skewness$ function is continuous which is necessary for the proof of our theorem.  Finally, Section \ref{sec:proof_th_asymptotic_strategyproofness} combines
these steps (lemmas \ref{lemma:convex_projection_strategyproof}, \ref{lemma:no_voter},  \ref{lemma:close_median},\ref{lemma:bounded_third_derivative}, and \ref{lemma:close_hessian}) to prove Theorem \ref{th:asymptotic_strategyproofness} .%

\input{asymptotic_higher_derivatives}

\input{asymptotic_sufficient_conditions}

\subsubsection{Finite-voter Guarantees}
\label{sec:finite}

We show here that for a large enough number of voters and with high probability, finite-voter approximations are well-behaved and, in some critical regards, approximate correctly the infinite case. 
The global idea of the proof is illustrated in Figure \ref{fig:proof_strategy}.
In particular, we aim to show that, when $\VOTER$ is large, the achievable set $\setAchieve$ is approximately an ellipsoid within a region where $\Loss_{1:\VOTER}$ is infinitely differentiable.
In particular, we show that, with arbitrarily high probability under the drawing of other voters' vectors, for $\VOTER$ large enough, the conditions of Lemma \ref{lemma:convex_projection_strategyproof} are satisfied for $\beta \triangleq \Theta(\VOTER^{-1})$ and $\strategyproofbound \triangleq \skewness(\hessian_\infty) + \varepsilon$.

\input{asymptotic_infinitely_differentiable}
\input{asymptotic_geometric_median}

\input{asymptotic_hessian}

\subsubsection{Skewness is Continuous}
\label{sec:proof_skewness_continuous}
The last piece that is required for the proof of Theorem~\ref{th:asymptotic_strategyproofness} is the fact that the function $\skewness$ is continuous. 
To get there, we first prove a couple of lemmas about symmetric matrices.

\begin{definition}
  We denote $\symmetric{d}$ the set of symmetric $d \times d$ real matrices. We denote $\twoDelement{X}{i}{j}$ the element of the $i$-th row and $j$-th column of the matrix $X$, and $\norm{X}{\infty} \triangleq \max_{i,j} \absv{\twoDelement{X}{i}{j}}$.
\end{definition}

\begin{lemma}
\label{lemma:min_spectrum_addition}
  For any symmetric matrices $\hessian, \sdpmatrix \in \symmetric{d}$,
  $\absv{\min \spectrum (\hessian) - \min \spectrum(\sdpmatrix)} \leq d \norm{\hessian - \sdpmatrix}{\infty}$.
\end{lemma}

\begin{proof}
  Consider a unit vector $\unitvector{}$. We have 
  \begin{align}
    &\unitvector{}^T \hessian \unitvector{} - \unitvector{}^T \sdpmatrix \unitvector{}
    = \unitvector{}^T (\hessian - \sdpmatrix) \unitvector{}
    = \sum_{i,j \in [d]} (\hessian[i,j] - \sdpmatrix[i,j]) \unitvector{}[i] \unitvector{}[j] \\ 
    &\qquad \leq \sum_{i,j \in [d]} \absv{\hessian[i,j] - \sdpmatrix[i,j]} \absv{\unitvector{}[i]} \absv{\unitvector{}[j]} 
    \leq \norm{\hessian - \sdpmatrix}{\infty} \left(\sum_{i \in [d]} \absv{\unitvector{}[i]} \right) \left( \sum_{j \in [d]} \absv{\unitvector{}[j]} \right) \\
    &\qquad = \norm{\hessian - \sdpmatrix}{\infty} \norm{\unitvector{}}{1}^2 
    \leq d \norm{\hessian - \sdpmatrix}{\infty} \norm{\unitvector{}}{2}^2
    = d \norm{\hessian - \sdpmatrix}{\infty},
  \end{align}
  where we used the well-known inequality $\norm{\parx}{1}^2 \leq d \norm{\parx}{2}^2$ (which follows from the convexity of $t \mapsto t^2$).
  Now consider $\unitvector{min}$ a unit eigenvector of the eigenvalue $\min \spectrum (\sdpmatrix)$ of the symmetric matrix $\sdpmatrix$.
  Then $\min \spectrum (\hessian) \leq \unitvector{min}^T \hessian \unitvector{min} 
  \leq \unitvector{min}^T \sdpmatrix \unitvector{min} + d \norm{\hessian - \sdpmatrix}{\infty}
  = \min \spectrum (\sdpmatrix) + d \norm{\hessian - \sdpmatrix}{\infty}$.
  Inverting the role of $\hessian$ and $\sdpmatrix$ then yields the lemma.
\end{proof}

\begin{lemma}
  \label{lemma:continuous_min_spectrum}
  The minimal eigenvalue is a continuous function of a symmetric matrix.
\end{lemma}

\begin{proof}
  This is an immediate corollary of the previous lemma. 
  As $\sdpmatrix \rightarrow \hessian$, we clearly have $\min \spectrum(\sdpmatrix) \rightarrow \min \spectrum(\hessian)$.
\end{proof}

\begin{lemma}
\label{lemma:continuous_skewness}
  $\skewness$ is continuous.
\end{lemma}

\begin{proof}
  Consider $\hessian, \sdpmatrix \succ 0$ two positive definite symmetric matrices. We have
  \begin{align}
    &\absv{\skewness(\hessian)-\skewness(\sdpmatrix)} 
    = \absv{ \sup_{\norm{\unitvector{}}{2} = 1} \set{\frac{\norm{\hessian \unitvector{}}{2}}{\unitvector{}^T \hessian \unitvector{}} -1}-\sup_{\norm{\unitvector{}}{2} = 1} 
    \set{\frac{\norm{ \sdpmatrix \unitvector{}}{2}}{\unitvector{}^T \sdpmatrix \unitvector{}} -1}} \\
    &\leq \sup_{\norm{\unitvector{}}{2} = 1} \set{\absv{\frac{\norm{\hessian \unitvector{}}{2}}{\unitvector{}^T \hessian \unitvector{}} - \frac{\norm{\sdpmatrix \unitvector{}}{2}}{\unitvector{}^T \sdpmatrix \unitvector{}}}}\\
    &= \sup_{\norm{\unitvector{}}{2} = 1} \set{\absv{\frac{\norm{\hessian \unitvector{}}{2}(\unitvector{}^T \sdpmatrix \unitvector{})-\norm{\sdpmatrix \unitvector{}}{2}(\unitvector{}^T \hessian \unitvector{})}{(\unitvector{}^T \hessian \unitvector{})(\unitvector{}^T \sdpmatrix \unitvector{})} }}\\
    &\leq \sup_{\norm{\unitvector{}}{2} = 1} \set{\absv{\frac{\norm{\hessian \unitvector{}}{2}\left( \unitvector{}^T \sdpmatrix \unitvector{} - \unitvector{}^T \hessian \unitvector{} \right)}{(\unitvector{}^T \hessian \unitvector{})(\unitvector{}^T \sdpmatrix \unitvector{})}} 
    + \absv{\frac{\norm{\hessian \unitvector{}}{2}-\norm{\sdpmatrix \unitvector{}}{2}}{\unitvector{}^T \sdpmatrix \unitvector{}}}}\\
    &\leq \sup_{\norm{\unitvector{}}{2} = 1} \set{(\skewness(\hessian)+1)\absv{\frac{\left( \unitvector{}^T \sdpmatrix \unitvector{} - \unitvector{}^T \hessian \unitvector{} \right)}{\unitvector{}^T \sdpmatrix \unitvector{}}} 
    +\absv{\frac{\norm{\hessian \unitvector{}}{2}-\norm{\sdpmatrix \unitvector{}}{2}}{\unitvector{}^T \sdpmatrix \unitvector{}}}} \label{equ:skewness_inequality}.
\end{align}
Now, for any unit vector $\unitvector{}$, we have
\begin{align}
  \absv{ \unitvector{}^T \sdpmatrix \unitvector{} - \unitvector{}^T \hessian \unitvector{}} 
  &\leq \sum_{i,j \in [d]} \absv{\unitvector{}[i]} \absv{\unitvector{}[j]} \absv{\sdpmatrix[i,j] - \hessian[i,j]} 
  \leq \norm{\sdpmatrix - \hessian}{\infty} \sum_{i,j \in [d]} \absv{\unitvector{}[i]} \absv{\unitvector{}[j]} \\
  &=  \norm{\sdpmatrix - \hessian}{\infty} \norm{\unitvector{}}{1}^2 
  = d \norm{\sdpmatrix - \hessian}{\infty} \norm{\unitvector{}}{2}^2
  = d \norm{\sdpmatrix - \hessian}{\infty},
\end{align}
using the inequality $\norm{\parx}{1}^2 \leq d \norm{\parx}{2}^2$.
Moreover, by triangle inequality, we also have 
\begin{align}
  &\absv{ \norm{\hessian \unitvector{}}{2}-\norm{\sdpmatrix \unitvector{}}{2}} 
  \leq \norm{\hessian \unitvector{} - \sdpmatrix \unitvector{}}{2} 
  \leq \sqrt{\sum_{i \in [d]} \left(\sum_{j \in d} \absv{\hessian[i,j] - \sdpmatrix[i,j]} \absv{\unitvector{}[j]} \right)^2} \\
  &\qquad \leq \sqrt{\sum_{i \in [d]} \left(\sum_{j \in d} \norm{\hessian - \sdpmatrix}{\infty} \absv{\unitvector{}[j]} \right)^2} 
  \leq \norm{\hessian - \sdpmatrix}{\infty} \sqrt{ d \norm{\unitvector{}}{1}^2} \\
  &\qquad \leq \norm{\hessian - \sdpmatrix}{\infty} \sqrt{ d^2 \norm{\unitvector{}}{2}^2} 
  \leq d \norm{\hessian - \sdpmatrix}{\infty} \norm{\unitvector{}}{2}.
\end{align}
Finally, note that $\unitvector{}^T \sdpmatrix \unitvector{} \geq \min \spectrum(\sdpmatrix)$.
Combining it all then yields
\begin{align}
   \absv{\skewness(\hessian)-\skewness(\sdpmatrix)} \leq \frac{ 2+\skewness(\hessian)}{\min \spectrum(\sdpmatrix)} d \norm{\hessian - \sdpmatrix}{\infty}.
\end{align}
By continuity of the minimal eigenvalue (Lemma \ref{lemma:continuous_min_spectrum}), we know that $\min \spectrum(\sdpmatrix) \rightarrow \min \spectrum(\hessian)$ as $\sdpmatrix \rightarrow \hessian$.
This allows us to conclude that $\absv{\skewness(\hessian)-\skewness(\sdpmatrix)} \rightarrow 0$ as $\sdpmatrix \rightarrow \hessian$, which proves the continuity of the $\skewness$ function.
\end{proof}

\input{asymptotic_final_proof}
\subsection{Upper and Lower Bounds for Skewness (Proof of Proposition \ref{prop:skewness_lowerbound})}
\label{sec:proof_prop_skewness_lowerbound}
\begin{proof}
We first prove the upper-bound. Consider an orthonormal eigenvector basis of $\sdpmatrix$ of vectors $\unitvector{1}, \ldots, \unitvector{d}$, with respective eigenvalues $\lambda_1, \ldots, \lambda_d$. We now focus on a unit vector $x$ in the form $\parx = \sum \beta_i \unitvector{i}$ with $\sum \beta_i^2 = 1$. 
  Note that $\sum \beta_i^2 \lambda_i$ and $\sum \beta_i^2 \lambda_i^2$ can then be viewed as weighted averages of $\lambda_i$'s and of their squares.
  As a result, we have $\sum \beta_i^2 \lambda_i \geq \lambda_{min}$ and $\sum \beta_i^2 \lambda_i^2 \leq \lambda_{max}^2$.
  As a result, we have
  \begin{equation}
    \frac{\norm{\sdpmatrix \parx}{2}^2}{(\parx^T \sdpmatrix \parx)^2}
= \frac{\sum \beta_i^2 \lambda_i^2}{\left( \sum \beta_i^2 \lambda_i \right)^2}
    \leq \frac{\lambda_{max}^2}{\lambda_{min}^2}.
  \end{equation}
  Taking the square root and subtracting one proves the upper-bound. We now move on to proving the lower-bound.
  Denote $\lambda_1$ and $\lambda_d$ the two extreme eigenvalues of $\sdpmatrix$, and $u_1$ and $u_d$ their orthogonal unit eigenvectors.
  Define $\parx = \frac{\unitvector{1}}{\sqrt{\lambda_1}} + \frac{\sqrt{\beta} \unitvector{d}}{\sqrt{\lambda_d}}$.
  We then have
  $\sdpmatrix \parx = \sqrt{\lambda_1} \unitvector{1} + \sqrt{\beta \lambda_d} \unitvector{d}$,
  $\norm{\parx}{2}^2 = \lambda_1^{-1} + \beta \lambda_d^{-1}$,
  $\parx^T \sdpmatrix \parx = 1 + \beta$, and
  $\norm{\sdpmatrix \parx}{2}^2 = \lambda_1 + \beta \lambda_d$.
  Combining this yields a ratio
  \begin{align}
    R(\beta)
    &\triangleq \frac{\norm{\parx}{2}^2 \norm{\sdpmatrix\parx}{2}^2}{(\parx^T \sdpmatrix \parx)^2}
    = \frac{(\lambda_1^{-1} + \beta \lambda_d^{-1}) (\lambda_1 + \beta \lambda_d)}{(1+\beta)^2}
    = \frac{\beta^2 + L \beta +1}{\beta^2 + 2 \beta + 1}
    = 1+\frac{L-2}{2+ \beta + \beta^{-1}},
  \end{align}
  where $L \triangleq \lambda_1^{-1} \lambda_d + \lambda_1 \lambda_d^{-1} \geq 2$. Note that, similarly, we have $\beta + \beta^{-1} \geq 2$. This implies $R(\beta) \leq 1+\frac{L-2}{4}$, with equality for $\beta =1$. The skewness is then greater than $\sqrt{R(1)} -1$. In two dimensions, the skewness is, in fact, equal to $\sqrt{R(1)} -1$ since $x$ represents essentially all possible vectors in this case.
  Multiplying the numerator and denominator of $R(1)$ by $\lambda_1 \lambda_d$ then yields the proposition.
\end{proof}

%% file: asymptotic_higher_derivatives.tex
\subsubsection{Higher Derivatives and Unit-norm Gradients}
\label{sec:global_lemmas}
Note that our analysis involves the third derivative tensor to guarantee the convexity of the achievable set (defined in~\eqref{eq:achivable_set}).
 Therefore, here we provide a discussion about higher-order derivatives.
We consider here a three-times continuously differentiable convex function $F$, and we study the set of points $\parz$ such that $\norm{\nabla F(\parz)}{2} \leq 1$. 
In particular, we provide a sufficient condition for the convexity of this set, based on the study of the first three derivatives of $F$.
This convexity guarantee then allows us to derive a sufficient condition on $\Loss_{1:\VOTER}$ to guarantee $\strategyproofbound$-strategyproofness.

To obtain such guarantee, let us recall a few facts about higher derivatives. 
In general, the $n$-th derivative of a function $F : \setR^d \rightarrow \setR$ at a point $\parz$ is a (symmetric) tensor $\nabla^n F(\parz) : \underbrace{\setR^d \otimes \ldots \otimes \setR^d}_{n \text{ times}} \rightarrow \setR$, which inputs $n$ vectors and outputs a scalar. 
This tensor $\nabla^n F(\parz)$ is linear in each of its $n$ input vectors.
More precisely, its value for input $[\parx_1 \otimes \ldots \otimes \parx_n]$ is
\begin{equation}
  \nabla^n F(\parz) [\parx_1 \otimes \ldots \otimes \parx_n] 
  = \sum_{i_1 \in [d]} \ldots \sum_{i_n \in [d]} \left( \parx_1 [i_1] \, \parx_2 [i_2] \ldots \, \parx_{n-1} [i_{n-1}] \, \parx_n [i_n] \right) \partial_{i_1 \ldots i_n}^n F(\parz),
\end{equation}
where $\partial_{i_1 \ldots i_n}^n F(\parz)$ is the $n$-th partial derivative of $F$ with respect to the coordinates $i_1$, $i_2$, \ldots, $i_n$ (by the symmetry of derivation, the order in which $F$ is derived along the different coordinates does not matter).

For $n=1$, we see that $\nabla F(\parz)$ is simply a linear form $\setR^d \rightarrow \setR$.
By Euclidean duality, $\nabla F(\parz)$ can thus be regarded as a vector, called the \emph{gradient}, such that $\nabla F(\parz) [\parx] = \parx^T \nabla F(\parz)$.
Note that if $F$ is assumed to be convex, but not differentiable, $\nabla F(\parz)$ represents its set of subgradients at point $\parz$, i.e., $\subgradient \in \nabla F(\parz)$ if and only if $F(\parz + \delta) \geq F(\parz) + \subgradient^T \delta$ for all $\delta \in \setR^d$.
From this definition, it follows straightforwardly that $\parz$ minimizes $F$ if and only if $0 \in \nabla F(\parz)$.

For $n=2$, $\nabla^2 F(\parz)$ is now a bilinear form $\setR^d \otimes \setR^d \rightarrow \setR$.
By isomorphism between (symmetric) bilinear forms and (symmetric) matrices, $\nabla^2 F(\parz)$ can equivalently be regarded as a (symmetric) matrix, called the \emph{Hessian matrix}, such that $\nabla^2 F(\parz) [\parx \otimes \pary] = \parx^T \left( \nabla^2 F(\parz) \right) \pary$.

A bilinear form $B : \setR^d \otimes \setR^d \rightarrow \setR$ is said to be \emph{positive semi-definite} (respectively, \emph{positive definite}), if $B [\parx \otimes \parx] \geq 0$ for all $\parx \in \setR^d$ (respectively, $B [\parx \otimes \parx] > 0$ for all $\parx \neq 0$).
If so, we write $B \succeq 0$ (respectively, $B \succ 0$).
Moreover, given any $\parx \in \setR^d$, the function $\pary \mapsto B[\parx \otimes \pary]$ becomes a linear form, which we denote $B[\parx]$.
When the context is clear, $B[\parx]$ can equivalently be regarded as a vector.
Finally, given two bilinear form $A, B : \setR^d \otimes \setR^d \rightarrow \setR$, we can define their composition $A \cdot B : \setR^d \otimes \setR^d \rightarrow \setR$ by $A \cdot B [\parx \otimes \pary] \triangleq A[\parx \otimes B[\pary]] = \parx^T AB \pary$, where, in the last equation, $A$ and $B$ are regarded as matrices.

We also need to analyze the third derivative of $F$, which can thus be regarded as a 3-linear form $\nabla^3 F(\parz) : \setR^d \otimes \setR^d \otimes \setR^d \rightarrow \setR$.
Note as well that, for any 3-linear form $W$ and any fixed first input $\parw \in \setR^d$, the function $(\parx \otimes \pary) \mapsto W [\parw \otimes \parx \otimes \pary]$ is now a bilinear (symmetric) form $\setR^d \otimes \setR^d \rightarrow \setR$.
This (symmetric) bilinear form will be written $W [\parw]$ or $W \cdot \parw$, which can thus equivalently be regarded as a (symmetric) matrix.
Similarly, $W[\parx \otimes \pary]$ can be regarded as a linear form $\setR^d \rightarrow \setR$, or, by Euclidean duality, as a vector in $\setR^d$.

Finally, we can state the following lemma, which provides a sufficient condition for the convexity of the sets of $\parz \in \setR^d$ with a unit-norm $F$-gradient.

\begin{lemma}
\label{lemma:convexity}
  Assume that $\convexset \subset \setR^d$ is convex 
  and that $\nabla^2 F(z) \cdot \nabla^2 F(z) + \nabla^3 F(z) \cdot \nabla F(z) \succeq 0$ for all $\parz \in \convexset$. 
  Then $\parz \mapsto \norm{\nabla F(\parz)}{2}^2$ is convex on $\convexset$.
\end{lemma}

\begin{proof}
  Fix $i \in [d]$. By Taylor approximation of $\partial_i F$ around $\parz$, for $\delta \rightarrow 0$, we have
  \begin{equation}
    \partial_i F(z+\delta) = \partial_i F(z) + \sum_{j \in [d]} \delta_j \partial_{ij}^2 F(z) + \frac{1}{2} \sum_{j,k \in [d]} \delta_j \delta_k \partial_{ijk}^3 F(z) + o(\norm{\delta}{2}^2).
  \end{equation}
  This equation can equivalently be written:
  \begin{equation}
    \nabla F(z+\delta) = \nabla F(z) + \nabla^2 F(z) [\delta] + \frac{1}{2} \nabla^3 F(z) [\delta \otimes \delta] + o(\delta^2).
  \end{equation}
  Plugging this into the computation of the square norm of the gradient yields:
  \begin{align}
    &\norm{\nabla F(z+\delta)}{2}^2
    = \norm{\nabla F(z) + \nabla^2 F(z) [\delta] + \frac{1}{2} \nabla^3 F(z) [\delta \otimes \delta] + o(\delta^2)}{2}^2 \\
    &= \norm{\nabla F(z)}{2}^2
    + 2 \nabla^2 F(z) \left[\nabla F(z) \otimes \delta\right]
    + \norm{\nabla^2 F(z) [\delta]}{2}^2
    + \nabla^3 F(z) \left[\nabla F(z) \otimes \delta \otimes \delta \right]
    + o(\norm{\delta}{2}^2) \\
    &= \norm{\nabla F(z)}{2}^2
    + 2 \nabla^2 F(z) \left[\nabla F(z) \otimes \delta\right]
    + \left( \nabla^2 F(z) \cdot \nabla^2 F(z) + \nabla^3 F(z) \cdot \nabla F(z) \right) [\delta \otimes \delta]
    + o(\norm{\delta}{2}^2).
  \end{align}
  Therefore, matrix $2 \left(\nabla^2 F(z) \cdot \nabla^2 F(z) + \nabla^3 F(z) \cdot \nabla F(z) \right)$ is the Hessian matrix of $z \mapsto \norm{\nabla F(z)}{2}^2 = \nabla F(\parz)^T \nabla F(\parz)$. Yet a twice differentiable function with a positive semi-definite Hessian matrix is convex.
\end{proof}

\begin{lemma}
\label{lemma:convex_achievable_set}
  Assume that $F$ is convex, and that there exists $\parz^* \in \setR^d$ and $\beta > 0$ such that, 
  for any unit vector $\unitvector{}$, there exists a subgradient $\subgradient \in \nabla F(\parz^* + \beta \unitvector{})$ such that $\unitvector{}^T \subgradient > 1$.
  Then the set $ \mathcal{A} \triangleq \set{\parz \in \setR^d \st \exists \subgradient \in \nabla F(\parz) \mathsep \norm{\subgradient}{2} \leq 1}$ of points where $\nabla F$ has a subgradient of at most a unit norm is included in the ball $\ball (\parz^*, \beta)$.
\end{lemma}

\begin{proof}
  Let $\parz \notin \ball(\parz^*, \beta)$. 
  Then there must exist $\gamma \geq \beta$ and a unit vector $\unitvector{}$ such that $\parz - \parz^* = \gamma \unitvector{}$.
  Denote $\parz_{\unitvector{}} \triangleq \parz^* + \beta \unitvector{}$.
  We then have $\parz - \parz_{\unitvector{}} = (\gamma - \beta) \unitvector{}$. 
  Moreover, we know that there exists $\subgradient_{\parz_{\unitvector{}}} \in \nabla F(\parz_{\unitvector{}})$ such that $\unitvector{}^T \subgradient_{\parz_{\unitvector{}}} > 1$.
  By convexity of $F$, for any $\subgradient_\parz \in \nabla F(\parz)$, we then have
  \begin{align}
    \left( \parz - \parz_{\unitvector{}} \right)^T 
    \left( \subgradient_\parz - \subgradient_{\parz_{\unitvector{}}} \right) 
    = (\gamma-\beta) \unitvector{}^T \left( \subgradient_\parz - \subgradient_{\parz_{\unitvector{}}} \right) \geq 0.
  \end{align}
  From this, it follows that $\norm{\subgradient_\parz}{2} \geq \unitvector{}^T \subgradient_\parz \geq \unitvector{}^T \subgradient_{\parz_{\unitvector{}}} > 1$.
  Thus $\parz \notin \mathcal{A} $. 
\end{proof}

%% file: asymptotic_sufficient_conditions.tex
\subsubsection{Sufficient Conditions for \texorpdfstring{$\strategyproofbound$}{a}-Strategyproofness}
\label{sec:sufficient_strategyproofness}

\begin{figure}[t]
    \centering
    \includegraphics[width=.3\textwidth]{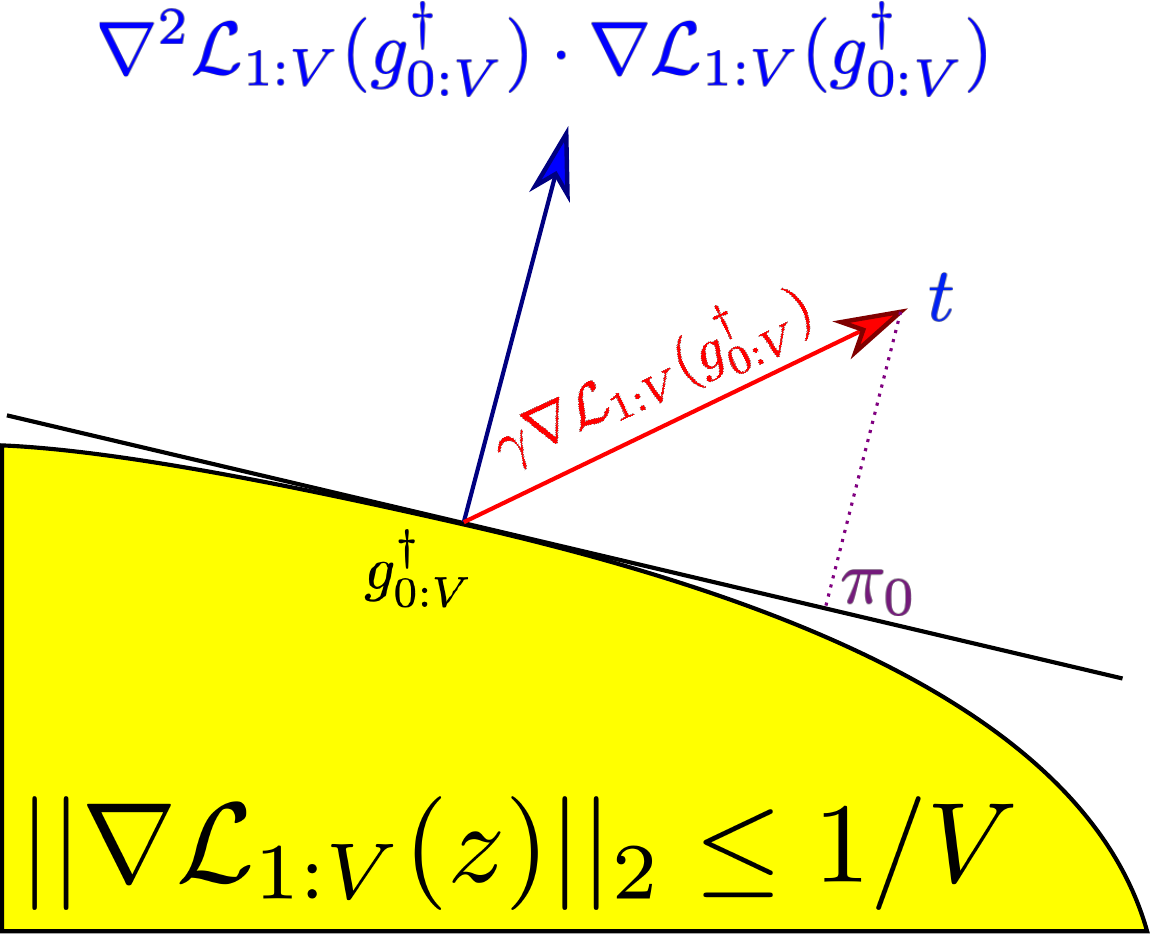}
    \caption{Illustration of what can be gained for target vector $\targetvector \triangleq \geometricmedian_{0:\VOTER}^\dagger + \gamma \nabla \Loss_{1:\VOTER}(\geometricmedian_{0:\VOTER}^\dagger)$. 
    The orthogonal projection $\pi_0$ of $\targetvector$ on the tangent hyperplane of the achievable set going through $\geometricmedian_{0:\VOTER}^\dagger$ yields a lower bound on what can be achieved by voter $0$ through their strategic vote $\strategicvote{0}$.
    This lower bound depends on the angle between $\nabla \Loss_{1:\VOTER}(\geometricmedian_{0:\VOTER}^\dagger)$ and the normal to the hyperplane $\nabla^2 \Loss_{1:\VOTER}(\geometricmedian_{0:\VOTER}^\dagger) \cdot \nabla \Loss_{1:\VOTER}(\geometricmedian_{0:\VOTER}^\dagger)$.}
    \label{fig:projection}
\end{figure}

 Recall from~\eqref{eq:achivable_set} that the achievable set $\setAchieve$ consists of the points $\parz$ such that 
there exists a subgradient $\subgradient \in \nabla_\parz \Loss_{1:\VOTER} (\parz)$ such that 
$\norm{\subgradient}{2}\leq 1/\VOTER$.
 Below, we identify a sufficient condition on $\setAchieve$ to guarantee $\strategyproofbound$-strategyproofness. 
 Note that as explained in the previous section, this analysis involves the third derivative tensor to guarantee the convexity of the achievable set, 
 so that the proof ideas illustrated in Figure~\ref{fig:projection} are applicable.

\begin{lemma}
\label{lemma:convex_projection_strategyproof}
  Assume that $\dim \paramfamily_\VOTER \geq 2$ and that the following conditions hold for some $\beta > 0$:
  \begin{itemize}
    \item {\bf Smoothness:} $\Loss_{1:\VOTER}$ is three-times continuously differentiable on  $\ball(\geometricmedian_{1:\VOTER}, 2 \beta)$.
    \item {\bf Contains $\setAchieve$:} For all unit vectors $\unitvector{}$, $\unitvector{}^T \nabla \Loss_{1:\VOTER} (\geometricmedian_{1:\VOTER} + \beta \unitvector{}) > 1/\VOTER$.
    \item {\bf Convex $\setAchieve$:} $\forall \parz \in \ball(\geometricmedian_{1:\VOTER}, \beta) \mathsep \nabla^2 \Loss_{1:\VOTER}(\parz) \cdot \nabla^2 \Loss_{1:\VOTER}(\parz) + \nabla^3 \Loss_{1:\VOTER}(\parz) \cdot \nabla \Loss_{1:\VOTER}(\parz) \succeq 0$.
    \item {\bf Bounded skewness:} $\forall \parz \in \ball(\geometricmedian_{1:\VOTER}, \beta) \mathsep \skewness (\nabla^2 \Loss_{1:\VOTER}(\parz)) \leq \alpha$.
  \end{itemize}
  Then the geometric median is $\strategyproofbound$-strategyproof for voter $0$.
\end{lemma}
\begin{proof}
  Given Lemma \ref{lemma:achievable_set}, we know that, for $\targetvector \in \setAchieve$, we have $\norm{\GeometricMedian(\targetvector, \paramfamily_\VOTER) - \targetvector}{2} = 0$, which guarantees $\strategyproofbound$-strategyproofness for such voters.
  
  Now assume $\targetvector \notin \setAchieve$, and recall that we defined $\geometricmedian^\dagger_{0:\VOTER} \triangleq \GeometricMedian(\targetvector, \paramfamily_\VOTER)$ as the truthful geometric median.
  By Lemma \ref{lemma:achievable_set}, we know that $\geometricmedian^\dagger_{0:\VOTER} \in \setAchieve$.
  Thus $\targetvector \neq \geometricmedian^\dagger_{0:\VOTER}$.
  Moreover, applying Lemma \ref{lemma:convex_achievable_set} to $F \triangleq \VOTER \Loss_{1:\VOTER}$ guarantees that $\setAchieve \subset \ball(\geometricmedian_{1:\VOTER},  \beta)$. The first condition shows that $\Loss_{1:\VOTER}$ is 3-times differentiable in a neighborhood of $\geometricmedian^\dagger_{0:\VOTER}$. Plus, given the third condition, by Lemma \ref{lemma:convexity}, we know that $\setAchieve$ is a convex set.

  Now, by definition, $\geometricmedian^\dagger_{0:\VOTER}$ must minimize the loss $\Loss_{0:\VOTER} (\targetvector, \cdot)$, i.e., we must have 
  \begin{align}
    0 = \nabla \Loss_{0:\VOTER} (\targetvector, \geometricmedian^\dagger_{0:\VOTER})
    = \unitvector{\geometricmedian^\dagger_{0:\VOTER} - \targetvector} + \VOTER \nabla \Loss_{1:\VOTER} (\geometricmedian^\dagger_{0:\VOTER}).
  \end{align}
  Equivalently, we have $\unitvector{\targetvector - \geometricmedian^\dagger_{0:\VOTER}} = \VOTER \nabla \Loss_{1:\VOTER} (\geometricmedian^\dagger_{0:\VOTER})$.
  This means that $\norm{\nabla \Loss_{1:\VOTER} (\geometricmedian^\dagger_{0:\VOTER})}{2} = 1/\VOTER$, and that there must exist $\gamma > 0$ such that $\targetvector = \geometricmedian^\dagger_{0:\VOTER} + \gamma \nabla \Loss_{1:\VOTER} (\geometricmedian^\dagger_{0:\VOTER})$.
  
  For $\delta \in \setR^d$ small enough, Taylor approximation then yields
  \begin{align}
    \norm{\nabla F(\geometricmedian^\dagger_{0:\VOTER} + \delta)}{2}^2
    &= \norm{\nabla F(\geometricmedian^\dagger_{0:\VOTER}) + \nabla^2 F(\geometricmedian^\dagger_{0:\VOTER}) [\delta] + o(\norm{\delta}{2})}{2}^2 \\
    &= \norm{\nabla F(\geometricmedian^\dagger_{0:\VOTER})}{2}^2 
      + 2 \nabla^2 F(\geometricmedian^\dagger_{0:\VOTER}) \left[ \nabla F(\geometricmedian^\dagger_{0:\VOTER}) \otimes \delta \right] 
      + o(\norm{\delta}{2}) \\
    &= 1 + 2 \subgradient^T \delta + o(\norm{\delta}{2}),
  \end{align}
  where $\subgradient \triangleq \nabla^2 F(\geometricmedian^\dagger_{0:\VOTER}) \cdot \nabla F(\geometricmedian^\dagger_{0:\VOTER})$.
  
  Since $\parz \mapsto \norm{\nabla F(\parz)}{2}^2$ is convex on $\ball(\geometricmedian_{1:\VOTER}, \beta)$, 
  we know that, in this ball, 
  $2 \subgradient$ is thus a subgradient of $\parz \mapsto \norm{\nabla F(\parz)}{2}^2$ at $\geometricmedian^\dagger_{0:\VOTER}$.
  Thus, in fact, for all $\delta \in \ball(\geometricmedian_{1:\VOTER}-\geometricmedian_{0:\VOTER}, \beta)$, we have $\norm{\nabla F(\geometricmedian^\dagger_{0:\VOTER} + \delta)}{2}^2 \geq 1 + 2 \subgradient^T \delta$.
  Now assume that $\geometricmedian^\dagger_{0:\VOTER} + \delta \in \setAchieve$.
  Then we must have $2 \subgradient^T \delta \leq \norm{\nabla F(\geometricmedian^\dagger_{0:\VOTER} + \delta)}{2}^2 - 1 \leq 0$.
  In other words, we must have $\setAchieve \subset \mathcal H$,
  where $\mathcal H \triangleq \set{\parz \in \setR^d \st \subgradient^T \parz \leq \subgradient^T \geometricmedian^\dagger_{0:\VOTER}}$ is the half space of the hyperplane 
  that goes through the truthful geometric median $\geometricmedian^\dagger_{0:\VOTER}$,
  and whose normal direction is $\subgradient$.
  
  Using Lemma \ref{lemma:achievable_set} and the inclusion $\setAchieve \subset \mathcal H$ then yields
  \begin{align}
    \inf_{\strategicvote{0} \in \setR^d} \norm{\GeometricMedian(\strategicvote{0}, \paramfamily_\VOTER) - \targetvector}{2}
    &= \inf_{\parz \in \setAchieve} \norm{\parz - \targetvector}{2}
    \geq \inf_{\parz \in \mathcal H} \norm{\parz - \targetvector}{2}.
  \end{align}
  Yet the minimal distance between a point $\targetvector$ and a half space $\mathcal H$ is reached by the orthogonal projection $\pi_0$ of $\targetvector$ onto $\mathcal H$, as depicted in Figure \ref{fig:projection}.
  We then have
  \begin{align}
    \norm{\targetvector - \pi_0}{2}
    &= \left( \gamma \nabla \Loss_{1:\VOTER} (\geometricmedian_{0:\VOTER}^\dagger) \right)^T \frac{\subgradient}{\norm{\subgradient}{2}}
    = \frac{\gamma \nabla^2 \Loss_{1:\VOTER}(\geometricmedian_{0:\VOTER}^\dagger) \left[ \nabla \Loss_{1:\VOTER}(\geometricmedian_{0:\VOTER}^\dagger) \otimes \nabla \Loss_{1:\VOTER}(\geometricmedian_{0:\VOTER}^\dagger) \right]}
    { \norm{\nabla^2 \Loss_{1:\VOTER}(\geometricmedian_{0:\VOTER}^\dagger) \cdot \nabla \Loss_{1:\VOTER}(\geometricmedian_{0:\VOTER}^\dagger)}{2}} \\
    &\geq \frac{\gamma \norm{\nabla \Loss_{1:\VOTER}(\geometricmedian_{0:\VOTER}^\dagger)}{2}}{1 + \skewness(\nabla^2 \Loss_{1:\VOTER} (\geometricmedian_{0:\VOTER}^\dagger))} 
    \geq \frac{\norm{\gamma \nabla \Loss_{1:\VOTER}(\geometricmedian_{0:\VOTER}^\dagger)}{2}}{1 + \strategyproofbound},
  \end{align}
  using our fourth assumption.
  Yet note that $\norm{\geometricmedian^\dagger_{0:\VOTER} - \targetvector}{2} = \norm{\gamma \nabla \Loss_{1:\VOTER} (\geometricmedian^\dagger_{0:\VOTER})}{2}$.
  We thus obtain $\norm{\geometricmedian^\dagger_{0:\VOTER} - \targetvector}{2} \leq (1+\strategyproofbound) \norm{\targetvector - \pi_0}{2} \leq (1+\strategyproofbound) \inf_{\strategicvote{0} \in \setR^d} \norm{\GeometricMedian(\strategicvote{0}, \paramfamily_\VOTER) - \targetvector}{2}$, which is the lemma.
\end{proof}

%% file: asymptotic_infinitely_differentiable.tex
\paragraph{\textbf{An Infinitely-differentiable Region.}}
Now, in order to apply Lemma \ref{lemma:convex_projection_strategyproof}, we need to identify a region near $\geometricmedian_\infty$ where, with high probability, the loss function $\Loss_{1:\VOTER}$ is infinitely differentiable.
To do this, we rely on the observation that, in high dimensions, random points are very distant from one another.
More precisely, the probability of randomly drawing a point $\varepsilon$-close to the geometric median $\geometricmedian_\infty$ is approximately proportional to $\varepsilon^d$, which is exponentially small in $d$.
This allows us to prove that, with high probability, none of the first $\VOTER$ voters will be $\VOTER^{-\ra}$-close to the geometric median, where $\ra > 1/d$ is a positive constant.

\begin{lemma}
\label{lemma:no_voter}
  Under Assumption \ref{ass:pdf}, for any $\delta_1 >0$, and $\ra> 1/d$, there exists $\VOTER_1(\delta_1) \in \setN$ such that, for $\VOTER \geq \VOTER_1(\delta_1)$, with probability at least $1-\delta_1$, 
  we have $\norm{\paramsub{\voter} - \geometricmedian_\infty}{2} \geq \VOTER^{-\ra}$ for all voters $\voter \in [\VOTER]$.
  In particular, in such a case, $\Loss_{1:\VOTER}$ is then infinitely differentiable in $\ball(\geometricmedian_\infty, \VOTER^{-\ra})$.
\end{lemma}
\begin{proof}
  Denote $p_\infty \triangleq 1+p(\geometricmedian_\infty)$ the probability density at $\geometricmedian_\infty$.
  Since $p$ is continuous, we know that there exists $\varepsilon_0 >0$ such that $p(\parz) \leq p_\infty$ for all $\parz \in \ball(\geometricmedian_\infty, \varepsilon_0)$.
  Thus, for any $0 < \varepsilon \leq \varepsilon_0$, we know that $\probability{\param \in \ball(\geometricmedian_\infty, \varepsilon)} \leq  volume_d(\varepsilon) p_\infty$,
  where $volume_d(\varepsilon)$ is the volume of Euclidean $d$-dimensional ball with radius $\varepsilon$. 
  Yet this volume is known to be upper-bounded by ${8\pi^2\varepsilon^d}/15$ \citep{David89}.
  Thus for $\VOTER \geq \varepsilon_0^{-1/\ra}$ (and thus $\VOTER^{-\ra} \leq \varepsilon_0$), we have $\probability{\param \in \ball(\geometricmedian_\infty, \VOTER^{-\ra})} \leq \frac{8\pi^2}{15} p_\infty \VOTER^{-\ra d}$.
  Now note that
  \begin{align}
    \probability{\forall \voter \in [\VOTER] \mathsep \paramsub{\voter} \notin \ball(\geometricmedian_\infty, \VOTER^{-\ra})}
    &= 1- \probability{\exists \voter \in [\VOTER] \mathsep \paramsub{\voter} \in \ball(\geometricmedian_\infty,  \VOTER^{-\ra})} \\
    &\geq 1- \sum_{\voter \in \VOTER} \probability{\paramsub{\voter} \in \ball(\geometricmedian_\infty,  \VOTER^{-\ra})}
    \geq 1 - \frac{8\pi^2}{15} p_\infty \VOTER^{1-\ra d}.
  \end{align}
  Now recall that $\ra > \frac{1}{d}$. We thus have $\frac{8\pi^2}{15} p_\infty \VOTER^{1-\ra d} \rightarrow 0$ as $\VOTER \rightarrow \infty$.
  But now taking $\VOTER \geq \VOTER_1(\delta_1) \triangleq \max\set{\varepsilon_0^{-1/\ra}, (8 \pi^2 p_\infty/15\delta_1)^{1/(\ra d-1)} }$, we see that, with probability at least $1-\delta_1$, no voter $\voter \in [\VOTER]$ is $\VOTER^{-\ra}$-close to $\geometricmedian_\infty$.
  Given the absence of singularity in $\ball(\geometricmedian_\infty, \VOTER^{-\ra})$ in such a case, $\Loss_{1:\VOTER}$ is infinitely differentiable in this region.
\end{proof}

%% file: asymptotic_geometric_median.tex
\paragraph{\textbf{Approximation of the Infinite Geometric Median.}}
The following lemma shows that as $\VOTER$ grows, $\geometricmedian_{1:\VOTER}$ gets closer to $\geometricmedian_\infty$ with high probability.

\begin{lemma}
\label{lemma:close_median}
  Under Assumption \ref{ass:pdf}, for any $\delta_2 >0$, and $0 < \rc < 1/2$, there exists $\VOTER_2(\delta_2) \in \setN$ such that, 
  for all $\VOTER \geq \VOTER_2(\delta_2)$, with probability at least $1-\delta_2$, 
  we have $\norm{\geometricmedian_{1:\VOTER} - \geometricmedian_\infty}{2} \leq \VOTER^{- \rc}$.
\end{lemma}

\begin{proof}
  Since $\nabla \Loss_\infty (\geometricmedian_\infty) = 0$ and $\Loss_\infty$ is three times differentiable, using Taylor's theorem around $\geometricmedian_\infty$, for any $\parz \in \ball(0,1)$, we have $\nabla \Loss_\infty (\geometricmedian_\infty + \parz) = \hessian_\infty \parz + O(\norm{\parz}{2}^2)$.
  In particular, there exist a constant $A$
  such that for any $\parz\in \ball(0,1)$, 
  we have $\norm{\nabla \Loss_\infty (\geometricmedian_\infty + \parz) - \hessian_\infty \parz}{2} \leq A\norm{\parz}{2}^2$.

  Now consider an orthonormal eigenvector basis $\unitvector{1}, \ldots, \unitvector{d}$ of $H_\infty$, with respective eigenvalues $\eigenvalue_1, \ldots, \eigenvalue_d$. Note that since $H_\infty$ is symmetric, we know that such a basis exists. We then define
  \begin{equation}
    \lambda_{min} \triangleq \inf_{\parz \in \ball(\geometricmedian_\infty,1)} \min Sp(\nabla^2 \Loss_\infty(\parz)),
  \end{equation}
  the minimum eigenvalue of the Hessian matrix $\nabla^2 \Loss_\infty (\parz)$ over the closed ball $\ball(\geometricmedian_\infty,1)$. 
  Note that using the same argument as Proposition \ref{porp:positive difinite}, we can say $\nabla^2 \Loss_\infty (\parz)$ is continuous and positive definite for all $\parz \in \ball(\geometricmedian_\infty,1)$, therefore, $\lambda_{min}$ is strictly positive.
  Now for any $i \in [d]$, $j \in \set{-1,1}$, and $0<\varepsilon < 1 $, we know that
  \begin{equation}
      \norm{\nabla \Loss_\infty (\geometricmedian_\infty + j\varepsilon \unitvector{i}) - \lambda_i j \varepsilon \unitvector{i}}{2} = \norm{\nabla \Loss_\infty (\geometricmedian_\infty + j\varepsilon \unitvector{i}) - \hessian_\infty j\varepsilon \unitvector{i}}{2}
      \leq A \norm{j\varepsilon \unitvector{i}}{2}^2
      = A \varepsilon^2.
  \end{equation}
  Now define $\eta \triangleq \min \set{\frac{1-2\rc}{4\rc}, 1}$.
  Since $0 < \rc < 1/2$, we clearly have $\eta > 0$.
  Moreover, for $\varepsilon <1$, since $2 \geq 1+\eta$, we have $\varepsilon^2 \leq \varepsilon^{1+\eta}$. 
  Therefore, $\norm{\nabla \Loss_\infty (\geometricmedian_\infty + j\varepsilon \unitvector{i}) - \hessian_\infty j\varepsilon \unitvector{i}}{2} \leq A \varepsilon^{1+\eta}$.
  For any voter $\voter \in [\VOTER]$, we then define the random unit vector
  \begin{equation}
    \randomvar_{ij \voter} \triangleq \frac{\paramsub{\voter} - \geometricmedian_\infty - j\varepsilon \unitvector{i}}{\norm{\paramsub{\voter} - \geometricmedian_\infty - j\varepsilon \unitvector{i}}{2}}.
  \end{equation}
  Note that, since $\paramdistribution$ is absolutely continuous with respect to the Lebesgue measure (Assumption \ref{ass:pdf}), all vectors $\randomvar_{ij \voter}$'s are well-defined with probability 1.
  By the definition of $\Loss_{1:V}$ and $\Loss_\infty$, we then have
  \begin{equation}
       \nabla \Loss_{1:V} (\geometricmedian_\infty + j\varepsilon \unitvector{i}) = \frac{1}{V}\sum_{\voter = 1}^{V} \randomvar_{ij \voter} 
       \quad\text{and}\quad \nabla \Loss_\infty(\geometricmedian_\infty + j\varepsilon \unitvector{i}) = \expect_{\paramsub{\voter}}{[\randomvar_{ij \voter}]}.
  \end{equation}
  Thus, for all $k \in [d]$, $\oneDelement{\nabla \Loss_{1:V} (\geometricmedian_\infty + j\varepsilon \unitvector{i})}{k}$ is just the average of $\VOTER$ i.i.d. random variables within the range $[-1,1]$, 
  and whose expectation is equal to $\oneDelement{\nabla \Loss_\infty(\geometricmedian_\infty + j\varepsilon \unitvector{i})}{k}$. 
  Therefore, by Chernoff bound, defining the event $\event_{ijk} (t) \triangleq \set{\absv{\oneDelement{\nabla \Loss_{1:V} (\geometricmedian_\infty + j\varepsilon \unitvector{i})}{k}- \oneDelement{\nabla \Loss_\infty(\geometricmedian_\infty + j\varepsilon \unitvector{i})}{k}}\leq t}$ for every $t>0$, 
  we obtain $\probability{\event_{ijk} (t)} \geq 1- 2 \exp{(-t^2\VOTER/2)}$.
  Defining $\error_{ij} = \nabla \Loss_{1:V} (\geometricmedian_\infty + j\varepsilon \unitvector{i}) - \lambda_ij\varepsilon \unitvector{i}$, 
  under event $\event_{ijk} \left( A \varepsilon^{1+\eta} \right)$, 
  by triangle inequality, we obtain 
  \begin{align}
      \absv{\oneDelement{\error_{ij}}{k}} 
      &\leq \absv{\nabla \Loss_{1:V} (\geometricmedian_\infty + j\varepsilon \unitvector{i})[k] - \nabla \Loss_{\infty} (\geometricmedian_\infty + j\varepsilon \unitvector{i})[k]}
      + \absv{\nabla \Loss_{\infty} (\geometricmedian_\infty + j\varepsilon \unitvector{i})[k] - \lambda_ij\varepsilon \unitvector{i}[k]} \\
      &\leq A \varepsilon^{1+\eta} 
      + \norm{\nabla \Loss_{\infty} (\geometricmedian_\infty + j\varepsilon \unitvector{i}) - \lambda_ij\varepsilon \unitvector{i}}{2} 
      \leq 2 A \varepsilon^{1+\eta}.
  \end{align}
  Denoting $\event^*$ the event where such inequalities hold for all $i,k \in [d]$ and $j \in \set{-1,1}$, and using union bound, we have
  \begin{align}
    \probability{\event^*}
    &\geq \probability{ \bigcap_{i \in [d]} \bigcap_{j \in \set{-1,1}} \bigcap_{k \in [d]} \event_{ijk} \left( A \varepsilon^2 \right)} 
    = \probability{ \neg \bigcup_{i \in [d]} \bigcup_{j \in \set{-1,1}} \bigcup_{k \in [d]} \neg \event_{ijk} \left( A \varepsilon^2 \right)}  \\
    &\geq 1 - \sum_{i \in [d]} \sum_{j \in \set{-1,1}} \sum_{k \in [d]} \probability{\neg \event_{ijk} \left( A \varepsilon^{1+\eta} \right) }   
    \geq 1-4d^2\exp{ \left( - A^2 \varepsilon^{2+2\eta} \VOTER/2 \right)}.
  \end{align}
  Now note that by Proposition \ref{prop:uniqueness}, we know that $\Loss_{1:\VOTER}$ is convex. Therefore, for any $i \in [d]$ and $j\in\set{-1,1}$, using the fact that $\geometricmedian_{1:\VOTER}$ minimizes $\Loss_{1:\VOTER}$, we have
  \begin{equation}
       (\geometricmedian_{1:V}-\geometricmedian_\infty - j\varepsilon \unitvector{i})^T\nabla \Loss_{1:V} (\geometricmedian_\infty + j\varepsilon \unitvector{i}) = (\geometricmedian_{1:V}-\geometricmedian_\infty - j\varepsilon \unitvector{i})^T(\lambda_i j \varepsilon \unitvector{i}+\error_{ij})  \leq 0.
  \end{equation}
  Rearranging the terms and noting that $\lambda_i>0$ then yields
  \begin{equation}
      (\geometricmedian_{1:V}-\geometricmedian_\infty)^T \left(j\unitvector{i}+\frac{\error_{ij}}{\varepsilon \lambda_i}\right) 
      \leq (j\varepsilon\unitvector{i})^T (j\unitvector{i}+\frac{\error_{ij}}{\varepsilon \lambda_i})
      = \varepsilon + \frac{j}{\lambda_i} \unitvector{i}^T\error_{ij}
      = \varepsilon + \frac{j}{\lambda_i} \oneDelement{\error_{ij}}{i}.
  \end{equation}
  Now define $\varepsilon_0 \triangleq \left(\lambda_{min} / 4dA\right)^{1/\eta}$.
  Under $\event^*$, for $\varepsilon \leq \varepsilon_0$, this then implies 
  $\norm{\error_{ij}}{\infty} \leq 2A \varepsilon^{1+\eta} \leq 2A \varepsilon \varepsilon_0^\eta = \frac{\varepsilon \lambda_{min}}{2d} \leq \frac{\varepsilon \lambda_{i}}{2d}$,
  For every $i \in [d]$ and $j\in\set{-1,1}$, we then have
  \begin{equation}
  \label{equ:condition}
      (\geometricmedian_{1:V}-\geometricmedian_\infty)^T \left( j\unitvector{i}+\frac{\error_{ij}}{\varepsilon \lambda_i} \right) 
      \leq \varepsilon + \frac{1}{\lambda_i} \norm{\error_{ij}}{\infty}
      \leq \varepsilon \left( 1 + \frac{1}{2d} \right)
      \leq \frac{3\varepsilon}{2}.
  \end{equation}
  Now denote $C \triangleq \norm{\geometricmedian_{1:V}-\geometricmedian_\infty}{\infty}$. 
  Thus, there exist $i \in [d]$ and $j \in \set{-1,1}$ such that $(\geometricmedian_{1:V}-\geometricmedian_\infty)[i] = jC$.
  We then obtain the lower bound
  \begin{align}
     (\geometricmedian_{1:V}-\geometricmedian_\infty)^T & \left( j\unitvector{i}+\frac{\error_{ij}}{\varepsilon \lambda_i} \right) 
     = C + \frac{(\geometricmedian_{1:V}-\geometricmedian_\infty)^T \error_{ij}}{\varepsilon \lambda_i}
     \geq C - \frac{ \norm{\geometricmedian_{1:V}-\geometricmedian_\infty}{2} \norm{\error_{ij}}{2}}{\varepsilon \lambda_i} \\
     \label{eq:lower_bound_close_median_2}
     &\geq C - \frac{ d \norm{\geometricmedian_{1:V}-\geometricmedian_\infty}{\infty} \norm{\error_{ij}}{\infty}}{\varepsilon \lambda_i} 
     \geq C - \frac{C}{2} 
     = \frac{\norm{\geometricmedian_{1:V}-\geometricmedian_\infty}{\infty}}{2},
  \end{align}
  where we used $\norm{\parx}{2} = \sqrt{\sum \parx[k]^2} \leq \sqrt{d \norm{\parx}{\infty}^2} = \sqrt{d} \norm{\parx}{\infty}$ 
  and the fact that $\norm{\error_{ij}}{\infty} \leq \frac{\varepsilon \lambda_{i}}{2d}$. 
  Combining Equations (\ref{equ:condition}) and (\ref{eq:lower_bound_close_median_2}) then yields, under $\event^*$, 
  the bound $\norm{\geometricmedian_{1:V}-\geometricmedian_\infty}{2} \leq \sqrt{d} \norm{\geometricmedian_{1:V}-\geometricmedian_\infty}{\infty} \leq 3 \varepsilon \sqrt{d}$.

  Now note that if $\VOTER \geq (3\sqrt{d}\varepsilon_0)^{-1/\rc}$, then we have $\varepsilon_\VOTER \triangleq \VOTER^{-\rc} / 3 \sqrt{d} \leq \varepsilon_0$.
  Thus, under $\event^*$ defined with $\varepsilon_\VOTER$, 
  the previous argument applies, which implies $\norm{\geometricmedian_{1:V}-\geometricmedian_\infty}{2} \leq \VOTER^{-\rc}$, as required by the lemma.
  
  Now take $\VOTER \geq V_2(\delta_2) \triangleq \max \left\{
  \left( \frac{ 2 (9d)^{1+\eta} }{A^2} \ln{ \frac{4d^2}{\delta_2} }\right)^{\frac{1}{1-2\rc -2\eta \rc}},
  (3\sqrt{d}\varepsilon_0)^{-1/\rc} \right\}$.
  By definition of $\eta$, we have $\eta \leq \frac{1-2\rc}{4\rc}$. 
  As a result, using also the assumption $\rc < 1/2$,
  we then have $1-2\rc-2\eta \rc \geq \frac{1-2\rc}{2} > 0$.
  It then follows that $\VOTER^{1-2\rc-2\eta} \geq \frac{ 2 (9d)^{1+\eta} }{A^2} \ln{ \frac{4d^2}{\delta_2} }$.
  We then have 
  \begin{align}
    \probability{\event^*} 
    &\geq 1 - 4d^2\exp{ \left( -A^2 \varepsilon_\VOTER^{2+2\eta} \VOTER/2 \right)} 
    = 1 - 4d^2\exp{ \left( - \frac{A^2 \VOTER^{1-2\rc -2\eta \rc}}{ 2 (9d)^{1+\eta} } \right)} 
    \geq 1-\delta_2,
  \end{align}
  which is what was needed for the lemma.
\end{proof}

%% file: asymptotic_hessian.tex
\paragraph{\textbf{Approximation of the Infinite Hessian Matrix.}}
To apply Lemma \ref{lemma:convex_projection_strategyproof}, we need to control the values of the Hessian matrix of $\Loss_{1:\VOTER}$.
In this section, we show that, similar to the finite-voter geometric median $\geometricmedian_{1:\VOTER}$, which is now known to be close to the infinite geometric median $\geometricmedian_\infty$, the Hessian matrix is close to the infinite Hessian matrix $\hessian_\infty$ at the infinite geometric median $\geometricmedian_\infty$.

\begin{lemma}
\label{lemma:close_hessian}
  Under Assumption \ref{ass:pdf}, 
  for $0 < 2\ra < \rb$, 
  for any $\varepsilon_3, \delta_3 >0$, 
  there exists $\VOTER_3(\varepsilon_3, \delta_3) \in \setN$ such that, 
  for all $\VOTER \geq \VOTER_3(\varepsilon_3, \delta_3)$, 
  with probability at least $1-\delta_3$, 
  there is no vote in the ball $\ball(\geometricmedian_{\infty}, \VOTER^{-\ra})$ and, 
  for all $\parz \in \ball(\geometricmedian_{\infty}, \VOTER^{-\rb})$, 
  we have $\norm{\nabla^2 \Loss_{1:\VOTER} (\parz) - \hessian_\infty}{\infty} \leq \varepsilon_3$.
\end{lemma}

Before proving Lemma \ref{lemma:close_hessian}, we first start with an observation about unit vectors.

\begin{lemma}
\label{lemma:unit_vector_bound}
  For any $0<\ra<\rb$, if $\norm{\parz}{2} \geq \VOTER^{-\ra}$ and $\norm{\rho}{2} \leq \VOTER^{-\rb}$, then for any $i\in[d]$, we have
  \begin{equation}
    \absv{\oneDelement{\unitvector{z}}{i}-\oneDelement{\unitvector{z+\rho}}{i}} = \mathcal{O}( \VOTER^{\ra-\rb})
  \end{equation}
\end{lemma}

\begin{proof}
  We have the inequalities
  \begin{align}
    \absv{\oneDelement{\unitvector{z}}{i}-\oneDelement{\unitvector{z+\rho}}{i}} 
    &= \absv{\frac{z[i]}{\norm{z}{2}} - \frac{(z+\rho)[i]}{\norm{z+\rho}{2}}} 
    = \absv{\frac{\norm{z+\rho}{2}\oneDelement{z}{i}-\norm{z}{2}\oneDelement{(z+\rho)}{i}}{\norm{z}{2}\norm{z+\rho}{2}}} \\
    &\leq \absv{\frac{(\norm{z+\rho}{2}-\norm{z}{2}) \absv{\oneDelement{z}{i}} -\norm{z}{2}\oneDelement{\rho}{i}}{\norm{z}{2} (\norm{z}{2}-\norm{\rho}{2})}} \\
    &\leq \absv{\frac{\norm{z+\rho}{2}-\norm{z}{2}}{\norm{z}{2}-\norm{\rho}{2}}} + \absv{\frac{\oneDelement{\rho}{i}}{\norm{z}{2}-\norm{\rho}{2}}} \\
    &\leq \frac{2\norm{\rho}{2}}{\norm{z}{2}-\norm{\rho}{2}},
  \end{align}
  where we used the fact that $\norm{z+\rho}{2}-\norm{z}{2}\leq\norm{\rho}{2}$. We then have
  \begin{equation}
    \absv{\oneDelement{\unitvector{z}}{i}-\oneDelement{\unitvector{z+\rho}}{i}} \leq \frac{2\VOTER^{-\rb}}{\VOTER^{-\ra}-\VOTER^{-\rb}}\leq \frac{2\VOTER^{\ra-\rb}}{1-\VOTER^{\ra-\rb}} = \mathcal{O} (\VOTER^{\ra-\rb}),
  \end{equation}
  which is the lemma.
\end{proof}

We now move on to the proof of Lemma \ref{lemma:close_hessian}.

\begin{proof}[Proof of Lemma \ref{lemma:close_hessian}]
  Applying Lemma \ref{lemma:no_voter} shows that for $\VOTER \geq \VOTER_1(\delta_3/2)$, under an event $\event_{no-voter}$ that holds with probability at least $1-\delta_3/2$, 
  the ball $\ball(\geometricmedian_{\infty}, \VOTER^{-\ra})$ contains no voters' preferred vectors.

  For any voter $\voter \in [\VOTER]$, and any $i,j \in [d]$, we define
  \begin{equation}
    \helement_{ij\voter} \triangleq \twoDelement{\nabla^2 \ell_2(\geometricmedian_\infty-\param_\voter)}{i}{j} 
    = \frac{\twoDelement{(I-\unitvector{\geometricmedian_\infty-\param_\voter} \unitvector{\geometricmedian_\infty-\param_\voter}^T)}{i}{j}}{\norm{{\geometricmedian_\infty-\param_\voter}}{2}}.
  \end{equation}
  Since $\paramdistribution$ is absolutely continuous with respect to the Lebesgue measure, 
  we know that $\helement_{ij\voter}$ is well-defined with probability 1.
  We then have
  \begin{equation}
    \twoDelement{\nabla^2 \Loss_{1:\VOTER}(\geometricmedian_\infty)}{i}{j} = \frac{1}{V}\sum_{i = 1}^{V} \helement_{ij \voter} \quad\text{and}\quad \twoDelement{\hessian_\infty}{i}{j} = \twoDelement{\nabla^2 \Loss_\infty (\geometricmedian_\infty)}{i}{j} = \expect_{\param_\voter}{[\helement_{ij\voter}]}.
  \end{equation}
  Moreover, we can upper-bound the variance of $\helement_{ij\voter}$ by
  \begin{equation}
    \Var[a_{ij \voter}] \leq \expect_{\param_\voter}[a_{ij \voter}^2]
    = \int_\PARAM \left( \frac{(I-\unitvector{\geometricmedian_\infty-\param} \unitvector{\geometricmedian_\infty-\param}^T)[i,j]}{\norm{{\geometricmedian_\infty-\param}}{2}}\right)^2 p(\param) d\param \leq \int_\PARAM  \frac{1}{\norm{{\geometricmedian_\infty-\param}}{2}^2} p(\param) d\param.
  \end{equation}
  By Lemma \ref{lemma: bounded integral}, we know that this integral is bounded, thus, we have $\Var[a_{ij \voter}] < \infty$. We then define the maximal variance $\sigma^2 \triangleq \max_{i,j} \Var[a_{ij \voter}]$ of the elements of the Hessian matrix.
  Since the voters' preferred vectors are assumed to be i.i.d, we then obtain
  \begin{equation}
    \Var\left[  \twoDelement{\nabla^2 \Loss_{1:\VOTER}(\geometricmedian_\infty)}{i}{j} \right] = \frac{1}{\VOTER^2}\sum_{\voter = 1}^\VOTER \Var[a_{ij \voter}] \leq  \frac{\sigma^2}{\VOTER}.
  \end{equation}
  Now applying Chebyshev's inequality on $\nabla^2 \Loss_{1:\VOTER}(\geometricmedian_\infty)[i,j]$ yields
  \begin{equation}
    \probability{\absv{\twoDelement{\nabla^2 \Loss_{1:\VOTER}(\geometricmedian_\infty)}{i}{j} - \twoDelement{\hessian_\infty}{i}{j}}\geq \varepsilon_3/2} \leq \frac{4\Var\left[  \twoDelement{\nabla^2 \Loss_{1:\VOTER}(\geometricmedian_\infty)}{i}{j} \right]}{\varepsilon_3^2} \leq \frac{4\sigma^2}{\VOTER \varepsilon_3^2}.
  \end{equation}
  Using a union bound, we then obtain
  \begin{equation}
  \label{equ:hessian_bound}
    \probability{\exists i, j \in [d], \absv{\twoDelement{\nabla^2 \Loss_{1:\VOTER}(\geometricmedian_\infty)}{i}{j} - \twoDelement{\hessian_\infty}{i}{j}}\geq {\varepsilon_3/2}} \leq \frac{4d^2 \sigma^2}{\VOTER \varepsilon_3^2}.
  \end{equation}
  Therefore, taking $\VOTER \geq \frac{8 d^2 \sigma^2}{\delta_3 \varepsilon_3^2}$, 
  the event $\event_{Hessian} \triangleq \set{\forall i,j \in [d] \mathsep \absv{\twoDelement{\nabla^2 \Loss_{1:\VOTER}(\geometricmedian_\infty)}{i}{j} - \twoDelement{\hessian_\infty}{i}{j}} \leq \varepsilon_3/2}$ 
  occurs with probability at least $1 - \delta_3/2$.
  Taking a union bound shows that, 
  for $\VOTER \geq \max \set{\VOTER_1(\delta_3/2), \frac{8 d^2 \sigma^2}{\delta_3 \varepsilon_3^2}}$, 
  the event $\event \triangleq \event_{no-vote} \cap \event_{Hessian}$ occurs with probability at least $1-\delta_3$.

  We now bound the difference between finite-voter Hessian matrices at $\geometricmedian_\infty$ and at a close point $\parz$, by
  \begin{align}
    \VOTER &\absv{\twoDelement{\nabla^2 \Loss_{1:\VOTER}(\parz)}{i}{j}-\twoDelement{\nabla^2 \Loss_{1:\VOTER}(\geometricmedian_\infty)}{i}{j}} 
    = \absv{ \sum_{\voter \in [\VOTER]} \frac{\twoDelement{(I-\unitvector{\parz-\param_\voter} \unitvector{\parz-\param_\voter}^T)}{i}{j}}{\norm{{\parz-\param_\voter}}{2}} 
    - \frac{\twoDelement{(I-\unitvector{\geometricmedian_\infty-\param_\voter} \unitvector{\geometricmedian_\infty-\param_\voter}^T)}{i}{j}}{\norm{{\geometricmedian_\infty-\param_\voter}}{2}}}\\
    &\leq \sum_{\voter \in [\VOTER]} \absv{ \frac{\twoDelement{(I-\unitvector{\parz-\param_\voter} \unitvector{\parz-\param_\voter}^T)}{i}{j}}{\norm{{\parz-\param_\voter}}{2}} 
    - \frac{\twoDelement{(I-\unitvector{\geometricmedian_\infty-\param_\voter} \unitvector{\geometricmedian_\infty-\param_\voter}^T)}{i}{j}}{\norm{{\geometricmedian_\infty-\param_\voter}}{2}}} \\
    &\leq \sum_{\voter\in[\VOTER]} \absv{ \frac
    {\twoDelement{I}{i}{j}(\norm{\geometricmedian_\infty-\param_\voter}{2}-\norm{{\parz-\param_\voter}}{2})}
    {\norm{{\parz-\param_\voter}}{2}\norm{{\geometricmedian_\infty-\param_\voter}}{2}}} 
    + \absv{ \frac{
    \oneDelement{\unitvector{\parz-\param_\voter}}{i}\oneDelement{\unitvector{\parz-\param_\voter}}{j} }{\norm{\geometricmedian_\infty-\param_\voter}{2}} 
    - \frac{ \oneDelement{\unitvector{\geometricmedian_\infty-\param_\voter}}{i} \oneDelement{\unitvector{\geometricmedian_\infty-\param_\voter}}{j} }
    {\norm{\parz-\param_\voter}{2}} }.
  \end{align}
  Note that, under $\event$, for all voters $\voter \in [\VOTER]$, 
  we have $\norm{\geometricmedian_\infty - \paramsub{\voter}}{2} \geq \VOTER^{-\ra}$.
  Now assume $\parz \in \ball(\geometricmedian_\infty, \VOTER^{-\rb})$, 
  Lemma \ref{lemma:unit_vector_bound} then applies with $\rho \triangleq \norm{\parz - \geometricmedian_\infty}{2} \leq \VOTER^{-\rb}$, 
  yielding $\absv{\oneDelement{\unitvector{\geometricmedian_\infty-\param_\voter}}{i} - \oneDelement{\unitvector{\parz-\param_\voter}}{i}} = \mathcal{O}(\VOTER^{\ra-\rb}) \leq 1$ for all $i \in [d]$. 
  Also, we have $\absv{\unitvector{}[i]} \leq \norm{\unitvector{}}{2} = 1$ for all unit vectors.
  Under $\event$, we then have 
  \begin{align}
    &\absv{ \frac{
    \oneDelement{\unitvector{\parz-\param_\voter}}{i}\oneDelement{\unitvector{\parz-\param_\voter}}{j} }{\norm{\geometricmedian_\infty-\param_\voter}{2}} 
    - \frac{ \oneDelement{\unitvector{\geometricmedian_\infty-\param_\voter}}{i} \oneDelement{\unitvector{\geometricmedian_\infty-\param_\voter}}{j} }
    {\norm{\parz-\param_\voter}{2}} } 
    \leq \absv{ \frac{
    \oneDelement{\unitvector{\parz-\param_\voter}}{i}\oneDelement{\unitvector{\parz-\param_\voter}}{j} }{\norm{\geometricmedian_\infty-\param_\voter}{2}} 
    - \frac{ \oneDelement{\unitvector{\parz-\param_\voter}}{i} \oneDelement{\unitvector{\geometricmedian_\infty-\param_\voter}}{j} }
    {\norm{\geometricmedian_\infty-\param_\voter}{2}} } \nonumber \\
    &\qquad \qquad + \absv{ \frac{ \oneDelement{\unitvector{\parz-\param_\voter}}{i} \oneDelement{\unitvector{\geometricmedian_\infty-\param_\voter}}{j} }
    {\norm{\geometricmedian_\infty-\param_\voter}{2}}
    - \frac{ \oneDelement{\unitvector{\parz-\param_\voter}}{i} \oneDelement{\unitvector{\parz-\param_\voter}}{j} }
    {\norm{\geometricmedian_\infty-\param_\voter}{2}} }
    + \absv{ \frac{ \oneDelement{\unitvector{\parz-\param_\voter}}{i} \oneDelement{\unitvector{\parz-\param_\voter}}{j} }
    {\norm{\geometricmedian_\infty-\param_\voter}{2}}
    - \frac{ \oneDelement{\unitvector{\geometricmedian_\infty-\param_\voter}}{i} \oneDelement{\unitvector{\geometricmedian_\infty-\param_\voter}}{j} }
    {\norm{\parz-\param_\voter}{2}} } \\
    &\leq \mathcal{O}(\VOTER^{2 \ra-\rb}) 
    + \frac{ 2 \absv{\norm{\geometricmedian_\infty - \param_\voter}{2} - \norm{\parz-\param_\voter}{2}} }{ \norm{\geometricmedian_\infty - \param_\voter}{2}  \norm{\parz-\param_\voter}{2} }
    \leq \mathcal{O}(\VOTER^{2 \ra-\rb}) + \frac{2 \VOTER^{-\rb}}{ \VOTER^{-\ra} (\VOTER^{-\ra} - \VOTER^{-\rb}) }
    = \mathcal{O}(\VOTER^{2 \ra-\rb}),
  \end{align}
  where, in the last line, we used the triangle inequality, 
  which implies $\absv{\norm{\geometricmedian_\infty - \param_\voter}{2} - \norm{\parz-\param_\voter}{2}} \leq \norm{\geometricmedian_\infty - \parz}{2} \leq \VOTER^{-\rb}$
  and $\norm{\parz - \paramsub{\voter}}{2} \geq \norm{\geometricmedian_\infty - \paramsub{\voter}}{2} - \norm{\geometricmedian_\infty - \parz}{2} \geq \VOTER^{-\ra} - \VOTER^{-\rb}$.
  Therefore, under $\event$, we have
  \begin{equation}
      \absv{\twoDelement{\nabla^2 \Loss_{1:\VOTER}(\parz)}{i}{j}-\twoDelement{\nabla^2 \Loss_{1:\VOTER}(\geometricmedian_\infty)}{i}{j}} 
      \leq \frac{1}{\VOTER}\sum_{\voter \in [\VOTER]} \mathcal{O}(\VOTER^{2 \ra-\rb}) 
      = \mathcal{O}(\VOTER^{2\ra-\rb}).
  \end{equation}
  We now use the fact that $2\ra < \rb$, which implies $\VOTER^{2\ra-\rb} \rightarrow 0$.
  Thus, for any $\varepsilon_3 > 0$, there exists a $\VOTER'_3(\varepsilon_3)$ such that, 
  for $\VOTER \geq \VOTER'_3(\varepsilon_3)$, we have
  \begin{equation}
    \absv{\twoDelement{\nabla^2 \Loss_{1:\VOTER}(\parz)}{i}{j}-\twoDelement{\nabla^2 \Loss_{1:\VOTER}(\geometricmedian_\infty)}{i}{j}} \leq \varepsilon_3/2.
  \end{equation}
  Choosing $\VOTER\geq V_3(\varepsilon_3,\delta_3) \triangleq \max \set{ \VOTER_1(\delta_3/2), \frac{8d^2\sigma^2}{\delta_3\varepsilon_3^2}, \VOTER'_3(\varepsilon_3)}$, 
  and combining the above guarantee with the guarantee about event $\event$ proved earlier, yields the result.
\end{proof}

\input{asymptotic_third_derivative}

%% file: asymptotic_third_derivative.tex
\paragraph{\textbf{Third-derivative Approximation.}}
Finally, to apply Lemma \ref{lemma:convex_projection_strategyproof}, we also need to control the third-derivative of $\Loss_{1:\VOTER}$ near the geometric median $\geometricmedian_{1:\VOTER}$.
In fact, for our purposes, it will be sufficient to bound its norm by a possibly increasing function in $\VOTER$, as long as this function grows slower than $\VOTER$.

\begin{definition}
We denote $\threeDelement{\nabla^3 \Loss(\parz)}{i}{j}{k}$, the third derivative of $\Loss(\parz)$ with respect to $\oneDelement{\parz}{i}$, $\oneDelement{\parz}{j}$, and $\oneDelement{\parz}{k}$, 
and $\norm{\nabla^3 \Loss(\parz)}{\infty} = \max_{i,j,k} \absv{\threeDelement{\nabla^3 \Loss(\parz)}{i}{j}{k}}$.
\end{definition}

\begin{lemma}
\label{lemma:bounded_third_derivative}
  Under Assumption \ref{ass:pdf}, for $\ra,\rb>0$,
  there exists $K \in \setR$ such that, for any $\delta_4 >0$, 
  there exists $\VOTER_4(\delta_4) \in \setN$ such that, 
  for all $\VOTER \geq \VOTER_4(\delta_4)$, 
  with probability at least $1-\delta_4$, 
  no other voter's vector lies in the ball $\ball(\geometricmedian_{\infty}, \VOTER^{-\ra})$ and,
  for all $\parz \in \ball(\geometricmedian_{\infty}, \VOTER^{-\rb})$, we have
  $ \norm{\nabla^3 \Loss_{1:\VOTER} (\parz)}{\infty} = K (1 + \VOTER^{3 \ra - \rb})$.
\end{lemma}

\begin{proof}
  We use the same proof strategy as the previous Lemma. First note that using Lemma \ref{lemma: bounded integral} for $d\geq5$, the variance of each element of the third derivative of $\Loss_{1:\VOTER}$ is bounded, i.e.,
  \begin{equation}
      \forall (i,j,k) \in [d]^3, \Var \left[\threeDelement{\nabla^3 \Loss_{1:\VOTER}(\geometricmedian_\infty)}{i}{j}{k}\right] 
      = \mathcal O \left( \int_\PARAM  \frac{1}{\norm{{\geometricmedian_\infty-\param}}{2}^4} p(\param) d\param \right) 
      \triangleq \sigma^2 < \infty.
  \end{equation}
  Similarly to the previous proof, we define the events
  \begin{align}
    \event_{\nabla^3}(t) &\triangleq \set{\forall i,j,k \in [d], \absv{\threeDelement{\nabla^3\Loss_{1:\VOTER}(\geometricmedian_\infty)}{i}{j}{k}-\threeDelement{ \nabla^3\Loss_\infty(\geometricmedian_\infty)}{i}{j}{k}} 
    \leq t}, \\
    \event_{no-vote} &\triangleq \set{\forall \voter \in [\VOTER], \paramsub{\voter} \notin \ball(\geometricmedian_\infty, \VOTER^{-\rb})}
    \quad \text{and} \quad 
    \event \triangleq \event_{\nabla^3} (\VOTER^{r}) \cap \event_{no-vote},
  \end{align} 
  where $r \triangleq \max\set{0, 3\ra - \rb} \geq 0$.
  Using Chebyshev’s bound and union bound, we know that,
  $\probability{\event_{\nabla^3}(t)} \geq 1 - d^3 \sigma^2 / \VOTER t^2$,
  where the $\mathcal O$ hides a constant derived from the upper bound on the variance of $\threeDelement{\nabla^3 \Loss_{1:\VOTER}(\geometricmedian_\infty)}{i}{j}{k}$, and which depends only on $\paramdistribution$.
  Therefore, we have $\probability{ \event_{\nabla^3} (\VOTER^{r}) } \geq 1 - d^3 \sigma^2 \VOTER^{-(1+2r)}$.
  Now, assuming $\VOTER \geq \VOTER_4 (\delta_4) \triangleq \max \set {\left(2 d^3 \sigma^2 / \delta\right)^{\frac{1}{1+2r}}, \VOTER_1(\delta_4/2)}$,
  we know that the event $\event$ occurs with probability at least $1-\delta_4$.

  Now, we bound the deviation of $\nabla^3\Loss_{1:V}(\parz)$ from $\nabla^3\Loss_{1:V}(\geometricmedian_\infty)$ for any $\parz \in  \ball(\geometricmedian_{\infty}, \VOTER^{-\rb})$.
  It can be shown that $\threeDelement{\nabla^3 \ell_2(\parz)}{i}{j}{k} = \frac{\threeDelement{f(z)}{i}{j}{k}]}{\norm{\parz}{2}^2}$, where
  \begin{equation}
      \threeDelement{f(z)}{i}{j}{k} \triangleq
      \begin{cases}
        {3\oneDelement{\unitvector{\parz}}{i}^3-3\oneDelement{\unitvector{\parz}}{i}}, & \text{if}\ i = j = k \\
        {3\oneDelement{\unitvector{\parz}}{j}^2\oneDelement{\unitvector{\parz}}{i}-\oneDelement{\unitvector{\parz}}{i}}, & \text{if}\ i \neq j = k \\
        {3\oneDelement{\unitvector{\parz}}{i}\oneDelement{\unitvector{\parz}}{j}\oneDelement{\unitvector{\parz}}{k}}, & \text{if}\ i \neq j \neq k
      \end{cases}.
  \end{equation}
  Since $\absv{\unitvector{}[i]} \leq 1$ for all unit vectors $\unitvector{}$ and all coordinates $i \in [d]$, 
  we see that $\absv{\threeDelement{f(\geometricmedian_\infty-\param_\voter)}{i}{j}{k}} \leq 6$.
  Moreover, using Lemma \ref{lemma:unit_vector_bound}, for any $i,j,k \in [d]$, we have
  \begin{equation}
      \absv{\threeDelement{f(\parz)}{i}{j}{k]} - \threeDelement{f(\geometricmedian_\infty)}{i}{j}{k}} = \mathcal O (\VOTER^{\ra-\rb}).
  \end{equation}
  Recall also that, like in the previous proof, 
  under event $\event$, for all voters $\voter \in [\VOTER]$, 
  we have $\norm{\geometricmedian_\infty - \param_\voter}{2} \geq \VOTER^{-\ra}$, 
  $\norm{\parz-\param_\voter}{2} \geq \VOTER^{-\ra} - \VOTER^{-\rb} = \Omega(\VOTER^{-\ra})$
  and $\absv{\norm{\geometricmedian_\infty - \param_\voter}{2} - \norm{\parz-\param_\voter}{2}} \leq \norm{\geometricmedian_\infty - \parz}{2} \leq \VOTER^{-\rb}$ (using the triangle inequality).
  We then have
  \begin{align}
      &\absv{\threeDelement{\nabla^3\Loss_{1:V}(\parz)}{i}{j}{k} - \threeDelement{\nabla^3\Loss_{1:V}(\geometricmedian_\infty)}{i}{j}{k}}
      =\absv{ \sum_{\voter \in [\VOTER]} \frac{\threeDelement{f(\parz-\param_\voter)}{i}{j}{k}}{\norm{\parz-\param_\voter}{2}^2} 
      - \sum_{\voter \in [\VOTER]} \frac{\threeDelement{f(\geometricmedian_\infty-\param_\voter)}{i}{j}{k}}{\norm{\geometricmedian_\infty-\param_\voter}{2}^2}}  \\
      &\leq \frac{1}{\VOTER} \sum_{\voter \in [\VOTER]} \absv{ 
      \frac{\threeDelement{f(\parz-\param_\voter)}{i}{j}{k}}{\norm{\parz-\param_\voter}{2}^2} 
      - \frac{\threeDelement{f(\geometricmedian_\infty -\param_\voter)}{i}{j}{k}}{\norm{\parz-\param_\voter}{2}^2}}
      + \absv{ 
      \frac{\threeDelement{f(\geometricmedian_\infty-\param_\voter)}{i}{j}{k}}{\norm{\parz-\param_\voter}{2}^2} 
      - \frac{\threeDelement{f(\geometricmedian_\infty -\param_\voter)}{i}{j}{k}}{\norm{\geometricmedian_\infty-\param_\voter}{2}^2}} \\
      &\leq \frac{1}{\VOTER} \sum_{\voter \in [\VOTER]} \frac{\mathcal O (\VOTER^{\ra-\rb})}
      {\norm{\parz-\param_\voter}{2}^2} 
      + 6 \frac{\absv{\norm{\parz-\param_\voter}{2}^2 - \norm{\geometricmedian_\infty-\param_\voter}{2}^2}}{\norm{\parz-\param_\voter}{2}^2 \norm{\geometricmedian_\infty-\param_\voter}{2}^2} \\
      &\leq \frac{\mathcal O (\VOTER^{\ra-\rb})}{ \Omega \left( \VOTER^{-\ra} \right)^2 }
      + \frac{6}{\VOTER} \sum_{\voter \in [\VOTER]} \absv{\norm{\parz-\param_\voter}{2} - \norm{\geometricmedian_\infty-\param_\voter}{2}} 
      \frac{ \norm{\parz-\param_\voter}{2} + \norm{\geometricmedian_\infty-\param_\voter}{2} }
      { \norm{\parz-\param_\voter}{2}^2 \norm{\geometricmedian_\infty-\param_\voter}{2}^2 } \\
      &\leq \mathcal O (\VOTER^{3\ra-\rb})
      + \frac{6}{\VOTER} \sum_{\voter \in [\VOTER]} \VOTER^{-\rb} \left( \frac{1}{\norm{\parz-\param_\voter}{2} \norm{\geometricmedian_\infty-\param_\voter}{2}^2} 
      + \frac{1}{\norm{\parz-\param_\voter}{2}^2 \norm{\geometricmedian_\infty-\param_\voter}{2}} \right) \\
      &\leq \mathcal O (\VOTER^{3\ra-\rb})
      + \frac{12 \VOTER^{-\rb}}{ \Omega \left( \VOTER^{-\ra} \right)^3}
      \leq \mathcal O (\VOTER^{3\ra - \rb}).
  \end{align}
  Combining this with the guarantee of event $\event$ then yields
  \begin{align}
    \absv{ \threeDelement{\nabla^3\Loss_{1:V}(\parz)}{i}{j}{k} }
    &\leq \absv{\threeDelement{\nabla^3\Loss_{\infty}(\geometricmedian_\infty)}{i}{j}{k}}
    + \absv{\threeDelement{\nabla^3\Loss_{\infty}(\geometricmedian_\infty)}{i}{j}{k} - \threeDelement{\nabla^3\Loss_{1:V}(\geometricmedian_\infty)}{i}{j}{k}} \nonumber \\
    &\qquad \qquad \qquad \qquad \qquad + \absv{ \threeDelement{\nabla^3\Loss_{1:V}(\geometricmedian_\infty)}{i}{j}{k} - \threeDelement{\nabla^3\Loss_{1:V}(\parz)}{i}{j}{k}} \\
    &\leq \mathcal{O}(1) + \mathcal{O} (\VOTER^{r}) + \mathcal{O} (\VOTER^{3\ra - \rb}) 
    = \mathcal{O}(1) + \mathcal{O} (\VOTER^{3\ra - \rb}),
  \end{align}
  using the definition of $r$. 
  Given that $\probability{\event} \geq 1-\delta_4$, taking a bound $K$ that can replace the $\mathcal O$ yields the lemma.
\end{proof}

%% file: asymptotic_final_proof.tex
\subsubsection{Proof of Theorem \ref{th:asymptotic_strategyproofness}}
\label{sec:proof_th_asymptotic_strategyproofness}
Finally, we can prove Theorem \ref{th:asymptotic_strategyproofness}.

\begin{proof}[Proof of Theorem \ref{th:asymptotic_strategyproofness}]
   Let $\varepsilon,\delta>0$. 
   Choose $\ra, \rc, \rb$ such that\footnote{clearly for $d\geq5$ such $\ra, \rc, \rb$ exist} $2/d<2\ra<\rb<\rc<1/2$, 
   and set $\delta_1 \triangleq \delta_2 \triangleq \delta_3 \triangleq \delta_4 \triangleq \delta/4$.
   Define also $\lambda_{min} \triangleq \min \spectrum (\nabla^2 \Loss_\infty (\geometricmedian_\infty))$ and $\varepsilon_3 \triangleq \min \set{\lambda_{min} / 2, \varepsilon_5}$, 
   where $\varepsilon_5$ will be defined later on, based on the continuity of $\skewness$ at $\hessian_\infty$.
   
   Now consider $\VOTER$ sufficiently large to satisfy the requirements of lemmas \ref{lemma:no_voter}, \ref{lemma:close_median}, \ref{lemma:close_hessian}, and \ref{lemma:bounded_third_derivative}. 
   Denoting $\event_\VOTER$ the event that contains the intersection of the guarantees of these lemmas, by union bound, we then know that $\probability{\event_\VOTER} \geq 1-\delta$.
   We will now show that, for $\VOTER$ large enough, under $\event_\VOTER$, 
   the geometric median restricted to the first $1+\VOTER$ voters is $(\skewness(\hessian_\infty) + \varepsilon)$-strategyproof for voter $0$.
   To do so, it suffices to prove that, under $\event_\VOTER$, the assumptions of Lemma \ref{lemma:convex_projection_strategyproof} are satisfied, for $\beta \triangleq 2 / \lambda_{min} \VOTER$.
   
   First, let us show that for $\VOTER$ large enough, under $\event_\VOTER$, the ball $\ball(\geometricmedian_{1:\VOTER}, 2 \beta)$ contains no preferred vector from the first $\VOTER$ voters.
   To prove this, let $\parz \in \ball(\geometricmedian_{1:\VOTER}, \beta)$.
   By triangle inequality, we have $\norm{\parz - \geometricmedian_\infty}{2} \leq \norm{\parz - \geometricmedian_{1:\VOTER}}{2} + \norm{\geometricmedian_{1:\VOTER} - \geometricmedian_\infty}{2} \leq 2 \beta + \VOTER^{-\rc} 
   = \mathcal O(\VOTER^{-1} + \VOTER^{-\rc}) = o(\VOTER^{-\ra})$,
   since $\ra < \rc < 1$.
   Thus, for $\VOTER$ large enough, we have $\parz \in \ball(\geometricmedian_\infty, \VOTER^{-\ra})$.
   But by Lemma \ref{lemma:no_voter}, under $\event_\VOTER$, this ball contains none of the preferred vectors from the first $\VOTER$ voters.
   As a corollary, $\Loss_{1:\VOTER}$ is then infinitely differentiable in $\ball(\geometricmedian_{1:\VOTER}, 2 \beta)$.
   The first condition of Lemma \ref{lemma:convex_projection_strategyproof} thus holds.
   
   We now move on to the second condition. 
   Note that, under event $\event_\VOTER$, by virtue of Lemma \ref{lemma:close_hessian},
   for all $\parz \in \ball(\geometricmedian_{1:\VOTER}, \beta) \subset \ball(\geometricmedian_{1:\VOTER}, 2\beta) \subset \ball (\geometricmedian_\infty, \VOTER^{-\ra})$,
   we have $\norm{\nabla^2 \Loss_{1:\VOTER} (\parz) - \hessian_\infty}{\infty} \leq \varepsilon_3 = \lambda_{min} / 2$.
   Lemma \ref{lemma:min_spectrum_addition} then yields 
   $\min \spectrum(\nabla^2 \Loss_{1:\VOTER} (\parz)) \geq \min \spectrum(\hessian_\infty) - \norm{\nabla^2 \Loss_{1:\VOTER} (\parz) - \hessian_\infty}{\infty} 
   \geq \lambda_{min} - \lambda_{min}/2 = \lambda_{min}/2$.
   By Taylor's theorem, we then know that, for any unit vector $\unitvector{}$,
   there exists $\parz \in [\geometricmedian_{1:\VOTER}, \geometricmedian_{1:\VOTER} + \beta \unitvector{}]$ such that
   \begin{align}
     \unitvector{}^T \nabla \Loss_{1:\VOTER} (\geometricmedian_{1:\VOTER} + \beta \unitvector{})
     &= \unitvector{}^T \left( \nabla \Loss_{1:\VOTER} (\geometricmedian_{1:\VOTER}) + \nabla^2 \Loss_{1:\VOTER} (\parz) [\beta \unitvector{}] \right) 
     = \beta \nabla^2 \Loss_{1:\VOTER} (\parz) [\unitvector{} \otimes \unitvector{}] \\
     &\geq \beta \min \spectrum(\nabla^2 \Loss_{1:\VOTER} (\parz)) 
     \geq \frac{2}{\lambda_{min} \VOTER} \lambda_{min} = 2/\VOTER > 1/\VOTER,
   \end{align}
   where we used the fact that $\nabla \Loss_{1:\VOTER} (\geometricmedian_{1:\VOTER}) = 0$.
   Thus the second condition of Lemma \ref{lemma:convex_projection_strategyproof} holds too.

  We now move on to the third condition.
  We have already shown that, under $\event_\VOTER$ and for all $\parz \in \ball(\geometricmedian_{1:\VOTER}, \beta)$, we have $\min \spectrum(\nabla^2 \Loss_{1:\VOTER}(\parz)) \geq \lambda_{min} / 2$. 
  From this, it follows that, for all $\parz \in \ball(\geometricmedian_{1:\VOTER}, \beta)$, we have $\min \spectrum(\nabla^2 \Loss_{1:\VOTER}(\parz) \cdot \nabla^2 \Loss_{1:\VOTER}(\parz)) \geq \lambda_{min}^2 / 4$.
  But now note that, for any coordinates $i,j \in [d]$, we have
  \begin{align}
    &\absv{\nabla^3 \Loss_{1:\VOTER}(\parz) \cdot \nabla \Loss_{1:\VOTER}(\parz) [i,j]}
    = \absv{\sum_{k \in [d]} \nabla^3 \Loss_{1:\VOTER}(\parz) [i, j, k] \nabla \Loss_{1:\VOTER}(\parz) [k]} \\
    &\leq \sum_{k \in [d]} \absv{\nabla^3 \Loss_{1:\VOTER}(\parz) [i, j, k]} \absv{\nabla \Loss_{1:\VOTER}(\parz) [k]} 
    \leq d \norm{\nabla^3 \Loss_{1:\VOTER}(\parz)}{\infty} \norm{\nabla \Loss_{1:\VOTER}(\parz)}{\infty} \\
    &\leq K d (1+\VOTER^{3\ra - \rb}) \beta 
    = \mathcal O (\VOTER^{-1} + \VOTER^{3 \ra -\rb -1}).
  \end{align}
  But since $2\ra < \rb < 1/2$, we have $3 \ra - \rb - 1 = \ra - 1 < -1/2 < 0$. 
  Thus the bound above actually goes to zero, as $\VOTER \rightarrow \infty$.
  In particular, for $\VOTER$ large enough, we must have $\norm{\nabla^3 \Loss_{1:\VOTER}(\parz) \cdot \nabla \Loss_{1:\VOTER}(\parz)}{\infty} \leq \lambda_{min}^2 /8$.
  As a result, by Lemma \ref{lemma:min_spectrum_addition}, for all $\parz \in \ball(\geometricmedian_{1:\VOTER}, \beta)$ and under $\event_\VOTER$, we then have 
  \begin{align}
     \min \spectrum &\left( \nabla^2 \Loss_{1:\VOTER}(\parz) \cdot \nabla^2 \Loss_{1:\VOTER}(\parz) + \nabla^3 \Loss_{1:\VOTER}(\parz) \cdot \nabla \Loss_{1:\VOTER}(\parz) \right) \\
     &\geq \min \spectrum(\nabla^2 \Loss_{1:\VOTER}(\parz) \cdot \nabla^2 \Loss_{1:\VOTER}(\parz)) - \norm{\nabla^3 \Loss_{1:\VOTER}(\parz) \cdot \nabla \Loss_{1:\VOTER}(\parz)}{\infty} \\
     &\geq \frac{\lambda_{min}^2}{4} - \frac{\lambda_{min}^2}{8} \geq \frac{\lambda_{min}^2}{8}.
  \end{align}
  Therefore $\nabla^2 \Loss_{1:\VOTER}(\parz) \cdot \nabla^2 \Loss_{1:\VOTER}(\parz) + \nabla^3 \Loss_{1:\VOTER}(\parz) \cdot \nabla \Loss_{1:\VOTER}(\parz) \succ 0$, 
  which is the third condition of Lemma~\ref{lemma:convex_projection_strategyproof}.

  Finally, the fourth and final condition of Lemma~\ref{lemma:convex_projection_strategyproof} holds by continuity of the function $\skewness$ (Lemma \ref{lemma:continuous_skewness}).
  More precisely, since $\hessian_\infty \succ 0$, we know that $\skewness$ is continuous in $\hessian_\infty$.
  Thus, there exists $\varepsilon_5 > 0$ such that, 
  if $A$ is a symmetric matrix with $\norm{\hessian_\infty - A}{\infty} \leq \varepsilon_5$,
  then $A \succ 0$ and $\skewness (A) \leq \skewness(\hessian_\infty) + \varepsilon$.
  Yet, by definition of $\varepsilon_3$ and Lemma \ref{lemma:close_hessian}, 
  we know that all hessian matrices $\nabla^2 \Loss_{1:\VOTER} (\parz)$ for $\parz \in \ball(\geometricmedian_{1:\VOTER}, \beta)$ satisfy the above property.
  Therefore, we know that for all such $\parz$, we have $\skewness (\nabla^2 \Loss_{1:\VOTER} (\parz)) \leq \skewness(\hessian_\infty) + \varepsilon$, 
  which is the fourth condition of Lemma \ref{lemma:convex_projection_strategyproof} with $\strategyproofbound \triangleq \skewness(\hessian_\infty) + \varepsilon$.
  
  Lemma \ref{lemma:convex_projection_strategyproof} thus applies.
  It guarantees that, for $\VOTER$ large enough, under the event $\event_\VOTER$ which occurs with probability at least $1-\delta$, 
  the geometric median is $(\skewness(\hessian_\infty) + \varepsilon)$-strategyproof for voter $0$.
  This corresponds to saying that the geometric median is asymptotically $\skewness(\hessian_\infty)$-strategyproof.
\end{proof}

%% file: App_skewed_generalizations.tex
\section{PROOFS AND DIFFERENT RESULTS FROM SECTION \ref{sec:skewed_preferences}}
\label{sec:proof_skewed_preferences}

\subsection{Sketch of Proof for Theorem \ref{th:skewed_geometric_median}}
\label{sec:proof_th_skewed_geometric_median}

\begin{figure}[t!]
    \centering
    \includegraphics[width=.6\linewidth]{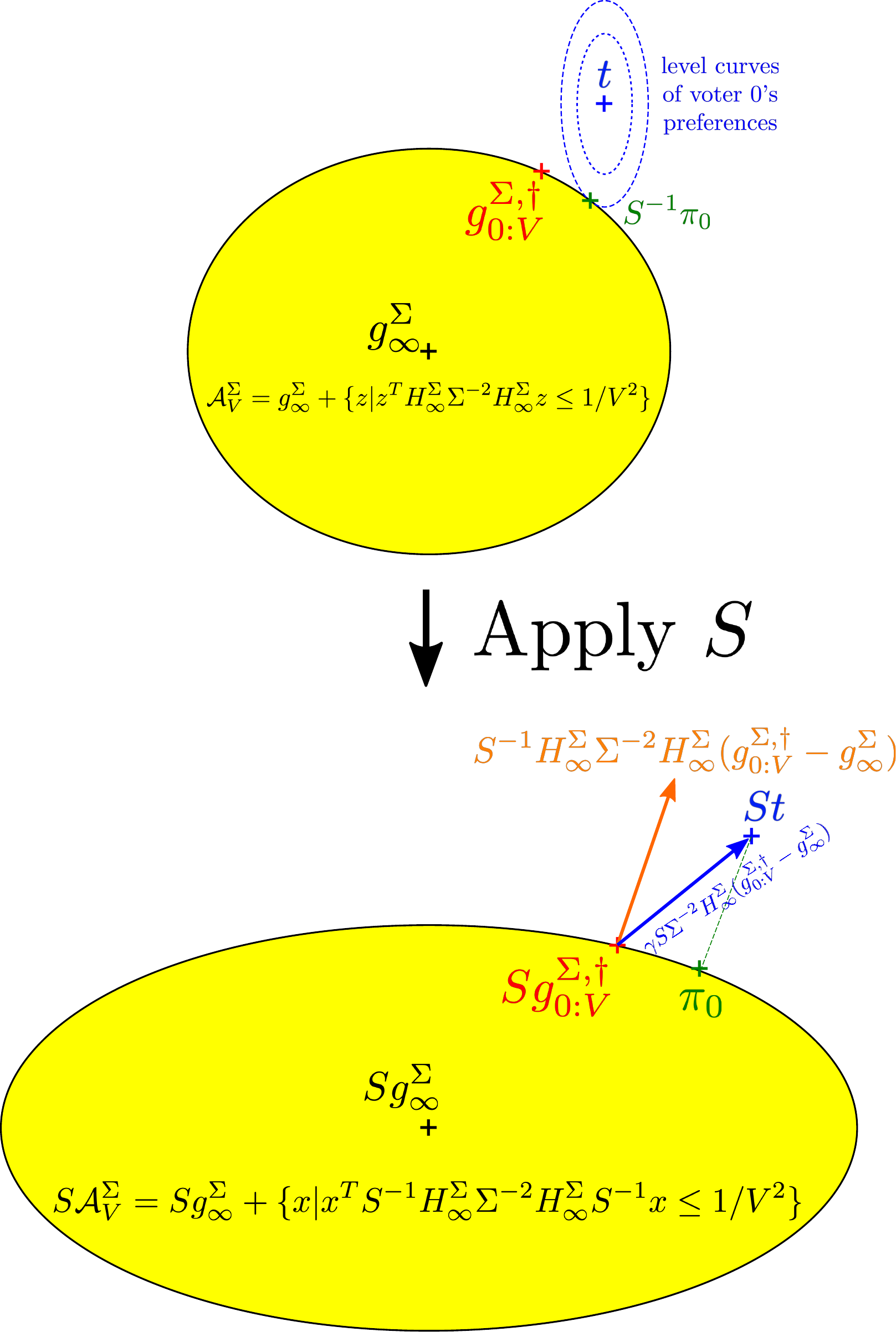}
    \caption{Proof techniques to determine the asymptotic strategyproofness of the $\adaptivematrix$-skewed geometric median for $\sdpmatrix$-skewed preferences. We skew space using $\sdpmatrix$, so that in the skewed space, voter $0$ wants to minimize the Euclidean distance between their preferred vector and the skewed geometric median. Strategyproofness then depends on the angle between the blue and orange vectors in the skewed space, as depicted in the figure.}
    \label{fig:skewed_preferences}
\end{figure}

\begin{proof}
  We  provide a sketch of proof, which is based on Figure~\ref{fig:skewed_preferences}. 
  By Taylor series and given concentration bounds, for $\VOTER$ large enough and $\parz \rightarrow \geometricmedian_\infty^\adaptivematrix$, the gradient of the skewed loss for $1+\VOTER$ voters is then approximately given by 
  \begin{align}
      (1+\VOTER) \nabla \Loss^\adaptivematrix_{0:\VOTER} (\strategicvote{0}, \paramfamily, \parz) 
      &\approx \frac{\adaptivematrix \adaptivematrix (\parz - \strategicvote{0})}{\norm{\parz - \strategicvote{0}}{\adaptivematrix}}
      + \VOTER \hessian^\adaptivematrix_\infty (\parz - \geometricmedian_\infty^\adaptivematrix)
      + o(\norm{\parz - \geometricmedian_\infty^\adaptivematrix}{2}) \\
      &= \adaptivematrix \unitvector{\adaptivematrix \parz - \adaptivematrix \strategicvote{0}} + \VOTER \hessian^\adaptivematrix_\infty (\parz - \geometricmedian_\infty^\adaptivematrix)
      + o(\norm{\parz - \geometricmedian_\infty^\adaptivematrix}{2}).
      \label{eq:taylor_skewed_gradient}
  \end{align}
  This quantity must cancel out for $\parz =
  \geometricmedian_{0:\VOTER}^{\adaptivematrix}$.
  Thus we must have $\adaptivematrix^{-1} \hessian^\adaptivematrix_\infty (\geometricmedian_{0:\VOTER}^{\adaptivematrix} - \geometricmedian_\infty^\adaptivematrix)
  \approx \frac{1}{\VOTER} \unitvector{\adaptivematrix \strategicvote{0} - \adaptivematrix \geometricmedian_{0:\VOTER}^{\adaptivematrix}}$,
  which implies $\norm{\adaptivematrix^{-1} \hessian^\adaptivematrix_\infty (\geometricmedian_{0:\VOTER}^{\adaptivematrix} - \geometricmedian_\infty^\adaptivematrix)}{2}^2 
  = 1/\VOTER^2$.
  The achievable set $\setAchieve^\adaptivematrix$ is thus approximately the ellipsoid $\set{\geometricmedian_\infty^\adaptivematrix + \parz \st \parz^T \hessian^\adaptivematrix_\infty \adaptivematrix^{-2} \hessian^\adaptivematrix_\infty \parz \leq 1/\VOTER^2 }$. In particular, for $\VOTER$ large enough, $\setAchieve^\adaptivematrix$ is convex.

  Meanwhile, denote $\geometricmedian_{0:\VOTER}^{\adaptivematrix, \dagger}$ the skewed geometric median when the strategic voter truthfully reports their preferred vector $\targetvector$. 
  By Equation (\ref{eq:taylor_skewed_gradient}), we must have $\adaptivematrix \adaptivematrix (\targetvector - \geometricmedian_{0:\VOTER}^{\adaptivematrix, \dagger}) \propto \hessian^\adaptivematrix_\infty (\geometricmedian_{0:\VOTER}^{\adaptivematrix, \dagger} - \geometricmedian_\infty^\adaptivematrix)$, 
  which implies $\targetvector - \geometricmedian_{0:\VOTER}^{\adaptivematrix, \dagger} \propto \adaptivematrix^{-2} \hessian^\adaptivematrix_\infty (\geometricmedian_{0:\VOTER}^{\adaptivematrix, \dagger} - \geometricmedian_\infty^\adaptivematrix)$.

  Now let us skew the space by matrix $\sdpmatrix$, i.e., we map each point $\parz$ in the original space to a point $\parx \triangleq \sdpmatrix \parz$ in the $\sdpmatrix$-skewed space.
  Interestingly, since $\norm{\parz-\paramsub{\voter}}{\sdpmatrix} = \norm{\parx - \sdpmatrix \paramsub{\voter}}{2}$, a voter with $\sdpmatrix$-skewed preferences in the original space now simply wants to minimize the Euclidean distance in the $\sdpmatrix$-skewed space.
  Now, note that since $\sdpmatrix$ is a linear transformation and since $\setAchieve^\adaptivematrix$ is convex, so is $\sdpmatrix \setAchieve^\adaptivematrix$.
  This allows us to re-use the orthogonal projection argument.
  Namely, denoting $\pi_0$ the orthogonal projection of $\sdpmatrix \targetvector$ onto the tangent hyperplane to $\sdpmatrix \setAchieve^\adaptivematrix$, 
  we have
  \begin{align}
    \inf_{\strategicvote{0} \in \setR^d} \norm{ \targetvector - \GeometricMedian^\adaptivematrix(\strategicvote{0}, \paramfamily_{1:\VOTER}) }{\sdpmatrix} 
    &= \inf_{\strategicvote{0} \in \setR^d} \norm{ \sdpmatrix \targetvector - \sdpmatrix \GeometricMedian^\adaptivematrix(\strategicvote{0}, \paramfamily_{1:\VOTER}) }{2}  \\
    &= \inf_{\parx \in \sdpmatrix \setAchieve^\adaptivematrix} \norm{\sdpmatrix \targetvector - \parx}{2}
    \geq \norm{ \sdpmatrix \targetvector - \pi_0 }{2}.
  \end{align}
  To compute $\pi_0$, note that, for a large number of voters and with high probability, the achievable set $\sdpmatrix \setAchieve^\adaptivematrix$ in the $\sdpmatrix$-skewed space is approximately the set of points $\sdpmatrix \geometricmedian_\infty^\adaptivematrix + \sdpmatrix \parz$ such that $\parz^T \hessian_\infty^\adaptivematrix \adaptivematrix^{-2} \hessian_\infty^\adaptivematrix \parz \leq 1 / \VOTER^2$.
  Equivalently, this corresponds to the set of points $\sdpmatrix \geometricmedian_\infty^\adaptivematrix+ \parx$ with $\parx \triangleq \sdpmatrix \parz$ (and thus $\parz = \sdpmatrix^{-1} \parx$) such that 
  $(\sdpmatrix^{-1} \parx)^T \hessian_\infty^\adaptivematrix \adaptivematrix^{-2} \hessian_\infty^\adaptivematrix (\sdpmatrix^{-1} \parx) = \parx^T (\sdpmatrix^{-1} \hessian_\infty^\adaptivematrix \adaptivematrix^{-2} \hessian_\infty^\adaptivematrix \sdpmatrix^{-1}) \parx \leq 1 / \VOTER^2$.
  This is still an ellipsoid. 
  The normal to the surface of $\sdpmatrix \setAchieve^\adaptivematrix$ at $\parxsub{0} = \sdpmatrix \geometricmedian_{0:\VOTER}^{\adaptivematrix, \dagger} - \sdpmatrix \geometricmedian_\infty^\adaptivematrix$ is then given by $\sdpmatrix^{-1} \hessian_\infty^\adaptivematrix \adaptivematrix^{-2} \hessian_\infty^\adaptivematrix \sdpmatrix^{-1} \parxsub{0} 
  = \sdpmatrix^{-1} \hessian_\infty^\adaptivematrix \adaptivematrix^{-2} \hessian_\infty^\adaptivematrix (\geometricmedian_{0:\VOTER}^{\adaptivematrix, \dagger} - \geometricmedian_\infty^\adaptivematrix)$.
  
  Meanwhile, since $\targetvector - \geometricmedian_{0:\VOTER}^{\adaptivematrix, \dagger} \propto \adaptivematrix^{-2} \hessian^\adaptivematrix_\infty (\geometricmedian_{0:\VOTER}^{\adaptivematrix, \dagger} - \geometricmedian_\infty^\adaptivematrix)$, 
  we know that there exists $\gamma > 0$ such that $\sdpmatrix \targetvector - \sdpmatrix \geometricmedian_{0:\VOTER}^{\adaptivematrix, \dagger} = \gamma \sdpmatrix \adaptivematrix^{-2} \hessian^\adaptivematrix_\infty (\geometricmedian_{0:\VOTER}^{\adaptivematrix, \dagger} - \geometricmedian_\infty^\adaptivematrix)$.
  Then

  \begin{align}
    &\frac{\norm{ \targetvector - \GeometricMedian(\targetvector, \paramfamily_{1:\VOTER}) }{\sdpmatrix} }{ \norm{ \targetvector - \GeometricMedian(\strategicvote{0}, \paramfamily_{1:\VOTER}) }{\sdpmatrix} }
    \leq \frac{ \norm{\sdpmatrix \targetvector - \sdpmatrix \geometricmedian_{0:\VOTER}^{\adaptivematrix, \dagger} }{2} }{ \norm{\sdpmatrix \targetvector - \pi_0 }{2} } \\
    &\leq \frac{ \gamma \norm{ \sdpmatrix \adaptivematrix^{-2} \hessian^\adaptivematrix_\infty (\geometricmedian_{0:\VOTER}^{\adaptivematrix, \dagger} - \geometricmedian_\infty^\adaptivematrix) }{2} }
    { (\gamma \sdpmatrix \adaptivematrix^{-2} \hessian^\adaptivematrix_\infty (\geometricmedian_{0:\VOTER}^{\adaptivematrix, \dagger} - \geometricmedian_\infty^\adaptivematrix))^T \frac{\sdpmatrix^{-1} \hessian_\infty^\adaptivematrix \adaptivematrix^{-2} \hessian_\infty^\adaptivematrix (\geometricmedian_{0:\VOTER}^{\adaptivematrix, \dagger} - \geometricmedian_\infty^\adaptivematrix)}{\norm{\sdpmatrix^{-1} \hessian_\infty^\adaptivematrix \adaptivematrix^{-2} \hessian_\infty^\adaptivematrix (\geometricmedian_{0:\VOTER}^{\adaptivematrix, \dagger} - \geometricmedian_\infty^\adaptivematrix)}{2}} } \\
    &= \frac{ \norm{ \sdpmatrix \adaptivematrix^{-2} \hessian^\adaptivematrix_\infty (\geometricmedian_{0:\VOTER}^{\adaptivematrix, \dagger} - \geometricmedian_\infty^\adaptivematrix) }{2} \norm{\sdpmatrix^{-1} \hessian_\infty^\adaptivematrix \adaptivematrix^{-2} \hessian_\infty^\adaptivematrix (\geometricmedian_{0:\VOTER}^{\adaptivematrix, \dagger} - \geometricmedian_\infty^\adaptivematrix)}{2}}
    { ( \sdpmatrix \adaptivematrix^{-2} \hessian^\adaptivematrix_\infty (\geometricmedian_{0:\VOTER}^{\adaptivematrix, \dagger} - \geometricmedian_\infty^\adaptivematrix) )^T (\sdpmatrix^{-1} \hessian_\infty^\adaptivematrix \adaptivematrix^{-2} \hessian_\infty^\adaptivematrix (\geometricmedian_{0:\VOTER}^{\adaptivematrix, \dagger} - \geometricmedian_\infty^\adaptivematrix)) }  \\
    &= \frac{ \norm{ \parysub{0} }{2} \norm{\sdpmatrix^{-1} \hessian_\infty^\adaptivematrix \sdpmatrix^{-1} \parysub{0}}{2}}
    { \parysub{0}^T \sdpmatrix^{-1} \hessian_\infty^\adaptivematrix \sdpmatrix^{-1} \parysub{0} } 
    \leq 1 + \skewness(\sdpmatrix^{-1} \hessian_\infty^\adaptivematrix \sdpmatrix^{-1}),
  \end{align}
  by defining $\parysub{0} \triangleq  \sdpmatrix \adaptivematrix^{-2} \hessian^\adaptivematrix_\infty (\geometricmedian_{0:\VOTER}^{\adaptivematrix, \dagger} - \geometricmedian_\infty^\adaptivematrix)$.
  This concludes the sketch of the proof.
  A more rigorous proof would need to follow the footsteps of our main proof (Theorem \ref{th:asymptotic_strategyproofness}).
\end{proof}

\subsection{Proof of Proposition \ref{proposition:skewed_hessian}}
\label{sec:proof_proposition_skewed_hessian}

\begin{proof}
  Let $i,j \in [d]$. Note that $\partial_j \left( (\adaptivematrix \adaptivematrix \parz) [i] \right) = \partial_j \left( \sum_k (\adaptivematrix \adaptivematrix)[i,k] \parz[k] \right) = (\adaptivematrix \adaptivematrix) [i,j]$.
  As a result, using Lemma~\ref{lemma:skewed_gradient}, we have
  \begin{align}
      \partial_{ij} \norm{\parz}{\adaptivematrix}
      &= \frac{\partial_j ((\adaptivematrix \adaptivematrix \parz)[i]) \norm{\parz}{\adaptivematrix} - (\adaptivematrix \adaptivematrix \parz)[i] \partial_j \norm{\parz}{\adaptivematrix}}{\norm{\parz}{\adaptivematrix}^2} \\
      &= \frac{(\adaptivematrix \adaptivematrix)[i,j]}{\norm{\parz}{\adaptivematrix}} - \frac{(\adaptivematrix \adaptivematrix \parz)[i] (\adaptivematrix \adaptivematrix \parz)[j]}{\norm{\parz}{\adaptivematrix}^3} \\
      &= \left(\frac{\adaptivematrix \adaptivematrix}{\norm{\parz}{\adaptivematrix}}\right)[i,j] - \left( \frac{\adaptivematrix \adaptivematrix \parz \parz^T \adaptivematrix \adaptivematrix}{\norm{\parz}{\adaptivematrix}^3} \right)[i,j] \\
      &= \left(\adaptivematrix \left( \frac{1}{\norm{\adaptivematrix \parz}{2}} 
      \left( I - \left( \frac{\adaptivematrix \parz}{\norm{\adaptivematrix \parz}{2}} \right) \left( \frac{\adaptivematrix \parz}{\norm{\adaptivematrix \parz}{2}} \right)^T \right) \right) \adaptivematrix \right)[i,j].
  \end{align}
  Combining all coordinates, replacing $\parz$ by $\parz - \paramsub{\voter}$, and averaging over all voters then yields the lemma.
\end{proof}

\subsection{The computation of $\adaptivematrix$-skewed Geometric Median}
Intuitively, the computation of the $\adaptivematrix$-skewed geometric median corresponds to skewing the space using the linear transformation $\adaptivematrix$, computing the geometric median in the skewed space, and de-skewing the computed geometric median by applying $\adaptivematrix^{-1}$. The following two lemmas formalize this intuition.
\begin{lemma}
$\Loss^\adaptivematrix_\infty (\parz, \paramdistribution) = \Loss_\infty(\adaptivematrix \parz, \adaptivematrix \paramdistribution)$ 
and $\Loss^\adaptivematrix_{0:\VOTER} (\strategicvote{0}, \paramfamily, \parz) = \Loss_{0:\VOTER} (\adaptivematrix \strategicvote{0}, \adaptivematrix \paramfamily, \adaptivematrix \parz)$.
\end{lemma}

\begin{proof}
  This is straightforward, by expanding the definition of the terms.
\end{proof}

\begin{lemma}
\label{lemma:skewedGMComp}
$\geometricmedian^\adaptivematrix_\infty (\paramdistribution) = \adaptivematrix^{-1} \geometricmedian_\infty(\adaptivematrix \paramdistribution)$ 
and $\geometricmedian^\adaptivematrix_{0:\VOTER} (\strategicvote{0}, \paramfamily) = \adaptivematrix^{-1} \geometricmedian_{0:\VOTER} (\adaptivematrix \strategicvote{0}, \adaptivematrix \paramfamily)$.
\end{lemma}

\begin{proof}
  By definition of $\geometricmedian_\infty(\adaptivematrix \paramdistribution)$, we know that it minimizes $\pary \mapsto \Loss_\infty(\pary, \adaptivematrix \paramdistribution)$.
  It is then clear that $\adaptivematrix^{-1} \geometricmedian_\infty(\adaptivematrix \paramdistribution)$ minimizes $\parz \mapsto \Loss_\infty(\adaptivematrix \parz, \adaptivematrix \paramdistribution)$.
  The case of $\geometricmedian^\adaptivematrix_{0:\VOTER}$ is similar.
\end{proof}

\subsection{No Shoe Fits Them All}
\label{sec:shoe}
In practice, we may expect different voters to assign a different importance to different dimensions.
Unfortunately, this leads to the following impossibility theorem for asymptotic strategyproofness of any skewed geometric median.

\begin{corollary}
\label{cor:no_shoe_fits_them_all}
Suppose voters $\voter, \voterbis$ have $\sdpmatrix_\voter$ and $\sdpmatrix_\voterbis$-skewed preferences, where the matrices $\sdpmatrix_\voter$ and $\sdpmatrix_\voterbis$ are not proportional.
Then no skewed geometric median is asymptotically strategyproof for both.
\end{corollary}

\begin{proof}
Asymptotic strategyproofness for $\sdpmatrix_\voter$ requires using a $\adaptivematrix$-skewed geometric median such that $\skewness( \sdpmatrix^{-1}_\voter \hessian_\infty^\adaptivematrix \sdpmatrix^{-1}_\voter) = 0$.
By Proposition~\ref{prop:skewness_lowerbound}, this means that all eigenvalues of $\sdpmatrix^{-1}_\voter \hessian_\infty^\adaptivematrix \sdpmatrix^{-1}_\voter$ must be equal, which implies that $\sdpmatrix^{-1}_\voter \hessian_\infty^\adaptivematrix \sdpmatrix^{-1}_\voter \propto I$.
But then, we must have $\hessian_\infty^\adaptivematrix \propto \sdpmatrix^2_\voter$.
As a result, we then have $\sdpmatrix^{-1}_\voterbis \hessian_\infty^\adaptivematrix \sdpmatrix^{-1}_\voterbis \propto \sdpmatrix^{-1}_\voterbis \sdpmatrix^2_\voter \sdpmatrix^{-1}_\voterbis$.
But, given our assumption about these matrices, this cannot be proportional to the identity.
Proposition~\ref{prop:skewness_lowerbound} then implies that $\skewness(\sdpmatrix^{-1}_\voterbis \hessian_\infty^\adaptivematrix \sdpmatrix^{-1}_\voterbis) > 0$, which means that the $\adaptivematrix$-skewed geometric median is not asymptotically strategyproof for voter $\voterbis$.
\end{proof}

We leave however open the problem of determining what shoe ``most fits them all''.
In other words, assuming a set $\mathcal{S}$ of skewing matrices, each of which may represent how different voters' preferences may be skewed, which $\adaptivematrix(\mathcal S)$-skewed geometric median guarantees asymptotic $\strategyproofbound$-strategyproofness for all voters, with the smallest possible value of $\strategyproofbound$? 
And what is this optimal uniform asymptotic strategyproofness guarantee $\strategyproofbound(\mathcal{S})$ that can be obtained?

%% file: alternative_unit_force.tex
\section{ALTERNATIVE UNIT FORCES}
\label{sec:alternative}
In this section we show that the fairness principle ``one voter, one vote with a unit force'' can be generalized to other vector votes when we use the right norm to  measure the norm of voters’ forces. First we consider the skewed geometric median, and then we analyze the minimizer of $\ell_p$ distances.

\subsection{Skewed Geometric Median}
Interestingly, we can also interpret the skewed geometric median as an operator that yields unit forces to the different voters, albeit the norm of the forces is not measured by the Euclidean norm.
To understand how forces are measured, let us better characterize the derivative of the skewed norm. 

\begin{lemma}
\label{lemma:skewed_gradient}
For all $\parz \in \setR^d-\set{0}$, we have $\nabla_\parz \norm{\parz}{\adaptivematrix} = \adaptivematrix \adaptivematrix \parz / \norm{\parz}{\adaptivematrix}$.
\end{lemma}

\begin{proof}
  Note that $\norm{\parz}{\adaptivematrix}^2 = \norm{\adaptivematrix \parz}{2}^2 = \sum_i (\adaptivematrix \parz)[i]^2 = \sum_i \left( \sum_j \adaptivematrix[i,j] \parx[j]\right)^2$. 
  We then have 
  \begin{align}
      \partial_i \norm{\parz}{\adaptivematrix}^2 
      &= \sum_j 2 \adaptivematrix[j,i] (\adaptivematrix \parz)[j] 
      = 2 \sum_j \sum_k \adaptivematrix[j,i] \adaptivematrix[j,k] \parz[k] \\
      &= 2 \sum_k \left( \sum_j \adaptivematrix[i,j] \adaptivematrix[j,k] \right) \parz[k]
      = 2 \sum_k \left( \adaptivematrix \adaptivematrix \right)[i,k] \parz[k]
      = 2 (\adaptivematrix \adaptivematrix \parz)[i].
  \end{align}
  From this, it follows that
  \begin{equation}
      \partial_i \norm{\parz}{\adaptivematrix} 
      = \partial_i \sqrt{\norm{\parz}{\adaptivematrix}^2}
      = \frac{\partial_i \norm{\parz}{\adaptivematrix}^2}{2 \sqrt{\norm{\parz}{\adaptivematrix}^2}}
      = \frac{(\adaptivematrix \adaptivematrix \parz) [i]}{\norm{\parz}{\adaptivematrix}}.
  \end{equation}
  Combining all coordinates yields the lemma.
\end{proof}

It is noteworthy that, using a $\adaptivematrix$-skewed loss, the gradient $\nabla_\parz \norm{\parz}{\adaptivematrix}$ is no longer colinear with $\parz$. In fact, it is not even colinear with $\adaptivematrix \parz$, which is the image of $\parz$ as we apply the linear transformation $\adaptivematrix$ to the entire space. 
Similarly, this pull is no longer of Euclidean unit force.
Nevertheless, it remains a unit force, as long as we measure its force with the appropriate norm.

\begin{lemma}
For all $\parz, \paramsub{\voter} \in \setR^d$, we have $\norm{\nabla_\parz \norm{\parz - \paramsub{\voter}}{\adaptivematrix}}{\adaptivematrix^{-1}} = 1$. 
Put differently, using the $\adaptivematrix$-skewed loss, voters have $\adaptivematrix^{-1}$-unit forces.
\end{lemma}

\begin{proof}
  Applying Lemma~\ref{lemma:skewed_gradient} yields 
  \begin{align}
    \norm{\nabla_\parz \norm{\parz - \paramsub{\voter}}{\adaptivematrix}}{\adaptivematrix^{-1}} 
    = \norm{\frac{\adaptivematrix \adaptivematrix \parz}{\norm{\parz}{\adaptivematrix}}}{\adaptivematrix^{-1}}
    = \frac{\norm{\adaptivematrix^{-1} \adaptivematrix \adaptivematrix \parz}{2}}{\norm{\parz}{\adaptivematrix}}
    = \frac{\norm{\adaptivematrix \parz}{2}}{\norm{\parz}{\adaptivematrix}}
    = 1,
  \end{align}
  which is the lemma.
\end{proof}

\subsection{$\ell_p$ Norm }
Interestingly, we prove below that considering other penalties measured by $\ell_p$ distances is equivalent to assigning $\ell_q$-unit forces to the voters. In particular, the coordinate-wise median can be interpreted as minimizing the $\ell_1$ distances or, equivalently, assigning votes of $\ell_\infty$ unit force. In particular, the coordinate-wise median, which is known to be strategyproof, indeed implements the principle {\it ``one voter, one unit-force vote''}. In other words, this principle can guarantee strategyproofness; this requires a mere change of norm.

\begin{proposition}
\label{prop:unit_force}
  Assume $\frac{1}{p} + \frac{1}{q} = 1$, with $p,q \in [1,\infty]$.
  Then considering an $\ell_p$ penalty is equivalent to considering that each voter has an $\ell_q$-unit force vote.
  More precisely, any subgradient of the $\ell_p$ penalty has at most a unit norm in $\ell_q$.
\end{proposition}

\begin{proof}
  Assume $\parx \neq 0$ and $1<p,q<\infty$. Then we have
  \begin{align}
    \absv{\partial_j \norm{\parx}{p}}^q
    &= \absv{\partial_j \left( \sum_{j \in [d]} \absv{\parx[j]}^p \right)^{1/p}}^q
    = \absv{\frac{1}{p} \left( \sum_{j \in [d]} \absv{\parx[j]}^p \right)^{(1/p)-1} \left( p \absv{\parx[j]}^{p-1} sign(\parx[j]) \right)}^q \\
    &= \norm{\parx}{p}^{q(1-p)} \absv{\parx[j]}^{q(p-1)}
    = \frac{ \absv{\parx[j]}^p }{ \norm{\parx}{p}^p },
  \end{align}
  using the equality $q(p-1) = p$ derived from $\frac{1}{p} + \frac{1}{q} = 1$. Adding up all such quantities for $j \in [d]$ yields
  \begin{align}
    \norm{\nabla \norm{\parx}{p}}{q}
    &= \left( \sum_{j \in [d]} \frac{ \absv{\parx[j]}^p }{ \norm{\parx}{p}^p } \right)^{1/q}
    = \left( \frac{ 1 }{ \norm{\parx}{p}^p } \sum_{j \in [d]} \absv{\parx[j]}^p \right)^{1/q} \\
    &= \left( \frac{ 1 }{ \norm{\parx}{p}^p } \norm{\parx}{p}^p \right)^{1/q}
    = 1.
  \end{align}
  Thus the gradient of the $\ell_p$ norm is unitary in $\ell_q$ norm, when $\parx \neq 0$.
  Note then that a subgradient $g$ at $0$ must satisfy $g^T \parx \leq \norm{\parx}{p}$ for all $\parx \in \setR^d$.
  This corresponds to saying that the operator norm of $\parx \mapsto g^T \parx$ must be at most one with respect to the norm $\ell_p$.
  Yet it is well-known that this operator norm is the $\ell_q$ norm of $g$.

  In the case $p=1$, then each coordinate is treated independently.
  On each coordinate, the derivative is then between $-1$ and $1$ (and can equal $[-1,1]$ if $\parx[j] = 0$).
  This means that the gradients are of norm at most 1.

  The last case left to analyze is when $p=\infty$.
  Denote $J^{max}(\parx) \triangleq \set{j \in [d] \st \parx[j] = \norm{\parx}{\infty}}$.
  When $\card{J^{max}(\parx)} = 1$, denoting $j$ the only element of $J^{max}(x)$ and $\unitvector{j}$ the $j$-th vector of the canonical basis, then the gradient of the $\ell_\infty$ is clearly $\unitvector{j}$, which is unitary in $\ell_1$ norm.
  Moreover, note that if $k \notin J^{max}(\parx)$, then we clearly have $\partial_k \norm{\parx}{\infty} = 0$.

  Now, denote $g \in \nabla \norm{\parx}{\infty}$, let $\pary \in \setR^d$, and assume for simplicity that $\parx \geq 0$.
  We know that
  \begin{equation}
    \norm{\parx + \varepsilon \pary}{\infty}
    \geq \norm{\parx}{\infty} + \varepsilon g^T \pary.
  \end{equation}
  For $\varepsilon > 0$ small enough, we then have
  \begin{equation}
    \norm{\parx}{\infty} + \varepsilon \max_{j \in J^{max}(\parx)} \pary[j]
    \geq \norm{\parx}{\infty} + \varepsilon \sum_{j \in J^{max}(\parx)} g[j] \pary[j],
  \end{equation}
  from which it follows that
  \begin{equation}
    \sum_{j \in J^{max}(\parx)} g[j] \pary[j] \leq \max_{j \in J^{max}(\parx)} \pary[j].
  \end{equation}
  Considering $\pary[j] = -1$ for $j \in J^{max}(\parx)$ and $\pary[k] = 0$ for all $k \neq j$ then implies $-g[j] \leq 0$, which yields $g[j] \geq 0$.
  Generalizing it for all $j$'s implies that $g \geq 0$.
  Now, considering $\pary[j] = 1$ for all $j \in J^{max}(\parx)$ then yields $\sum_{j \in J^{max}(\parx)} g[j] = \norm{g}{1} \leq 1$, which concludes the proof for $\parx \geq 0$.
  The general case can be derived by considering axial symmetries.
\end{proof}

%% file: _main.bbl
\begin{thebibliography}{}

\bibitem[Abadi et~al., 2015]{Tensorflow2015}
Abadi, M., Agarwal, A., Barham, P., Brevdo, E., Chen, Z., Citro, C., Corrado,
  G., Davis, A., Dean, J., Devin, M., Ghemawat, S., Goodfellow, I., Harp, A.,
  Irving, G., Isard, M., Jia, Y., Jozefowicz, R., Kaiser, L., Kudlur, M.,
  Levenberg, J., Mané, D., Monga, R., Moore, S., Murray, D., Olah, C.,
  Schuster, M., Shlens, J., Steiner, B., Sutskever, I., Talwar, K., Tucker, P.,
  Vanhoucke, V., Vasudevan, V., Viégas, F., Vinyals, O., Warden, P.,
  Wattenberg, M., Wicke, M., Yu, Y., and Zheng, X. (2015).
\newblock {TensorFlow}: Large-scale machine learning on heterogeneous
  distributed systems.

\bibitem[Acharya et~al., 2022]{acharya2021}
Acharya, A., Hashemi, A., Jain, P., Sanghavi, S., Dhillon, I.~S., and Topcu, U.
  (2022).
\newblock Robust training in high dimensions via block coordinate geometric
  median descent.
\newblock In Camps-Valls, G., Ruiz, F. J.~R., and Valera, I., editors, {\em
  Proceedings of The 25th International Conference on Artificial Intelligence
  and Statistics}, volume 151 of {\em Proceedings of Machine Learning
  Research}, pages 11145--11168. PMLR.

\bibitem[Alistarh et~al., 2018]{alistarh18}
Alistarh, D., Allen-Zhu, Z., and Li, J. (2018).
\newblock Byzantine stochastic gradient descent.
\newblock In Bengio, S., Wallach, H., Larochelle, H., Grauman, K.,
  Cesa-Bianchi, N., and Garnett, R., editors, {\em Advances in Neural
  Information Processing Systems}, volume~31. Curran Associates, Inc.

\bibitem[Allouah et~al., 2022]{AllouahGHV22}
Allouah, Y., Guerraoui, R., Hoang, L., and Villemaud, O. (2022).
\newblock Robust sparse voting.
\newblock {\em CoRR}, abs/2202.08656.

\bibitem[Alon et~al., 2010]{Alon2009}
Alon, N., Feldman, M., Procaccia, A., and Tennenholtz, M. (2010).
\newblock Strategyproof approximation mechanisms for location on networks.
\newblock {\em Center for Rationality and Interactive Decision Theory, Hebrew
  University, Jerusalem, Discussion Paper Series}.

\bibitem[Barber{\`a} et~al., 1997]{barbera1997voting}
Barber{\`a}, S., Mass{\'o}, J., and Neme, A. (1997).
\newblock Voting under constraints.
\newblock {\em journal of economic theory}, 76(2):298--321.

\bibitem[Barber{\`a} et~al., 1998]{barbera1998strategy}
Barber{\`a}, S., Mass{\'o}, J., and Serizawa, S. (1998).
\newblock Strategy-proof voting on compact ranges.
\newblock {\em games and economic behavior}, 25(2):272--291.

\bibitem[Bhat and Klein, 2020]{bhat20}
Bhat, P. and Klein, O. (2020).
\newblock Covert hate speech: white nationalists and dog whistle communication
  on {T}witter.
\newblock In {\em Twitter, the Public Sphere, and the Chaos of Online
  Deliberation}, pages 151--172. Springer.

\bibitem[Blanchard et~al., 2017]{BlanchardMGS17}
Blanchard, P., Mhamdi, E. M.~E., Guerraoui, R., and Stainer, J. (2017).
\newblock Machine learning with adversaries: Byzantine tolerant gradient
  descent.
\newblock In Guyon, I., von Luxburg, U., Bengio, S., Wallach, H.~M., Fergus,
  R., Vishwanathan, S. V.~N., and Garnett, R., editors, {\em Advances in Neural
  Information Processing Systems 30: Annual Conference on Neural Information
  Processing Systems 2017, 4-9 December 2017, Long Beach, CA, {USA}}, pages
  119--129.

\bibitem[Border and Jordan, 1983]{border1983straightforward}
Border, K.~C. and Jordan, J.~S. (1983).
\newblock Straightforward elections, unanimity and phantom voters.
\newblock {\em The Review of Economic Studies}, 50(1):153--170.

\bibitem[Brady and Chambers, 1995]{Brady17}
Brady, R.~L. and Chambers, C.~P. (1995).
\newblock A spatial analogue of may's theorem.
\newblock {\em Social Choice and Welfare}, 71.

\bibitem[Brandt et~al., 2016]{brandt16}
Brandt, F., Conitzer, V., Endriss, U., Lang, J., and Procaccia, A.~D. (2016).
\newblock {\em Handbook of computational social choice}.
\newblock Cambridge University Press.

\bibitem[Brimberg, 2017]{brimberg17}
Brimberg, J. (2017).
\newblock The fermat—weber location problem revisited.
\newblock {\em Mathematical Programming}, 49.

\bibitem[Chen et~al., 2017]{chen17}
Chen, Y., Su, L., and Xu, J. (2017).
\newblock Distributed statistical machine learning in adversarial settings:
  Byzantine gradient descent.
\newblock {\em Proc. ACM Meas. Anal. Comput. Syst.}, 1(2).

\bibitem[Chung and Ely, 2007]{chung2007foundations}
Chung, K.-S. and Ely, J.~C. (2007).
\newblock Foundations of dominant-strategy mechanisms.
\newblock {\em The Review of Economic Studies}, 74(2):447--476.

\bibitem[Cohen et~al., 2016]{CohenLMPS16}
Cohen, M.~B., Lee, Y.~T., Miller, G.~L., Pachocki, J., and Sidford, A. (2016).
\newblock Geometric median in nearly linear time.
\newblock In Wichs, D. and Mansour, Y., editors, {\em Proceedings of the 48th
  Annual {ACM} {SIGACT} Symposium on Theory of Computing, {STOC} 2016,
  Cambridge, MA, USA, June 18-21, 2016}, pages 9--21. {ACM}.

\bibitem[Dinh et~al., 2020]{DinhTN20}
Dinh, C.~T., Tran, N.~H., and Nguyen, T.~D. (2020).
\newblock Personalized federated learning with moreau envelopes.
\newblock In Larochelle, H., Ranzato, M., Hadsell, R., Balcan, M., and Lin, H.,
  editors, {\em Advances in Neural Information Processing Systems 33: Annual
  Conference on Neural Information Processing Systems 2020, NeurIPS 2020,
  December 6-12, 2020, virtual}.

\bibitem[El-Mhamdi et~al., 2018]{pmlr-v80-mhamdi18a}
El-Mhamdi, E.-M., Guerraoui, R., and Rouault, S. (2018).
\newblock The hidden vulnerability of distributed learning in {B}yzantium.
\newblock In Dy, J. and Krause, A., editors, {\em Proceedings of the 35th
  International Conference on Machine Learning}, volume~80 of {\em Proceedings
  of Machine Learning Research}, pages 3521--3530. PMLR.

\bibitem[Escoffier et~al., 2011]{esco11}
Escoffier, B., Gourv{\`e}s, L., Kim~Thang, N., Pascual, F., and Spanjaard, O.
  (2011).
\newblock Strategy-proof mechanisms for facility location games with many
  facilities.
\newblock In Brafman, R.~I., Roberts, F.~S., and Tsouki{\`a}s, A., editors,
  {\em Algorithmic Decision Theory}, pages 67--81, Berlin, Heidelberg. Springer
  Berlin Heidelberg.

\bibitem[Farhadkhani et~al., 2022a]{pmlr-v162-farhadkhani22a}
Farhadkhani, S., Guerraoui, R., Gupta, N., Pinot, R., and Stephan, J. (2022a).
\newblock {B}yzantine machine learning made easy by resilient averaging of
  momentums.
\newblock In Chaudhuri, K., Jegelka, S., Song, L., Szepesvari, C., Niu, G., and
  Sabato, S., editors, {\em Proceedings of the 39th International Conference on
  Machine Learning}, volume 162 of {\em Proceedings of Machine Learning
  Research}, pages 6246--6283. PMLR.

\bibitem[Farhadkhani et~al., 2021]{FarhadkhaniGH21}
Farhadkhani, S., Guerraoui, R., and Hoang, L.-N. (2021).
\newblock Strategyproof learning: Building trustworthy user-generated datasets.
\newblock {\em ArXiV}.

\bibitem[Farhadkhani et~al., 2022b]{equivalence}
Farhadkhani, S., Guerraoui, R., Hoang, L.-N., and Villemaud, O. (2022b).
\newblock An equivalence between data poisoning and byzantine gradient attacks.
\newblock In {\em Proceedings of the 39th International Conference on Machine
  Learning}, Proceedings of Machine Learning Research.

\bibitem[Feigenbaum and Sethuraman, 2015]{AAAIW1510182}
Feigenbaum, I. and Sethuraman, J. (2015).
\newblock Strategyproof mechanisms for one-dimensional hybrid and obnoxious
  facility location models.
\newblock {\em AAAI Workshops}.

\bibitem[Fotakis and Tzamos, 2013]{fotakis13}
Fotakis, D. and Tzamos, C. (2013).
\newblock On the power of deterministic mechanisms for facility location games.
\newblock In Fomin, F.~V., Freivalds, R., Kwiatkowska, M., and Peleg, D.,
  editors, {\em Automata, Languages, and Programming}, pages 449--460, Berlin,
  Heidelberg. Springer Berlin Heidelberg.

\bibitem[Gibbard, 1973]{gibbard1973manipulation}
Gibbard, A. (1973).
\newblock Manipulation of voting schemes: a general result.
\newblock {\em Econometrica: journal of the Econometric Society}, pages
  587--601.

\bibitem[Goel and Hann{-}Caruthers, 2020]{goel2020}
Goel, S. and Hann{-}Caruthers, W. (2020).
\newblock Coordinate-wise median: Not bad, not bad, pretty good.
\newblock {\em CoRR}, abs/2007.00903.

\bibitem[Gu and Yang, 2021]{Gu21}
Gu, Z. and Yang, Y. (2021).
\newblock Detecting malicious model updates from federated learning on
  conditional variational autoencoder.
\newblock In {\em 2021 IEEE International Parallel and Distributed Processing
  Symposium (IPDPS)}, pages 671--680.

\bibitem[Han et~al., 2015]{Han15}
Han, S., Topcu, U., and Pappas, G.~J. (2015).
\newblock An approximately truthful mechanism for electric vehicle charging via
  joint differential privacy.
\newblock In {\em 2015 American Control Conference (ACC)}, pages 2469--2475.

\bibitem[Hansen et~al., 1985]{hansen85}
Hansen, P., Peeters, D., Richard, D., and Thisse, J.-F. (1985).
\newblock The minisum and minimax location problems revisited.
\newblock {\em Operations Research}, 33(6):1251--1265.

\bibitem[Hoang, 2017]{Hoang17}
Hoang, L.~N. (2017).
\newblock Strategy-proofness of the randomized condorcet voting system.
\newblock {\em Soc. Choice Welf.}, 48(3):679--701.

\bibitem[Hoang, 2020]{hoang20a}
Hoang, L.~N. (2020).
\newblock Science communication desperately needs more aligned recommendation
  algorithms.
\newblock {\em Frontiers in Communication}, 5:115.

\bibitem[Kairouz et~al., 2021]{Kairouz21}
Kairouz, P., McMahan, H.~B., Avent, B., Bellet, A., Bennis, M., Bhagoji, A.~N.,
  Bonawitz, K., Charles, Z., Cormode, G., Cummings, R., D’Oliveira, R. G.~L.,
  Eichner, H., Rouayheb, S.~E., Evans, D., Gardner, J., Garrett, Z., Gascón,
  A., Ghazi, B., Gibbons, P.~B., Gruteser, M., Harchaoui, Z., He, C., He, L.,
  Huo, Z., Hutchinson, B., Hsu, J., Jaggi, M., Javidi, T., Joshi, G., Khodak,
  M., Konecný, J., Korolova, A., Koushanfar, F., Koyejo, S., Lepoint, T., Liu,
  Y., Mittal, P., Mohri, M., Nock, R., Özgür, A., Pagh, R., Qi, H., Ramage,
  D., Raskar, R., Raykova, M., Song, D., Song, W., Stich, S.~U., Sun, Z.,
  Suresh, A.~T., Tramèr, F., Vepakomma, P., Wang, J., Xiong, L., Xu, Z., Yang,
  Q., Yu, F.~X., Yu, H., and Zhao, S. (2021).
\newblock Advances and open problems in federated learning.
\newblock {\em Foundations and Trends® in Machine Learning}, 14(1–2):1--210.

\bibitem[Karimireddy et~al., 2022]{karimireddy2022byzantinerobust}
Karimireddy, S.~P., He, L., and Jaggi, M. (2022).
\newblock Byzantine-robust learning on heterogeneous datasets via bucketing.
\newblock In {\em International Conference on Learning Representations}.

\bibitem[Kim and Roush, 1984]{KIM198429}
Kim, K. and Roush, F. (1984).
\newblock Nonmanipulability in two dimensions.
\newblock {\em Mathematical Social Sciences}, 8(1):29--43.

\bibitem[Kyropoulou et~al., 2019]{kyr10}
Kyropoulou, M., Ventre, C., and Zhang, X. (2019).
\newblock Mechanism design for constrained heterogeneous facility location.
\newblock In {\em Algorithmic Game Theory: 12th International Symposium, SAGT
  2019, Athens, Greece, September 30 – October 3, 2019, Proceedings}, page
  63–76, Berlin, Heidelberg. Springer-Verlag.

\bibitem[Lopuhaa and Rousseeuw, 1989]{lopuhaa1989}
Lopuhaa, H.~P. and Rousseeuw, P.~J. (1989).
\newblock On the relation between s-estimators and m-estimators of multivariate
  location and covariance.
\newblock {\em The Annals of Statistics}, pages 1662--1683.

\bibitem[Lu et~al., 2009]{pinyan09}
Lu, P., Wang, Y., and Zhou, Y. (2009).
\newblock Tighter bounds for facility games.
\newblock In Leonardi, S., editor, {\em Internet and Network Economics}, pages
  137--148, Berlin, Heidelberg. Springer Berlin Heidelberg.

\bibitem[Lubin and Parkes, 2012]{lubin12}
Lubin, B. and Parkes, D.~C. (2012).
\newblock Approximate strategyproofness.
\newblock {\em Current Science}, 103(9):1021--1032.

\bibitem[McMahan et~al., 2017]{McMahan17}
McMahan, B., Moore, E., Ramage, D., Hampson, S., and Arcas, B. A.~y. (2017).
\newblock {Communication-Efficient Learning of Deep Networks from Decentralized
  Data}.
\newblock In Singh, A. and Zhu, J., editors, {\em Proceedings of the 20th
  International Conference on Artificial Intelligence and Statistics},
  volume~54 of {\em Proceedings of Machine Learning Research}, pages
  1273--1282. PMLR.

\bibitem[Mena, 2020]{mena20}
Mena, P. (2020).
\newblock Cleaning up social media: The effect of warning labels on likelihood
  of sharing false news on {F}acebook.
\newblock {\em Policy \& internet}, 12(2):165--183.

\bibitem[Michelman, 2020]{michelman20}
Michelman, P. (2020).
\newblock Can we amplify the good and contain the bad of social media?
\newblock {\em MIT Sloan Management Review}, 62(1):1--5.

\bibitem[Minsker, 2015]{STANISLAV15}
Minsker, S. (2015).
\newblock Geometric median and robust estimation in banach spaces.
\newblock {\em Bernoulli}, 21(4):2308--2335.

\bibitem[Moulin, 1980]{Moulin80}
Moulin, H. (1980).
\newblock On strategy-proofness and single peakedness.
\newblock {\em Public Choice}, 35(4):437--455.

\bibitem[Noothigattu et~al., 2018]{NoothigattuGADR18}
Noothigattu, R., Gaikwad, S. N.~S., Awad, E., Dsouza, S., Rahwan, I.,
  Ravikumar, P., and Procaccia, A.~D. (2018).
\newblock A voting-based system for ethical decision making.
\newblock In McIlraith, S.~A. and Weinberger, K.~Q., editors, {\em Proceedings
  of the Thirty-Second {AAAI} Conference on Artificial Intelligence, (AAAI-18),
  the 30th innovative Applications of Artificial Intelligence (IAAI-18), and
  the 8th {AAAI} Symposium on Educational Advances in Artificial Intelligence
  (EAAI-18), New Orleans, Louisiana, USA, February 2-7, 2018}, pages
  1587--1594. {AAAI} Press.

\bibitem[Pillutla et~al., 2022]{pillutla2019robust}
Pillutla, K., Kakade, S.~M., and Harchaoui, Z. (2022).
\newblock Robust aggregation for federated learning.
\newblock {\em IEEE Transactions on Signal Processing}, 70:1142--1154.

\bibitem[Polyak and Juditsky, 1992]{polyak92}
Polyak, B.~T. and Juditsky, A.~B. (1992).
\newblock Acceleration of stochastic approximation by averaging.
\newblock {\em SIAM Journal on Control and Optimization}, 30(4):838--855.

\bibitem[Procaccia and Tennenholtz, 2013]{Procaccia13}
Procaccia, A.~D. and Tennenholtz, M. (2013).
\newblock Approximate mechanism design without money.
\newblock {\em ACM Trans. Econ. Comput.}

\bibitem[Rajput et~al., 2019]{rajput19}
Rajput, S., Wang, H., Charles, Z., and Papailiopoulos, D. (2019).
\newblock Detox: A redundancy-based framework for faster and more robust
  gradient aggregation.
\newblock In Wallach, H., Larochelle, H., Beygelzimer, A., d\textquotesingle
  Alch\'{e}-Buc, F., Fox, E., and Garnett, R., editors, {\em Advances in Neural
  Information Processing Systems}, volume~32. Curran Associates, Inc.

\bibitem[Ribeiro et~al., 2020]{ribeiro20}
Ribeiro, M.~H., Jhaver, S., Zannettou, S., Blackburn, J., Cristofaro, E.~D.,
  Stringhini, G., and West, R. (2020).
\newblock Does platform migration compromise content moderation? evidence from
  r/the{\_}donald and r/incels.
\newblock {\em CoRR}, abs/2010.10397.

\bibitem[Satterthwaite, 1975]{satterthwaite1975strategy}
Satterthwaite, M.~A. (1975).
\newblock Strategy-proofness and arrow's conditions: Existence and
  correspondence theorems for voting procedures and social welfare functions.
\newblock {\em Journal of economic theory}, 10(2):187--217.

\bibitem[Smith and Vamanamurthy, 1989]{David89}
Smith, D.~J. and Vamanamurthy, M.~K. (1989).
\newblock How small is a unit ball?
\newblock {\em Mathematics Magazine}, 62(2):101--107.

\bibitem[So et~al., 2021]{So21}
So, J., Güler, B., and Avestimehr, A.~S. (2021).
\newblock Byzantine-resilient secure federated learning.
\newblock {\em IEEE Journal on Selected Areas in Communications},
  39(7):2168--2181.

\bibitem[Sui and Boutilier, 2015]{Sui15}
Sui, X. and Boutilier, C. (2015).
\newblock Approximately strategy-proof mechanisms for (constrained) facility
  location.
\newblock In {\em Proceedings of the 2015 International Conference on
  Autonomous Agents and Multiagent Systems}, AAMAS '15, page 605–613,
  Richland, SC. International Foundation for Autonomous Agents and Multiagent
  Systems.

\bibitem[Tang et~al., 2020]{pingzhong20}
Tang, P., Yu, D., and Zhao, S. (2020).
\newblock Characterization of group-strategyproof mechanisms for facility
  location in strictly convex space.
\newblock In {\em Proceedings of the 21st ACM Conference on Economics and
  Computation}, EC '20, page 133–157, New York, NY, USA. Association for
  Computing Machinery.

\bibitem[Walsh, 2020]{walsh20}
Walsh, T. (2020).
\newblock Strategy proof mechanisms for facility location in euclidean and
  manhattan space.
\newblock {\em CoRR}, abs/2009.07983.

\bibitem[Wang et~al., 2015]{wang15}
Wang, Q., Ye, B., Xu, T., Lu, S., and Guo, S. (2015).
\newblock Approximately truthful mechanisms for radio spectrum allocation.
\newblock {\em IEEE Transactions on Vehicular Technology}, 64(6):2615--2626.

\bibitem[Whitten-Woodring et~al., 2020]{whitten2020poison}
Whitten-Woodring, J., Kleinberg, M.~S., Thawnghmung, A., and Thitsar, M.~T.
  (2020).
\newblock Poison if you don’t know how to use it: Facebook, democracy, and
  human rights in myanmar.
\newblock {\em The International Journal of Press/Politics}, 25(3):407--425.

\bibitem[Wu et~al., 2020]{Wu20}
Wu, Z., Ling, Q., Chen, T., and Giannakis, G.~B. (2020).
\newblock Federated variance-reduced stochastic gradient descent with
  robustness to {B}yzantine attacks.
\newblock {\em IEEE Transactions on Signal Processing}, 68:4583--4596.

\bibitem[Yue, 2019]{yue2019weaponization}
Yue, N. (2019).
\newblock The" weaponization" of facebook in myanmar: A case for corporate
  criminal liability.
\newblock {\em Hastings LJ}, 71:813.

\bibitem[Zhang and Li, 2014]{zhang2014}
Zhang, Q. and Li, M. (2014).
\newblock Strategyproof mechanism design for facility location games with
  weighted agents on a line.
\newblock {\em Journal of Combinatorial Optimization}, 28(4):756--773.

\bibitem[Zinkevich et~al., 2010]{Zinkevich10}
Zinkevich, M., Weimer, M., Li, L., and Smola, A. (2010).
\newblock Parallelized stochastic gradient descent.
\newblock In Lafferty, J., Williams, C., Shawe-Taylor, J., Zemel, R., and
  Culotta, A., editors, {\em Advances in Neural Information Processing
  Systems}, volume~23. Curran Associates, Inc.

\end{thebibliography}
